\documentclass[11pt]{article}

\usepackage[top=1.15in,bottom=1.15in,left=1.15in,right=1.15in]{geometry}
\usepackage{amsmath,amssymb,amsthm}
\usepackage{graphicx}
\usepackage{booktabs}
\usepackage{enumerate}
\usepackage{setspace}
\usepackage{xcolor}
\usepackage{natbib}
\usepackage{authblk}
\usepackage{hyperref}
\hypersetup{
  colorlinks=true,
  linkcolor=blue,
  citecolor=blue,
  urlcolor=blue
}

\theoremstyle{plain}
\newtheorem{theorem}{Theorem}

\theoremstyle{remark}
\newtheorem{Remark}{Remark}
\theoremstyle{definition}

\newtheorem{Condition}{Condition}
\newtheorem{Example}{Example}

\def\proof{\par\noindent\textit{Proof.}\ }

\def\epf{\quad $\blacksquare$}

\def\hat{\widehat}

\def\mathX{\mathcal{X}}
\def\mathY{\mathcal{Y}}
\def\mathZ{\mathcal{Z}}
\def\mathH{\mathcal{H}}

\def\mathA{\mathcal{A}}
\def\mathB{\mathcal{B}}
\def\mathV{\mathcal{V}}
\def\FDP{\text{FDP}}
\def\FDR{\text{FDR}}
\def\FF{\mathbb{F}}

\title{Revisiting dependence in multiple testing: empirical distribution approaches for FDP control}

\author[1]{Fangyong Zheng}
\author[2]{Pengfei Li}
\author[3]{Yuan Jiang}
\author[1]{Tao Yu}

\affil[1]{Department of Statistics and Data Science, National University of Singapore, 117546, Singapore}
\affil[2]{Department of Statistics and Actuarial Sciences, University of Waterloo, N2L 3G1, ON, Canada}
\affil[3]{Department of Statistics, Oregon State University, Corvallis, 97331, OR, United States}

\date{\small
ORCID: Pengfei Li 0000-0003-2165-9157; Yuan Jiang 0000-0001-6409-9159; Tao Yu 0000-0002-7921-1403\\
\texttt{yu.tao@nus.edu.sg}}

\begin{document}
\maketitle

\begin{abstract}
Large-scale multiple hypothesis testing is central to the analysis of high-throughput data, where controlling false discoveries is critical. Classical procedures typically rely on theoretical null distributions and often adjust for dependence among test statistics, but these approaches may be misleading when the empirical distribution of null statistics deviates from theoretical assumptions. Motivated by this observation, we investigate the role of the empirical cumulative distribution function (c.d.f.) of null test statistics in controlling the false discovery proportion (FDP). We first show that, under an oracle scenario where the empirical c.d.f. of the test statistics for all null hypotheses is known, FDP control can be achieved optimally regardless of the dependence structure, highlighting that explicit modeling of dependence may be unnecessary. Building on this insight, we propose an empirical c.d.f.-based FDP control (eFDP) method, implemented via a multivariate mixture model framework and a nonparametric estimation procedure for the empirical c.d.f.s, establish its asymptotic convergence, and construct an FDP control procedure that achieves asymptotic FDP control. Extensive simulations demonstrate that eFDP attains more accurate FDP control and higher power than existing approaches, particularly under strong dependence, and analysis of a high-dimensional breast cancer gene expression dataset confirms its practical utility.
\end{abstract}



\section{Introduction}
The problem of large-scale multiple hypothesis testing—where hundreds or thousands of tests are conducted simultaneously—commonly arises in the statistical analysis of high-throughput data. Early work in this area focused on controlling the familywise error rate (FWER) \citep{bonferroni1936teoria, vsidak1967rectangular, holm1979simple, simes1986improved, holland1987improved, hochberg1988sharper, rom1990sequentially} and the generalized familywise error rate (gFWER) \citep{dudoit2004multiple, pollard2004choice, lehmann2012generalizations}, which limit the number of allowable false discoveries to one or a small fixed number. However, as the number of simultaneous tests increases, controlling FWER or gFWER becomes overly conservative and substantially reduces statistical power.

\par
To better address the multiplicity problem, \citet{benjamini1995controlling} introduced the false discovery rate (FDR) for large-scale multiple hypothesis testing. The FDR is defined as the expected value of the FDP, where the FDP is the proportion of false discoveries among all discoveries, given by 
\[\FDR = E(\FDP) = E\left( \frac{V}{R} \right),
\]
with the convention that $\FDP = 0$ when $R = 0$, and where $R$ and $V$ denote the numbers of discoveries and false discoveries, respectively. Since its introduction, a wide range of procedures have been developed to control the FDR, including the Benjamini--Hochberg (BH) procedure \citep{benjamini1995controlling}, the Benjamini--Yekutieli (BY) procedure \citep{benjamini2001control}, the adaptive $p$-value procedure \citep{benjamini2000adaptive}, the $q$-value procedure \citep{storey2002direct}, and the local false discovery rate (Lfdr) approach \citep{efron2001empirical}, among others. These procedures are known to control the FDR under independence of the test statistics (or equivalently, the $p$-values), and many also remain effective under certain forms of dependence. For example, \citet{benjamini2001control} showed that the BH procedure controls the FDR under positive regression dependency on a subset (PRDS) of the true null hypotheses, while \citet{StoreyTaylorSiegmund2004} established FDR control of the $q$-value procedure under weak dependence. The BY procedure, in the same spirit as the BH procedure, is able to control the FDR under general dependence structures; however, it is often overly conservative in practice, resulting in a substantial loss of testing power.

In practice, the dependence structure among test statistics can be arbitrary and often more complex, thus causing additional complications for controlling or estimating FDR. \citet{efron2007correlation} showed that the correlation structure plays a critical role in multiple hypothesis testing procedures. \citet{schwartzman2011effect} found the explicit impact of correlations on the mean and variance of the FDR estimator and argued that the correlations may increase the bias and variance of the estimator substantially. \citet{heesen2015inequalities} studied upper and lower bounds for FDR under various dependence structures of $p$-values. \citet{mei2024asymptotic} further showed how weak dependence can influence the variance of FDP even though it does not affect FDR asymptotically.

A series of large-scale multiple hypothesis testing methods have been developed to address the dependence structures among test statistics. For example, \citet{sun2009large} used a Markov chain to model local correlations among tests and introduced a multiple testing procedure based on the hidden Markov model. \citet{friguet2009factor}, \citet{fan2012estimating}, and \citet{FanHan2017FDP} considered a multivariate Gaussian distribution for the test statistics and proposed to estimate FDR by estimating and subtracting off the principal component structure in the correlations among test statistics. \citet{wei2008incorporating} and \citet{tansey2018false} incorporated spatial structure into the probability components of the ``two groups model'' (a.k.a., ``mixture model'') for large multiple-testing problems. \citet{ghosal2011predicting} modeled the distribution of the probit transform of the $p$-values by mixtures of multivariate skew-normal distributions. \citet{fithian2022conditional} introduced the dependence-adjusted Benjamini--Hochberg (dBH) procedure to achieve finite-sample FDR control for dependent tests by adaptively calibrating a separate BH-adjusted $p$-value cutoff for each hypothesis. Other works such as \citet{xie2011optimal}, \citet{sun2015false}, and \citet{heller2021optimal} obtained precise results by finding optimal methods for FDR control with dependent test statistics. More recently, \citet{DuGuoSunZou2023SDA} proposed a symmetrized data aggregation approach that achieves FDR control under general dependence without explicitly modeling the dependence structure. Their method constructs test statistics with a symmetry property under the null, enabling valid FDR control and improved power in complex dependence settings.

Most of the above works derived multiple testing procedures by directly modeling or adjusting for the dependence structure among test statistics. Alternatively, multiple testing procedures can be derived using the distribution of the test statistics under the null hypotheses \citep{efron2004large, efron2007correlation, efron2007size}. In this line of research, \citet{efron2004large} advocated the use of the empirical null distribution of the test statistics instead of its theoretical counterpart. Likewise, the test-statistic dependence affects the empirical distribution of the test statistics under the null hypotheses, for example, shifts in the center or changes in the dispersion. To address this issue, \citet{efron2007size} proposed estimating an empirical null distribution using either ``central matching'' or ``MLE fitting'' within a parametric framework. \citet{jin2007estimating} targeted this null distribution with a frequency-domain approach based on the empirical characteristic function to estimate the same parameters. \citet{schwartzman2008empirical} proposed a ``mode matching'' method for fitting an empirical null distribution when the theoretical null distribution belongs to any exponential family. Other parametric estimation methods include \citet{park2011estimation} and \citet{gauran2018empirical} that used mixture of normal distributions and zero-inflated discrete mixture distributions, respectively. Although differing in methodology, these works share the common feature of relying on a parametric model that the empirical null distribution follows, and together they highlight the importance of explicitly accounting for the empirical distribution of null test statistics in large-scale inference, rather than relying solely on theoretical null distributions derived for individual tests.

Motivated by the observation that reliance on the theoretical null may be misleading in large-scale multiple testing, we investigate the role of the empirical distributions of test statistics for the null hypotheses in controlling the FDP. We first show that, in an idealized oracle scenario where the empirical distribution of the null test statistics is fully known, FDP control can be achieved with optimal efficiency regardless of the dependence structure among the test statistics and without requiring knowledge of the theoretical null distribution of the test statistics. This suggests that explicitly modeling or adjusting for dependence, while commonly pursued in the literature, may not be necessary when the empirical distribution of the test statistics under null is accurately estimated. Building on this insight, we propose an eFDP method. Specifically, we develop a working model and an estimation procedure for the empirical c.d.f.s of test statistics for both null and alternative hypotheses within the framework of multivariate mixture models. We establish the convergence of our estimators to the corresponding true empirical c.d.f.s, which in turn enables rigorous control of the FDP. While the model formulation and estimation procedure are inspired by \citet{YuQinLi2026}, whose theoretical results rely on independence of observations, our theoretical development is substantially different and more challenging. In particular, we verify that the required technical conditions are mild and illustrate through various examples that they accommodate a broad range of dependence structures as special cases. Finally, we construct an FDP control procedure based on the estimated empirical $p$-values and demonstrate that it asymptotically achieves FDP control, providing a practical approach for large-scale multiple testing that leverages the empirical c.d.f.s rather than dependence modeling.

To evaluate the performance of the proposed eFDP method, we have conducted extensive simulation studies, which show that it achieves more accurate FDP control than existing approaches in most scenarios while maintaining higher power. Notably, under strong dependence structures such as compound symmetry, methods that explicitly model dependence may fail to control the FDP, whereas the eFDP approach remains reliable. We also apply the eFDP approach to the real genomic dataset GSE25066, where it identifies a reasonable proportion of significant genes, capturing the key discoveries made by conservative methods while avoiding the liberal behavior of procedures such as BH. These results suggest that our method provides a stable and interpretable set of discoveries in high-dimensional biological data, complementing the theoretical guarantees established in our study.

The rest of the paper is organized as follows. Section~\ref{section-motivation} shows that, under an oracle setting where the empirical c.d.f.\ of the test statistics for all null hypotheses is known, FDP control can be achieved efficiently regardless of dependence of test statistics and without requiring knowledge of their theoretical null distribution, highlighting the importance of accurately estimating this function. Sections~\ref{section-estimate-empirical-cdf} and~\ref{section-FDP-control} together develop the proposed eFDP method: Section~\ref{section-estimate-empirical-cdf} constructs a nonparametric estimator for the empirical c.d.f.s of null test statistics within a multivariate mixture model framework and establishes its asymptotic convergence, while Section~\ref{section-FDP-control} introduces the resulting eFDP procedure and proves its asymptotic validity. Section~\ref{section-simulation} presents simulation results demonstrating improved FDP control and power, while Section~\ref{section-real-data} illustrates the practical performance of eFDP on a genomic dataset. Section~\ref{section-discussion} concludes the paper. Technical details of Sections~\ref{section-estimate-empirical-cdf} and~\ref{section-FDP-control} are deferred to the supplementary material.

\section{FDP Control with Empirical Distribution} \label{section-motivation}

Let $\mathX = \{X_1, \ldots, X_{p_0}\}$ denote a collection of random variables from Group 0 (the null group), and $\mathY = \{Y_1, \ldots, Y_{p_1}\}$ denote a collection from Group 1 (the alternative group). Let $p = p_0 + p_1$, and define $\mathZ = \{Z_1, \ldots, Z_p\} = \mathX \cup \mathY$. In the context of multiple testing, $Z_i$ represents the test statistic for testing the hypothesis
\begin{eqnarray*}
H_{0,i}: Z_i \in \mathX \quad \text{versus} \quad H_{1,i}: Z_i \in \mathY.
\end{eqnarray*}
Accordingly, we define the null and alternative index sets as
\begin{eqnarray} \label{def-H-0-H-1}
\mathH_0 = \{i: Z_i \in \mathX\}, \quad \mathH_1 = \{i: Z_i \in \mathY\}.
\end{eqnarray}
Throughout this paper, we assume that the test statistics $Z_1, \ldots, Z_p$ contain no ties.

In practice, we observe only $\mathZ$; that is, $\mathX$ and $\mathY$ are not available. In this section, however, we assume that the empirical distribution based on $\mathX$ is accessible:
\[
\FF_0(t) = \frac{1}{p_0} \sum_{i=1}^{p_0} I(X_i \leq t),
\]
and the method for estimating $\FF_0(t)$ will be discussed in the next section.

We emphasize that $\FF_0(t)$ plays a pivotal role in our subsequent development. With its availability, we can show that FDP control can be carried out optimally without imposing any distributional assumptions on the test statistics in $\mathZ$; specifically, the $Z_1, \ldots, Z_p$ may be arbitrarily dependent and have arbitrary distributions.

We note that $\FF_0(t)$ captures the full empirical distributional information of the observed statistics in $\mathX$ and can be used to quantify how closely an observed $Z_i$ aligns with or deviates from the null group. Without loss of generality, we assume that smaller values of $Z_i$ provide stronger evidence against the null; for presentational convenience, this convention is adopted throughout. Accordingly, we define the empirical $p$-value of $Z_i$ as $S_i = \FF_0(Z_i)$. These empirical $p$-values resemble conventional $p$-values commonly used in practice:
\begin{eqnarray}
\{ S_i : i \in \mathH_0 \} = \left\{ \frac{1}{p_0}, \ldots, \frac{p_0}{p_0} \right\}, \label{eq-S-i-H-0}
\end{eqnarray}
which closely mimics the uniform distribution under the null when the test statistics are independent and identically distributed (i.i.d.) and drawn from the theoretical null distribution. In contrast, the set $\{ S_i : i \in \mathH_1 \}$ has a different empirical distribution. Thus, for a given threshold $t \in (0,1)$, the $i$th test can be conducted as $I(S_i \leq t)$, and, following the strategy of \citet{StoreyTaylorSiegmund2004}, we can estimate $\FDP(t)$ by
\begin{eqnarray}
\widehat{\FDP}(t)=\frac{\frac{p_0}{p}\cdot t}{\frac{1}{p}\cdot \max\left\{\sum\limits_{i=1}^pI(S_i\leq t),1\right\}}. \label{eq-def-hat-FDP-t} 
\end{eqnarray}
To control $\FDP(t)$ at pre-specificied level $\alpha \in (0,1)$, let 
\begin{eqnarray}
t_\alpha=\sup\limits_{t}\left\{t:\widehat{\FDP}(t)\leq \alpha\right\}, \label{eq-t-alpha}
\end{eqnarray}
and the corresponding tests are given by $t_i = I(S_i \leq t_\alpha)$, $i=1,\ldots, p$. Furthermore, for each $t$, the corresponding true FDP is given by 
\begin{eqnarray*}
    \FDP(t)= \frac{\sum\limits_{i \in \mathcal{H}_0}I\left(S_i \leq t \right)}{\max \left\{\sum \limits_{i=1}^p I\left(S_i \leq t \right),1\right\}}.
\end{eqnarray*}

We have the following theorem, which formalizes the advantage of using the empirical null distribution for FDP control. Part (a) guarantees that the FDP is controlled at the target level $\alpha$ when the threshold $t_\alpha$ is applied. Part (b) shows that increasing the threshold beyond $t_\alpha$, i.e., adopting a more aggressive rule, either leads to exceeding the nominal FDP level $\alpha$ or does not increase the number of rejections. Taken together, the theorem implies that, with the empirical c.d.f.\ available, the standard FDP control procedure using $t_\alpha$ achieves essentially optimal performance. In practice, this suggests that explicitly modeling dependence among test statistics may be unnecessary with the availability of the empirical null.

\begin{theorem} \label{theorem-1}
With $t_\alpha$ defined by \eqref{eq-t-alpha}, we have 
\begin{enumerate}[(a). ]
    \item  ${\rm{FDP}}(t_\alpha)$ is controlled by $\alpha$: 
    ${\rm{FDP}}(t_\alpha) \leq \alpha$;
    \item  For any $t>t_\alpha$, we have either
    ${\rm{FDP}}(t) >  \alpha$ or $R(t) = R(t_\alpha)$, where $R(t) = \sum \limits_{i=1}^p I\left(S_i \leq t \right)$.  
\end{enumerate}
\end{theorem}
\proof \ Without loss of generality, throughout we assume that $t_\alpha<1$. We show Part (a) first. Observing \eqref{eq-S-i-H-0}, for each $t \in (0,1)$, we have 
\begin{eqnarray*}
\FDP(t) &=& \frac{\sum\limits_{i \in \mathcal{H}_0}I\left(S_i \leq t \right)}{\max \left\{\sum \limits_{i=1}^p I\left(S_i \leq t \right),1\right\}}=\frac{\frac{1}{p}\sum\limits_{i \in \mathcal{H}_0}I\left(S_i \leq t \right)}{\frac{1}{p}\cdot \max \left\{\sum \limits_{i=1}^p I\left(S_i \leq t \right),1\right\}} \\
&=& \frac{\lfloor p_0 t \rfloor/p}{\frac{1}{p}\cdot \max \left\{\sum \limits_{i=1}^p I\left(S_i \leq t \right),1\right\}}. 
\end{eqnarray*}
Thus, based on the definition of $\widehat{\FDP}(t)$ in \eqref{eq-def-hat-FDP-t}, we have 
\begin{eqnarray}
    \FDP(t) = \widehat{\FDP}(t), && \quad \text{if } t = i/p_0, \text{ for some } i\in \{1,\ldots, p_0\}; \nonumber \\
    \FDP(t) < \widehat{\FDP}(t), && \quad \text{otherwise}. \label{eq-thm-1-1}
\end{eqnarray}
In particular, setting $t = t_\alpha$ in \eqref{eq-thm-1-1} and observing the definition of $t_\alpha$ given by \eqref{eq-t-alpha}, Part (a) is valid. 

We proceed to show Part (b). By the definition of $S_i$, we have
\begin{eqnarray}
    S_i = \frac{1}{p_0} \sum \limits_{j=1}^{p_0} I(X_j \leq Z_i) \in \{0/p_0, 1/p_0, \ldots, p_0/p_0\}.  \label{eq-thm-1-2}
\end{eqnarray}
Note that $t_\alpha< t < 1$ can be partitioned to be 
\begin{eqnarray*}
    (t_\alpha, 1) = \left(t_\alpha, \frac{\lfloor t_\alpha p_0 \rfloor + 1)}{p_0}\right) \bigcup \left\{\bigcup_{i=1}^{p_0 - \lfloor t_\alpha p_0 \rfloor-1}\left[\frac{\lfloor t_\alpha p_0 \rfloor + i}{p_0}, \frac{\lfloor t_\alpha p_0 \rfloor + i+1}{p_0}\right)\right\}. 
\end{eqnarray*}
We consider the above sets one by one. 
If $t\in (t_\alpha, (\lfloor t_\alpha p_0 \rfloor + 1)/p_0)$, we have
\[\left(t_\alpha, \frac{\lfloor t_\alpha p_0 \rfloor + 1}{p_0}\right)\subset \left(\frac{\lfloor t_\alpha p_0 \rfloor}{p_0}, \frac{\lfloor t_\alpha p_0 \rfloor + 1}{p_0} \right),\]
which together with \eqref{eq-thm-1-2} implies that $S_i \notin (t_\alpha, (\lfloor t_\alpha p_0 \rfloor + 1)/p_0)$ for every $i\in\{1,\ldots, p\}$. Thus, $R(t) = R(t_\alpha)$. 

If $t\in\left[\frac{\lfloor t_\alpha p_0 \rfloor + i}{p_0}, \frac{\lfloor t_\alpha p_0 \rfloor + i+1}{p_0}\right)$, for an $i\in \{1,\ldots, p_0 - \lfloor t_\alpha p_0 \rfloor-1\}$, we have 
\[
\left[\frac{\lfloor t_\alpha p_0 \rfloor + i}{p_0}, \frac{\lfloor t_\alpha p_0 \rfloor + i+1}{p_0}\right) \cap \{0/p_0, 1/p_0, \ldots, p_0/p_0\} = \frac{\lfloor t_\alpha p_0 \rfloor + i}{p_0},
\]
which together with \eqref{eq-thm-1-2}, the definition of $\FDP(t)$, and \eqref{eq-thm-1-1} leads to
\begin{eqnarray*}
    \FDP(t) = \FDP\left(\frac{\lfloor t_\alpha p_0 \rfloor + i}{p_0}\right) = \widehat{\FDP}\left(\frac{\lfloor t_\alpha p_0 \rfloor + i}{p_0}\right). 
\end{eqnarray*}
Consequently, since $\frac{\lfloor t_\alpha p_0 \rfloor + i}{p_0}> t_\alpha$, and by the definition of $t_\alpha$ given in \eqref{eq-t-alpha}, we have 
\begin{eqnarray*}
    \FDP(t) = \widehat{\FDP}\left(\frac{\lfloor t_\alpha p_0 \rfloor + i}{p_0}\right) > \alpha. 
\end{eqnarray*}
We have verified Part (b) and completed the proof of this theorem. \epf

\section{Empirical Distribution Estimation with Multivariate Mixture Models} \label{section-estimate-empirical-cdf}

\subsection{Working model for the empirical c.d.f.s} \label{Section-Working-Model}

Following the discussion in Section~\ref{section-motivation}, for data with unknown dependence structure, the key challenge in achieving reliable FDP control lies in knowledge of the empirical distribution of the test statistics under the null. We propose an estimation method for these empirical distributions in a framework where the test statistics are multivariate random vectors. With a slight abuse of notation, we denote the observed test statistics for the $i$th hypothesis as $Z_i = (Z_{i,1}, \ldots, Z_{i,K})^T$, a $K$-dimensional vector with $K \geq 3$. As in Section~\ref{section-motivation}, let $\mathZ = \{Z_1, \ldots, Z_p\} = \mathX \cup \mathY$, where $\mathX = \{X_1, \ldots, X_{p_0}\}$ and $\mathY = \{Y_1, \ldots, Y_{p_1}\}$ correspond to the collections of vectors from the null group (Group 0) and alternative group (Group 1), respectively. For each dimension $k = 1, \ldots, K$, we denote by $\mathZ_k = \{Z_{1,k}, \ldots, Z_{p,k}\}$ the observed statistics on the $k$th dimension, and similarly define $\mathX_k$ and $\mathY_k$ so that $\mathZ_k = \mathX_k \cup \mathY_k$. 

In practice, the multivariate test statistics $Z_i$ can be obtained by splitting the available raw data across dimensions. For example, in our simulation study (Section~\ref{Sim-Data-Processing}), we describe an approach to generate such test statistics for two-group gene expression data. Similar data-splitting strategies can be applied in other practical settings to construct multidimensional test statistics that satisfy the structure as needed.

Let $\FF(\cdot)$ denote the joint empirical c.d.f.\ based on $\mathZ$, and let $\FF_0(\cdot)$ and $\FF_1(\cdot)$ denote the joint empirical c.d.f.s based on $\mathX$ and $\mathY$, respectively. For each dimension $k$, let $\FF_{0,k}(\cdot)$ and $\FF_{1,k}(\cdot)$ denote the corresponding marginal empirical c.d.f.s based on $\mathX_k$ and $\mathY_k$, respectively. We propose a working model to guide the estimation of $\FF_{0,k}(\cdot)$ and $\FF_{1,k}(\cdot)$ within the framework of multivariate mixture models. Specifically, for each $k = 1, \ldots, K$, we pretend that $X_{i,k}, i = 1, \ldots, p_0$, are i.i.d. draws from $\FF_{0,k}(\cdot)$, and $Y_{i,k}, i = 1, \ldots, p_1$, are i.i.d. draws from $\FF_{1,k}(\cdot)$. Furthermore, for each $i$, we pretend that conditional independence across dimensions is satisfied: given that $i \in \mathH_m$ for $m = 0$ or $1$, the components $Z_{i,1}, \ldots, Z_{i,K}$ are independent. Under this working model, the vectors $Z_i$, $i = 1, \ldots, p$, are treated as i.i.d.\ random vectors with c.d.f.
\begin{eqnarray}
     \FF_{\theta}(z) = \lambda_0  \FF_{0,\theta}(z) + \lambda_1 \FF_{1,\theta}(z), \label{eq-Emp-Dist-Z}
\end{eqnarray}
where
\[
\FF_{0,\theta}(z) = \prod_{k=1}^K F_{0,k}(z_k), \qquad  
\FF_{1,\theta}(z) = \prod_{k=1}^K F_{1,k}(z_k),
\]
and $\lambda_1 = 1 - \lambda_0$. The parameter is given by
\begin{eqnarray*}
    \theta = \left\{\lambda_0, F_{0,1}, \ldots, F_{0,K}, F_{1,1}, \ldots, F_{1,K} \right\},
\end{eqnarray*}
which belongs to the parameter space
\begin{eqnarray}
\Theta = \left\{\theta: \lambda_0 \in (0,1),\; F_{m,k}(\cdot) \text{ are c.d.f.s},\; m = 0,1,\; k = 1, \ldots, K \right\}. \label{eq-def-Theta}
\end{eqnarray}
The corresponding true parameter value is denoted by
\begin{eqnarray*}
    \theta_0 = \left\{\lambda_{0,0}, \FF_{0,1}, \ldots, \FF_{0,K}, \FF_{1,1}, \ldots, \FF_{1,K} \right\},
\end{eqnarray*}
where $\lambda_{0,0} = p_0/p$.

\subsection{Binomial likelihood estimation of parameters} \label{Section-Bin-Estimation}

We consider a binomial likelihood approach for estimating $\theta$ \citep{Qin2014,YU2023, YuQinLi2026}. For each $i,j=1,\ldots, p, i\neq j$, let $I_{i,j} = I(Z_i \leq Z_j) = (I_{i,j,1}, \ldots, I_{i,j,K})^T$, where $I_{i,j,k} = I(Z_{i,k} \leq Z_{j,k})$. Then, if we pretend that $Z_i$ are i.i.d. and follow the model \eqref{eq-Emp-Dist-Z}, we establish a binomial likelihood objective function, with a similar strategy as that in \citet{YuQinLi2026}. In particular, for any $z = (z_1, \ldots, z_K)^T \in \mathbb{R}^{K}$, we have 
\begin{eqnarray*}
    P\left(I(Z_i \leq Z_j) = I(z\leq Z_j) \Big| Z_j \right) = \sum_{m=0}^1\lambda_m \prod_{k=1}^K F_{m,k}^{I(z_k \leq Z_{j,k})}(Z_{j,k})\cdot \bar F_{m,k}^{I(z_k > Z_{j,k})}(Z_{j,k}),
\end{eqnarray*}
where $\bar F_{m,k}(\cdot) = 1-F_{m,k}(\cdot)$.
Thus, by treating $I(Z_i \leq Z_j)$ as the multivariate binary response data, and borrow the idea of the composite likelihood \citep{Cristiano2011, Kwonsang2023}, we establish the log-likelihood: 
\begin{eqnarray*}
    \ell(\theta) = \sum_{j=1}^p \sum_{i=1}^p \log \left\{ \sum_{m=0}^1\lambda_m \prod_{k=1}^K F_{m,k}^{I_{i,j,k}}(Z_{j,k})\cdot \bar F_{m,k}^{1-I_{i,j,k}}(Z_{j,k}) \right\}. 
\end{eqnarray*}
The estimator for $\theta$ is defined as
\begin{eqnarray}
    \widehat \theta = \left\{\widehat \lambda_0, \widehat \FF_{0,1}(\cdot), \ldots, \widehat \FF_{1,K}(\cdot) \right\} = \arg\max_{\theta \in \Theta} \ell(\theta), \label{eq-theta-hat-def}
\end{eqnarray}
where $\Theta$ is defined by \eqref{eq-def-Theta}.

The optimization problem \eqref{eq-theta-hat-def} can be solved numerically using the following algorithm. For $s = 1,2,\ldots$, let $\theta^{(s)} = \left\{\lambda_m^{(s)}, \FF_{m,k}^{(s)}, m=0,1; k = 1,\ldots, K\right\}$ be the estimate in the $s$th step. Let 
\begin{eqnarray*}
    \lambda_{i,j,m}^{(s)} = \frac{\lambda_m^{(s)}\prod_{k=1}^K \left[ \left\{ \FF_{m,k}^{(s)}(Z_{j,k}) \right\}^{I_{i,j,k}} \left\{ \bar\FF_{m,k}^{(s)}(Z_{j,k}) \right\}^{1-I_{i,j,k}}  \right]}{\sum_{m_1=1}^2 \lambda_{m_1}^{(s)}\prod_{k=1}^K \left[ \left\{ \FF_{m_1,k}^{(s)}(Z_{j,k}) \right\}^{I_{i,j,k}} \left\{ \bar\FF_{m_1,k}^{(s)}(Z_{j,k}) \right\}^{1-I_{i,j,k}}  \right]}. 
\end{eqnarray*}
Then, we obtain $\theta^{(s+1)}$ by 
\begin{eqnarray*}
    \lambda_0^{(s+1)} = \frac{1}{p^2}\sum_{j=1}^p\sum_{i=1}^p  \lambda_{i,j,0}^{(s)},
\end{eqnarray*}
$\lambda_1^{(s+1)} = 1 - \lambda_0^{(s+1)}$, and 
\begin{eqnarray*}
    \FF_{m,k}^{(s+1)}(\cdot) = \arg\max_{F_{m,k}} \sum_{j=1}^p w_{m,j}^{(s)} \left[ \xi_{m,k,j}^{(s)} \log F_{m,k}(Z_{j,k}) + \left\{1-\xi_{m,k,j}^{(s)}\right\} \log \bar F_{m,k}(Z_{j,k}) \right],
\end{eqnarray*}
subject to $F_{m,k}(\cdot)$ being a c.d.f., where
\begin{eqnarray*}
    \xi_{m,k,j}^{(s)} = \frac{\sum_{i=1}^p \lambda_{i,j,m}^{(s)} I_{i,j,k}}{\sum_{i=1}^p \lambda_{i,j,m}^{(s)}}, \qquad  w_{m,j}^{(s)} = \sum_{i=1}^p \lambda_{i,j,m}^{(s)}. 
\end{eqnarray*}

More details of the development and the theoretical results for the convergence of the algorithm are given in Section 2 of the supplementary material. 

\subsection{Asymptotic properties} \label{section-asymp}

We explore the asymptotic properties of $\widehat \theta$. Notably, the developments in Sections~\ref{Section-Working-Model} and~\ref{Section-Bin-Estimation} are motivated by the intuition that $Z_1, \ldots, Z_p$ can be pretended to be i.i.d.\ and follow \eqref{eq-Emp-Dist-Z}. Consequently, the methodology and algorithm presented in Section~\ref{Section-Bin-Estimation} are similar in strategy to those in \citet{YuQinLi2026}. However, due to the dependence among test statistics, the theoretical analysis is substantially different and more challenging. To address this, we impose the following condition specifying the required dependence structure among the test statistics.
\begin{Condition}\label{Condition-1}
For $m=0, 1$, $\FF_m(\cdot)$ and $\FF_{m,\theta_0}(\cdot)$ satisfies, as $p\to \infty$,
\begin{eqnarray*}
\sup_{z\in \mathbb{R}^K}|\FF_m(z) - \FF_{m,\theta_0}(z)| = o\left\{ (\log p)^{-1} \right\}, \quad a.s.
\end{eqnarray*}
\end{Condition}
We observe that Condition~\ref{Condition-1} essentially imposes a requirement on the dependence structure of the test statistics under the empirical measure $\FF_m(\cdot)$. In particular, this condition can be expressed as
\begin{eqnarray*}
    \sup_{z \in \mathbb{R}^K} \left| E\left\{\prod_{k=1}^K I(Z_k^* \leq z_k)\right\} - \prod_{k=1}^K E\left\{I(Z_k^* \leq z_k)\right\} \right| = o\left\{ (\log p)^{-1} \right\}, \quad \text{a.s.},
\end{eqnarray*}
where $E$ is taken under $Z^* = (Z_1^*, \ldots, Z_K^*)^T \sim \FF_m$ for $m = 0$ or $1$. Intuitively, for $K=2$ and $m=0$, the difference
\[
E\left\{\prod_{k=1}^2 I(Z_k^* \leq z_k)\right\} - \prod_{k=1}^2 E\left\{I(Z_k^* \leq z_k)\right\}
\]
corresponds to the sample covariance between $\{I(X_{i,1} \leq z_1), i=1,\ldots, p_0\}$ and $\{I(X_{i,2} \leq z_2), i=1,\ldots, p_0\}$. Likewise, for $K=2$ and $m=1$, it corresponds to the sample covariance between $\{I(Y_{i,1} \leq z_1), i=1,\ldots, p_1\}$ and $\{I(Y_{i,2} \leq z_2), i=1,\ldots, p_1\}$.

In multiple testing research, it is common to assume specific dependence structures; see, for example, \citet{benjamini2001control, sun2009large, fan2012estimating}. Condition~\ref{Condition-1} is relatively mild and can be satisfied by many of the dependence structures frequently assumed in the literature and encountered in practice. In the following, we provide several examples. For continuity of presentation, the proofs of these examples are deferred to the supplementary material.
\begin{Example} \label{Example-1}
    Assume that there exist multivariate c.d.f.s $F_0^*(z)$ and $F_1^*(z)$ (possibly random) defined on $z\in \mathbb{R}^K$, such that:
\begin{equation}
    \sup_{z} |\FF_0(z) - F_0^*(z)| = o\left\{ (\log p)^{-1} \right\} \quad \text{and} \quad \sup_{z} |\FF_1(z) - F_1^*(z)| = o\left\{ (\log p)^{-1} \right\}, \quad a.s., \label{eq-example-1-1}
\end{equation}
as $p\to \infty$. Furthermore, $F_0^*(z)$ and $F_1^*(z)$ have the form: 
\begin{eqnarray*}
    F_0^*(z) = \prod_{k=1}^K F_{0,k}^*(z_k) \quad \text{and} \quad  F_1^*(z) = \prod_{k=1}^K F_{1,k}^*(z_k),
\end{eqnarray*}
for $F_{m,k}^*(\cdot), m=0,1; k=1,\ldots, K$ being c.d.f.s defined on $\mathbb{R}$. Then, Condition \ref{Condition-1} is satisfied. 
\end{Example}

\begin{Example} \label{Example-2}
    Let $U_1,\ldots, U_K$ be i.i.d. random variables, such that $E(U_k) = \mu_k, \text{var}(U_k) = 1$; let $\epsilon_{i,k}, i=1,\ldots,\widetilde p; k=1,\ldots, K$ be i.i.d. random variables such that $E(\epsilon_{i,k}) = 0, \text{var}(\epsilon_{i,k}) = 1$. Suppose that $\{U_k, k=1,\ldots, K\}$ and $\{\epsilon_{i,k}, i=1,\ldots, \widetilde p; k=1,\ldots, K\}$ are independent. Let $V_i = (V_{i,1},\ldots, V_{i,K})^T, i=1,\ldots, \widetilde p$, where 
    \begin{eqnarray}
    V_{i,k} = \sqrt{\rho} U_k + \sqrt{1-\rho} \epsilon_{i,k}, \label{eq-Example-2-1}
    \end{eqnarray}
    with $\rho\in [0,1)$ being a given constant. Then for each dimension $k = 1,\ldots, K$, $V_{1,k},\ldots, V_{\widetilde p, k}$ are dependent random variables, with the variance-covariance matrix having the {\it compound symmetry structure}. Denote by $\FF_V(\cdot)$ and $\FF_{V,k}(\cdot)$ the empirical c.d.f.s based on $\{V_i, i=1,\ldots, \widetilde p\}$ and $\{V_{1,k},\ldots, V_{\widetilde p,k}\}$ respectively. We have, as $\widetilde p \to \infty$,
    \begin{eqnarray}
        \sup_{z\in \mathbb{R}^K} \left| \FF_V(z) - \prod_{k=1}^K \FF_{V,k}(z_k)\right| = O\left\{\sqrt{\frac{\log(\widetilde p)}{\widetilde p}} \right\}, \quad a.s. \label{eq-Example-2-1-1}
    \end{eqnarray}
The above conclusion implies that if the test statistics in $\mathX$ and in $\mathY$ respectively satisfy the compound symmetry structure \eqref{eq-Example-2-1}, then Condition \ref{Condition-1} is satisfied. 
\end{Example}

\begin{Example}  \label{Example-3}
The $\alpha$--mixing coefficient of two $\sigma$--algebras $\mathA$ and $\mathB$ is defined as
\begin{eqnarray}
    \alpha(\mathA, \mathB) = \sup_{A\in \mathA, B\in \mathB}\left| P(A\cap B) - P(A)P(B)\right|. \nonumber 
\end{eqnarray}
Let $\mathV = \{V_t\}_{t=1,\ldots,\widetilde p}$ be a sequence of $K$-dimensional stationary time series with c.d.f. $F_V(\cdot)$ for $V_t$. Define the $\sigma$-algebras:
\[
\mathA_i = \sigma(V_t, t\leq i), \qquad 
\mathB_j = \sigma(V_t, t \geq j),
\]
and denote
\begin{eqnarray}
\alpha_V(n) = \sup_{k\geq 1}\alpha\left(\mathA_k, \mathB_{k+n} \right). \nonumber 
\end{eqnarray}
Let $\FF_V(v)$ be the empirical c.d.f. of $\{V_t\}_{t=1}^{\widetilde p}$.  
If $\mathV$ satisfies
\begin{eqnarray}
    \alpha_{V}(n) &\leq& \exp(-cn), \label{eq-Condition-3-1} \\
    \xi^2 &<&  \infty, \label{eq-Condition-3-2} \\
    V_{t,1},\ldots,V_{t,K} &&\text{are independent for each } t, \label{eq-Condition-3-3} 
\end{eqnarray}
where $c>0$ is a universal constant, 
\begin{eqnarray}
\xi^2 = \sup_{v\in \mathbb{R}^K} \sup_{t\geq 1}\left[\text{var}\{W_t(v)\} + 2 \sum_{j > t}\left|\text{cov}\{W_t(v), W_j(v)\}\right| \right], \nonumber 
\end{eqnarray}
and $W_t(v)=I(V_t\le v)-F_V(v)$, then
\begin{equation}
\sup_{v\in \mathbb{R}^K}
\left|\FF_{V}(v) - \prod_{k=1}^K \FF_{V,k}(v_k)\right|
=
O\!\left(\sqrt{\frac{\log \widetilde p}{\widetilde p}}\right),\quad a.s. \nonumber 
\end{equation}
Consequently, if both $\mathX$ and $\mathY$ satisfy \eqref{eq-Condition-3-1}--\eqref{eq-Condition-3-3}, then Condition \ref{Condition-1} holds.
\end{Example}

We observe that the main requirement in Example \ref{Example-3} is the $\alpha$-mixing condition. This is a standard and relatively weak notion of dependence, requiring that dependence between distant observations vanishes asymptotically. Among common mixing concepts, it is one of the weakest, as stronger notions (e.g., $\beta$- or $\phi$-mixing) imply $\alpha$-mixing but not conversely \citep[see, e.g., Section 2.1]{Doukhan1994} and the discussion in \citet{Bradley2005}. Consequently, $\alpha$-mixing accommodates a broad class of stochastic processes and is widely used in statistics and econometrics to establish asymptotic results for dependent data; in particular, many standard time series models, including ARMA processes, satisfy $\alpha$-mixing conditions under mild assumptions \citep[Chapter 2]{FanYao2003}. Building on Example \ref{Example-3}, we next consider a more concrete setting. In Example \ref{Example-4}, we show that Condition \ref{Condition-1} is satisfied when the test statistics from the null and alternative groups respectively follow ARMA-type dependence structures, thereby demonstrating that our framework applies to commonly used time series models.

\begin{Example} \label{Example-4}
    Let $V_t, t = 1,\ldots, p$, be a $K$-dimensional stationary $\text{ARMA}(c_1, c_2)$ time series, with error terms being i.i.d. continuous random vectors. For each $t$, $V_{t,1}, \ldots, V_{t,K}$ are independent. Recall $\mathH_0, \mathH_1$ defined in \eqref{def-H-0-H-1}. Let $Z_t = V_t$, if $t \in \mathH_0$, and let $Z_t = V_t + \mu$ if $t \in \mathH_1$, where $\mu \neq 0$ is a constant. Then Condition \ref{Condition-1} is satisfied.  
\end{Example}

Besides Condition \ref{Condition-1}, we need the following Conditions. 

\begin{Condition}\label{Condition-2}
There exist multivariate c.d.f.s $F_0^*(z)$ and $F_1^*(z)$ (possibly random; see Example \ref{Example-2}) defined on $z\in \mathbb{R}^K$, such that:
\begin{eqnarray*}
    \sup_{z} |\FF_0(z) - F_0^*(z)| = o(1) \quad \text{and} \quad \sup_{z} |\FF_1(z) - F_1^*(z)| = o(1), \quad a.s.,
\end{eqnarray*}
as $p\to \infty$.
\end{Condition}

\begin{Condition}\label{Condition-3}
    $\lambda_{0,0} = p_0/p$ satisfies
    \begin{eqnarray*}
        \lambda_{0,0} = \lambda_0^* + o(1), \quad \text{as }p\to \infty,
    \end{eqnarray*}
    with $0<\lambda_0^*<1$ being a constant. Denote $\lambda_1^* = 1- \lambda_0^*$.
\end{Condition}

\begin{Condition}\label{Condition-4}
Consider $F_m^*(z)$ given in Condition \ref{Condition-2}, and $\lambda_m^*$ given in Condition \ref{Condition-3}, $m=0,1$. Let 
\begin{eqnarray*}
F^*(z) = \lambda_0^* F_0^*(z) + \lambda_1^* F_1^*(z), 
\end{eqnarray*}
and for $m=0,1$, denote by $F_{m,k}^*(\cdot)$ the $k$th marginal c.d.f. of $F_m^*(\cdot)$. For any $\theta_1, \theta_2 \in \Theta$, where we denote $\theta_r = \{\lambda_{0,r}, F_{0,1,r}, \ldots, F_{1,K,r}\}$, $r = 1,2$. For each $\omega$ in the sample space, if 
\begin{eqnarray*}
    \int \left\{ \FF_{\theta_1}^\omega (z) - \FF_{\theta_2}^\omega(z) \right\}^2 d F^{*\omega} (z) =0, 
\end{eqnarray*}
then $\lambda_{0,1} = \lambda_{0,2}$ and $F_{m_1, k, 1}^{\omega}(\cdot) = F_{m_1,k,2}^{\omega}(\cdot)$ almost surely in $F_{m_2,k}^{*\omega}(\cdot)$ for every $m_1, m_2 \in \{0,1\}, k=1,\ldots, K$.
\end{Condition}

\begin{Remark}
Conditions \ref{Condition-1}--\ref{Condition-4} are mild. 
Condition \ref{Condition-1} imposes requirements on the dependence structure of the test statistics; as illustrated in Examples \ref{Example-1}--\ref{Example-4}, it is satisfied by a wide range of dependence settings. Conditions \ref{Condition-2} and \ref{Condition-3} essentially require the convergence of the joint empirical c.d.f.s under $\mathH_0$ and $\mathH_1$, as well as the proportion $\lambda_{0,0}$. Condition \ref{Condition-4} ensures the identifiability of the model.
\end{Remark}

We have the following theorem, which establishes the asymptotic properties of our estimators. These properties play pivotal roles in the subsequent development for the FDP control. 

\begin{theorem} \label{theorem-2}
Assume Conditions \ref{Condition-1}--\ref{Condition-4}. We have, as $p\to \infty$,
\begin{enumerate}[(a). ]
    \item $\widehat \lambda_0 = p_0/p + o(1)$, a.s. 

    \item For $m_1, m_2 \in \{0,1\}$, and $k=1,\ldots, K$, we have 
    \begin{eqnarray*}
        \int \left\{\widehat \FF_{m_1,k}(t) - \FF_{m_1, k}(t) \right\}^2 d \FF_{m_2,k}(t) = o(1), \quad a.s.
    \end{eqnarray*}
    
\end{enumerate}
    
\end{theorem}

\begin{Remark}
Part (b) of Theorem~\ref{theorem-2} establishes the convergence of the empirical c.d.f. estimators to the true empirical c.d.f.s, rather than the theoretical c.d.f.s, the latter of which is often the focus in classical statistical theory. Importantly, this convergence holds both under the null and alternative hypotheses. In the mixture model framework, such results are useful beyond FDP control, for example, in inference on mixture components. Part (a) provides a useful byproduct, showing that the estimator for the proportion of null hypotheses converges to the true proportion. Estimating this proportion has been extensively studied in the literature \citep{storey2003statistical, langaas2005estimating, wang2010slim} and plays a critical role in many multiple testing procedures. Since our procedure provides a consistent estimate of the true proportion, it also allows for reliable intuition about the potential power of some multiple testing method: if the proportion of alternative hypotheses identified by the method is much smaller than the estimated proportion, the procedure may have limited power.
\end{Remark}

\section{FDP Control} \label{section-FDP-control}

In Section~\ref{section-estimate-empirical-cdf}, within the framework of multivariate mixture data, we developed estimators for the empirical c.d.f.s and established their asymptotic properties under dependence. Specifically, for $K$-dimensional mixture data $\mathZ = \{Z_1, \ldots, Z_p\}$, we constructed $\widehat \FF_{m,k}(\cdot)$ as an estimator of $\FF_{m,k}(\cdot)$, the empirical distribution of observations in group $m$ along dimension $k$. In this section, we build on these empirical c.d.f.\ estimators of the null distributions to develop the eFDP approach for controlling the FDP.

Based on the discussion in Section~\ref{section-motivation}, one could in principle select a single dimension $k$ and perform multiple testing using $\widehat \FF_{m,k}(\cdot)$. However, such an approach is inefficient from an information perspective, as it ignores the empirical c.d.f.\ estimates from other dimensions, which may also carry important group-related information.


On the other hand, following the empirical $p$-value strategy in Section~\ref{section-motivation}, for each $i = 1, \ldots, p$, we can obtain a set of empirical $p$-value estimates. With a slight abuse of notation, we denote this set by $\widehat S_i$:
\begin{eqnarray*}
    \widehat S_i = \left\{\widehat F_{0,1}(Z_{i,1}), \ldots, \widehat F_{0,K}(Z_{i,K})\right\}. 
\end{eqnarray*}
Following the strategy of \citet{zhang2011multiple}, to better utilize the available information, we consider performing the tests based on an aggregated version of $\widehat S_i$. In particular, we consider
\begin{eqnarray*}
    \widehat M_i = \text{median}\left(\widehat S_i\right). 
\end{eqnarray*}
In line with the approach in Section~\ref{section-motivation}, our FDP control procedure proceeds as follows. For a given threshold $t \in (0,1)$, the $i$th test is conducted as $I(\widehat M_i \leq t)$, and, adopting the approach of \citet{zhang2011multiple}, we estimate $\FDP(t)$ by
\begin{eqnarray}
\widehat{\FDP}(t)=\frac{\widehat \lambda_0\cdot M(t)}{\frac{1}{p}\cdot \max\left\{\sum\limits_{i=1}^pI(\widehat M_i\leq t),1\right\}}, \label{eq-def-hat-FDP-t-median} 
\end{eqnarray}
where $M(t)$ denotes the c.d.f. of the median of $K$ i.i.d. $\text{Uniform}(0,1)$ random variables. To control $\FDP(t)$ at a pre-specified level $\alpha \in (0,1)$, let 
\begin{eqnarray}
t_\alpha=\sup\limits_{t}\left\{t:\widehat{\FDP}(t)\leq \alpha\right\}, \label{eq-t-alpha-median}
\end{eqnarray}
and the corresponding tests are given by $t_i = I(\widehat M_i \leq t_\alpha)$, $i=1,\ldots, p$. Since for each $t$, the true value of $\FDP(t)$ is defined by 
\begin{eqnarray}
\text{FDP}(t)=\frac{\sum\limits_{i \in \mathcal{H}_0}I(\hat{M}_i\leq t)}{\max\left\{\sum\limits_{i=1}^p I(\widehat{M}_i\leq t),1\right\}}, \label{eq-def-FDP-t-median}
\end{eqnarray}
the true FDP based on our procedure is given by $\text{FDP}(t_\alpha)$. In addition to Conditions \ref{Condition-1}--\ref{Condition-4} given in Section \ref{Section-Bin-Estimation}, we need the following condition. 

\begin{Condition} \label{Condition-5}
    The empirical c.d.f.s of the empirical $p$-values $\FF_{0,j}(Y_{i,j}), i = 1, \ldots, p_1$ for Group 1, i.e., 
    \begin{eqnarray*}
        G_j(t) = \frac{1}{p_1} \sum_{i=1}^{p_1} I \left\{\FF_{0,j}(Y_{i,j}) \leq t \right\}, \quad j = 1,\ldots, K,
    \end{eqnarray*}
    satisfies, as $p\to \infty$,
    \begin{eqnarray}
        \sup_{t\in [0,1]} |G_j(t) - G_j^*(t)| = o(1),  \quad a.s., \label{eq-Condition-5-1}
    \end{eqnarray}
    for $G_j^*(t), j=1,\ldots, K$ being  Lipschitz continuous functions. Furthermore, for the specified $\alpha > 0$, there exists a $t_\alpha^\infty > 0$, such that 
    \begin{eqnarray}
        \frac{\lambda_0^*}{1-\lambda_0^*} \cdot \frac{1-\alpha}{\alpha} M(t_\alpha^\infty) <  G^*(t_\alpha^\infty), \label{eq-Condition-5-2}
    \end{eqnarray}
    where 
    \begin{eqnarray*}
        G^*(t) = \sum\limits_{k=\left\lceil \frac{K}{2} \right\rceil }^K\sum\limits_{|S|=k} \prod_{j\in S} G_j^*(t) \prod_{j \notin S}\left\{1-G_j^*(t) \right\},
    \end{eqnarray*}
    and $S$ denotes a subset of $\{1, \ldots, K\}$, $|S|$ is the number of elements in set $S$. 
    
\end{Condition}

\begin{Remark}
Condition~\ref{Condition-5} imposes two intuitive and, in fact, mild requirements on the empirical $p$-values. First, \eqref{eq-Condition-5-1} ensures that, as the number of tests $p$ grows, the empirical c.d.f.s of the $p$-values from the alternative group converge to well-behaved, Lipschitz continuous limit functions $G_j^*$. This convergence guarantees the stability of the empirical $p$-value distributions and supports reliable large-sample inference. Second, \eqref{eq-Condition-5-2} encodes a form of stochastic ordering between the limiting distributions under the null and alternative hypotheses, requiring that $p$-values corresponding to true alternatives are, on average, smaller than those under the null. This separation ensures that a threshold $t_\alpha^\infty$ exists such that the FDP control procedure can effectively distinguish true signals from nulls in the asymptotic limit. Together, these conditions formalize the intuition that stable and well-separated empirical $p$-value distributions are sufficient for the validity of FDP control in large-scale testing.
\end{Remark}

We have the following theorem. 

\begin{theorem}\label{median fdr theorem}
     Assume Conditions \ref{Condition-1}--\ref{Condition-5}. For any $\alpha\in(0,1)$,  $t_\alpha$ is given by \eqref{eq-t-alpha-median}. When $p\to \infty$, we have 
     \begin{eqnarray*}
     {\rm{FDP}}(t_\alpha) \leq \alpha +o(1),\quad a.s. 
     \end{eqnarray*}
\end{theorem}

\begin{Remark}
The development for Theorem \ref{median fdr theorem} is technically involved. A primary difficulty arises from the fact that the empirical $p$-values used in the FDP procedure are dependent, with a complex and largely unknown dependence structure; Condition \ref{Condition-1} imposes only a mild requirement on the dependence. This complicates the analysis of the FDP as a stochastic quantity. In addition, as discussed in Section \ref{section-asymp}, the estimators of the null c.d.f.s are shown to converge to their empirical counterparts under the empirical measure only in the $L_2$ sense. Since these estimators are subsequently used to form the empirical $p$-values, the lack of stronger convergence (such as uniform convergence) poses further challenges in rigorously establishing FDP control. To the best of our knowledge, both the nonparametric estimation of the empirical null c.d.f.s and the establishment of their convergence to the empirical c.d.f.s under general dependence are new, which further adds to the technical difficulty of the analysis.
\end{Remark}

\section{Simulation} \label{section-simulation} 


\subsection{Data simulation}

The simulated data mimic the structure of gene expression data used in the real data analysis in Section \ref{section-real-data}. Throughout the numerical studies, we set $p = 1000$ and $n_0 = n_1 = 120$. We generate $U_1, \ldots, U_{n_0}$ as samples from the control group and $V_1, \ldots, V_{n_1}$ from the diseased group. For each subject, $U_i$ (or $V_i$) is $p$-dimensional, corresponding to $p$ genes; that is, $U_i = (U_{i,1}, \ldots, U_{i,p})^T$ and $V_i = (V_{i,1}, \ldots, V_{i,p})^T$. We randomly select $p_1 = 200$ indices from $\{1, \ldots, p\}$ to form $\mathH_1$, and let the remaining indices form $\mathH_0$. 

The data are generated using a normal copula approach such that, marginally, $U_{i,j} \sim F(\cdot)$ for $i=1,\ldots,n_0$ and $j=1,\ldots,p$; $V_{i,j} \sim F(\cdot)$ for $i=1,\ldots,n_1$ and $j \in \mathH_0$; and $V_{i,j} \sim G(\cdot)$ for $i=1,\ldots,n_1$ and $j \in \mathH_1$, where $F(\cdot)$ and $G(\cdot)$ will be specified later. 

Specifically, let $\Sigma_\rho$ be a $p \times p$ covariance matrix with unit diagonal entries. We generate $W_i = (W_{i,1}, \ldots, W_{i,p})^T \sim \mathcal{N}(0_p, \Sigma_\rho)$ and define
\begin{eqnarray*}
    U_{i,j} = F^{-1}\left( \Phi(W_{i,j}) \right), \quad j=1,\ldots,p,
\end{eqnarray*}
where $\Phi(\cdot)$ is the standard normal c.d.f. The dependence among $U_{i,1}, \ldots, U_{i,p}$ is induced by the off-diagonal entries of $\Sigma_\rho$, which will be specified later. Similarly, we generate $\widetilde W_i = (\widetilde W_{i,1}, \ldots, \widetilde W_{i,p})^T \sim \mathcal{N}(0_p, \Sigma_\rho)$ and define
\begin{eqnarray*}
    V_{i,j} = F^{-1}\left( \Phi(\widetilde W_{i,j}) \right), \quad j \in \mathH_0; \qquad 
    V_{i,j} = G^{-1}\left( \Phi(\widetilde W_{i,j}) \right), \quad j \in \mathH_1,
\end{eqnarray*}
so that $V_{i,j} \sim F$ for $j \in \mathH_0$ and $V_{i,j} \sim G$ for $j \in \mathH_1$ marginally.

Based on the choices of $F(\cdot)$, $G(\cdot)$, and $\Sigma_\rho$, we consider the following settings:
\begin{itemize}
    \item Study 1: $F \sim \mathcal{N}(1, 1)$, $G \sim \mathcal{N}(0.6, 1)$, and $\Sigma_\rho$ follows a compound symmetry structure with $\sigma_{i,j} = \rho$ for $i \neq j$;

    \item Study 2: $F \sim \mathcal{N}(1, 1)$, $G \sim \mathcal{N}(0.6, 1)$, and $\Sigma_\rho$ follows an AR(1) dependence structure, i.e., $\sigma_{i,j} = \rho^{|i-j|}$;

    \item Study 3: $F \sim \text{Lognormal}(1, 1)$, $G \sim \text{Lognormal}(0.6, 1)$, and $\Sigma_\rho$ follows a compound symmetry structure with $\sigma_{i,j} = \rho$ for $i \neq j$;

    \item Study 4: $F \sim \text{Lognormal}(1, 1)$, $G \sim \text{Lognormal}(0.6, 1)$, and $\Sigma_\rho$ follows an AR(1) dependence structure, i.e., $\sigma_{i,j} = \rho^{|i-j|}$.
\end{itemize}
Note that the compound symmetry structure in Studies 1 and 3 does not satisfy the weak dependence assumption. Moreover, due to the strong skewness of the lognormal distribution, the true distributions of the test statistics are not well approximated by the $t$ or normal distributions.

Furthermore, in the real gene expression data example, we observe that the empirical distributions of many genes exhibit bimodal patterns, suggesting the possibility of an underlying mixture distribution. We conjecture that this may arise from various practical factors. For instance, the data may include both male and female patients, leading to differential gene expression across subgroups, or the data may be contaminated due to misclassification, where some diseased subjects are labeled as healthy and vice versa. To mimic such scenarios, we also consider the following settings.
\begin{itemize}
\item Study 5: For each $i = 1,\ldots, n_0$, we generate $U_i$ from a mixture model:
\begin{eqnarray*}
    U_i = \delta_i U_i^{(1)} + (1 - \delta_i) U_i^{(2)},
\end{eqnarray*}
where $\delta_i$, $U_i^{(1)}$, and $U_i^{(2)}$ are mutually independent, with $\delta_i \sim \text{Bernoulli}(0.7)$. The components $U_i^{(1)}$ and $U_i^{(2)}$ are generated using the same normal copula approach as in Studies~1--4. Specifically, we take $\Sigma_\rho$ to follow an AR(1) structure, i.e., $\sigma_{i,j} = \rho^{|i-j|}$, and generate $U_i^{(1)}$ such that marginally
\begin{eqnarray*}
    U_{i,j}^{(1)} \sim \text{Lognormal}(1, 1), \quad j = 1,\ldots, p,
\end{eqnarray*}
and similarly generate $U_i^{(2)}$ such that
\begin{eqnarray*}
    U_{i,j}^{(2)} \sim \text{Lognormal}(1.5, 1), \quad j = 1,\ldots, p.
\end{eqnarray*}
As a result, $U_{i,1}, \ldots, U_{i,p}$ are dependent, and each marginal follows the mixture distribution
\begin{eqnarray*}
    U_{i,j} \sim 0.7 \cdot \text{Lognormal}(1, 1) + 0.3 \cdot \text{Lognormal}(1.5, 1).
\end{eqnarray*}

The data $V_i$ are generated analogously. For $j \in \mathH_0$,
\begin{eqnarray*}
    V_{i,j} \sim 0.9 \cdot \text{Lognormal}(1, 1) + 0.1 \cdot \text{Lognormal}(1.5, 1),
\end{eqnarray*}
and for $j \in \mathH_1$,
\begin{eqnarray*}
    V_{i,j} \sim 0.9 \cdot \text{Lognormal}(0.55, 1) + 0.1 \cdot \text{Lognormal}(1.5, 0.6).
\end{eqnarray*}

\item Study 6: We adopt the same setup as in Study~5, except that the marginal distributions are given by
\begin{eqnarray*}
    U_{i,j} &\sim& 0.9 \cdot \text{Lognormal}(1, 1) + 0.1 \cdot \text{Lognormal}(1.5, 1), \\
    V_{i,j} &\sim& 0.7 \cdot \text{Lognormal}(1, 1) + 0.3 \cdot \text{Lognormal}(1.5, 1), \quad \text{for } j \in \mathH_0, \\
    V_{i,j} &\sim& 0.7 \cdot \text{Lognormal}(0.55, 1) + 0.3 \cdot \text{Lognormal}(1.5, 0.6), \quad \text{for } j \in \mathH_1.
\end{eqnarray*}

\end{itemize}

For $\rho$, we consider the values $\rho = 0.2, 0.4, 0.6, 0.8$, where larger values of $\rho$ correspond to stronger dependence in the data. For each setting, the simulation is repeated 500 times.

\subsection{Data processing} \label{Sim-Data-Processing}
Based on the simulated data $\{U_i\}_{i=1}^{n_0}$ and $\{V_i\}_{i=1}^{n_1}$, we construct the $K$-dimensional test statistics $Z_1, \ldots, Z_p$ for our method. Throughout, we set $K = 5$. Let $\widetilde n_0 = n_0 / K$, and randomly partition $\{U_i\}_{i=1}^{n_0}$ into $K$ equally sized subsets:
\begin{eqnarray*}
    \{U_i\}_{i=1}^{n_0} = \{U_i^{(1)}\}_{i=1}^{\widetilde n_0} \cup \ldots \cup \{U_i^{(K)}\}_{i=1}^{\widetilde n_0}.
\end{eqnarray*}
Similarly, let $\widetilde n_1 = n_1 / K$ and partition $\{V_i\}_{i=1}^{n_1}$ into
\begin{eqnarray*}
    \{V_i\}_{i=1}^{n_1} = \{V_i^{(1)}\}_{i=1}^{\widetilde n_1} \cup \ldots \cup \{V_i^{(K)}\}_{i=1}^{\widetilde n_1}.
\end{eqnarray*}
For each dimension $k = 1,\ldots, K$, we compute the test statistics
\begin{eqnarray*}
    Z_{j,k} = \frac{\bar V_{\cdot, j}^{(k)} - \bar U_{\cdot, j}^{(k)}}{\sqrt{\frac{S_{U^{(k)},j}^2}{\widetilde n_0} + \frac{S_{V^{(k)},j}^2}{\widetilde n_1}}}, \quad j=1,\ldots,p,
\end{eqnarray*}
where $S_{U^{(k)},j}^2$ and $S_{V^{(k)},j}^2$ denote the sample variances of $\{U_{i,j}^{(k)}\}_{i=1}^{\widetilde n_0}$ and $\{V_{i,j}^{(k)}\}_{i=1}^{\widetilde n_1}$, respectively. The resulting $Z_1, \ldots, Z_p$ are then used as the input for our method.

We compare the eFDP method (denoted Our) with three alternative methods: (1) the method of \citet{fan2012estimating, FanHan2017FDP} (denoted PFA); (2) the method of \citet{DuGuoSunZou2023SDA} (denoted SDA); and (3) the Benjamini--Hochberg procedure \citep{benjamini1995controlling} (denoted BH). For PFA and SDA, we use the raw data $\{U_i\}_{i=1}^{n_0}$ and $\{V_i\}_{i=1}^{n_1}$ as input to their respective implementations. For the BH procedure, we compute standard two-sample $t$-statistics to construct the test statistics.

\subsection{Comparison of FDP control}

We compare the performance of different methods based on the FDPs and the true discovery proportions (TDPs), where 
\begin{eqnarray*}
    \text{FDP} = \frac{V}{R} \qquad \text{TDP} = \frac{R-V}{p_1}. 
\end{eqnarray*}
The results for Studies 1--4 are displayed in four panels of Figure \ref{Figure-Examples-1--4}.
From the figure, we observe that our method achieves accurate and robust control of the FDP across all considered settings. The averaged FDPs over 500 repetitions are consistently close to the nominal level $\alpha = 0.2$, while the TDPs are the largest among almost all competing methods. This performance is expected, as our method is designed to accommodate a wide range of dependence structures among the test statistics and does not rely on correct specification of the null distribution. Consequently, it remains stable under both strong dependence and distributional misspecification.

The performance of the PFA method depends on the underlying data-generating mechanism. In Study~1 (top-left panel of Figure \ref{Figure-Examples-1--4}), the averaged FDP of the PFA method increases from below the nominal level to above $\alpha = 0.2$ as the dependence parameter $\rho$ increases from 0.2 to 0.8. Although its TDP improves with increasing dependence, this gain in power is achieved at the expense of inflated FDP, and the resulting TDP remains lower than that of our method. In contrast, in Study~2 (top-right panel of Figure \ref{Figure-Examples-1--4}), the PFA method performs well, with FDP control and TDP comparable to those of our method. This behavior is likely due to the fact that the data in this study follow an AR(1) dependence structure with normally distributed marginal distributions, which closely match the assumptions by the PFA method. In Study~4 (bottom-right panel of Figure \ref{Figure-Examples-1--4}), however, the PFA method underperforms relative to our method: although its FDP remains below the nominal level, its TDP is  smaller. This loss of power is likely due to the fact that the data are generated as lognormal distributions, which results in the constructed test statistics not well approximated by normal or $t$-distributions, thereby violating the method’s distributional assumptions.

The SDA method exhibits mixed performance across the examples. In Study~1, it shows reasonable behavior: although its FDP is conservative, its TDP remains competitive and close to those of our method and the PFA method. In Study~2 (top-right panel of Figure \ref{Figure-Examples-1--4}), the SDA method also performs comparably to our method, the PFA method, and the BH procedure, achieving similar FDP and TDP values. However, its performance deteriorates in Studies~3 and~4. In Study~3, its FDP increases and exceeds the nominal level as $\rho$ increases from 0.2 to 0.8, and is consistently higher than that of the PFA method; meanwhile, its TDP is substantially lower than those of both our method and the PFA method. In Study~4, although the FDP is only slightly above $\alpha = 0.2$, the TDP is significantly smaller than those of all other methods. These results indicate that, in these settings, the apparent FDP control of the SDA method is achieved at the expense of power. A possible explanation is that the SDA method relies on a central limit theorem approximation at a key step, which may be inaccurate when the data are generated from heavily skewed distributions such as the lognormal distribution.

For the BH procedure, its performance in Studies~2 and~4 is reasonable. In these studies, the data are generated from an AR(1) dependence structure, which satisfies the weak dependence conditions under which the BH method is known to control the false discovery rate; see, for example, \citet{benjamini2001control} and related subsequent work. In contrast, in Studies~1 and~3, which are based on compound symmetry dependence, the BH method tends to produce conservative FDP control and lower TDP compared to both the PFA method and our method. This behavior may be partly explained by the fact that, although compound symmetry with nonnegative correlation can satisfy PRDS in the one-sided testing setting, such guarantees generally do not extend to two-sided tests, where the PRDS condition may fail even under nonnegative dependence; see \citet{fithian2022conditional}. As a result, the theoretical conditions under which the BH procedure achieves exact FDR control may not be fully satisfied, leading to more conservative behavior and reduced power in practice.


The results for Studies~5--6 are presented in two panels of Figure~\ref{Figure-Examples-5--6}. 
These two studies pose substantial challenges for FDP control, as the test statistics are not only dependent, 
but also deviate markedly from the theoretical $t$ and standard normal distributions due to the presence of mixture structures. 
Such distributional contamination reflects realistic gene expression settings and undermines the validity of standard asymptotic approximations. In particular, both the variance and the center of the marginally distributions of the resultant test statistics are deviate from the theoretical distribution, i.e., $t$ or standard normal distribution. 

In both studies, our method continues to achieve accurate and stable FDP control across all values of the dependence parameter $\rho$. 
The averaged FDPs remain consistently close to the nominal level $\alpha = 0.2$, demonstrating the robustness of our method 
in the presence of both strong dependence and severe distributional misspecification, as discussed in Studies~1--4. 
Moreover, our method maintains competitive power in these challenging settings.

\begin{figure}[!htb]
    \centering
    \begin{minipage}{0.49\linewidth}
        \centering
        \includegraphics[width=\linewidth]{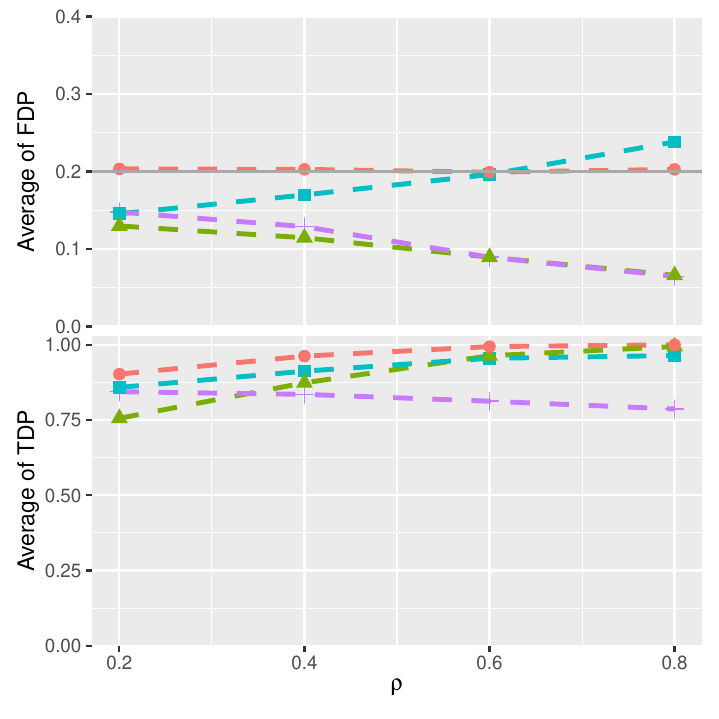}
    \end{minipage}\hfill
    \begin{minipage}{0.49\linewidth}
        \centering
        \includegraphics[width=\linewidth]{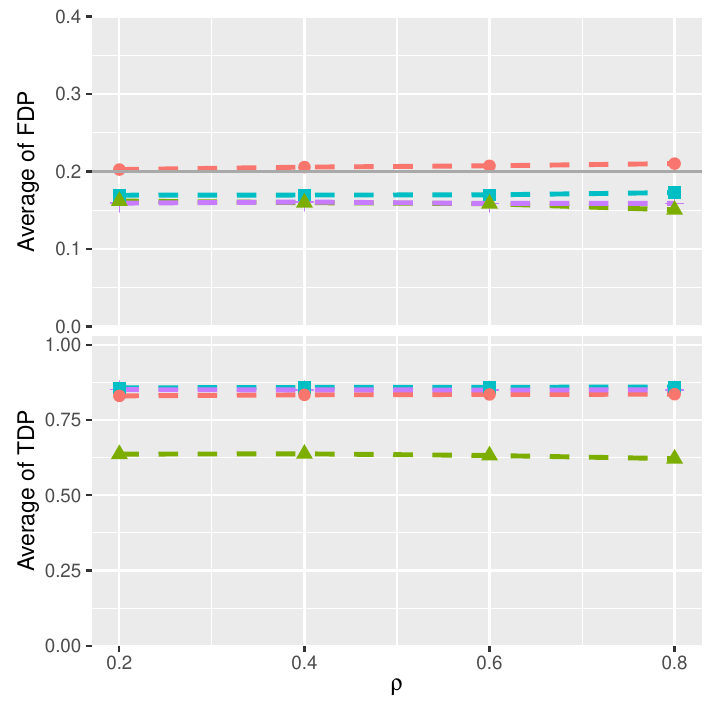}
    \end{minipage}

    \vspace{0.3em}

    \begin{minipage}{0.49\linewidth}
        \centering
        \includegraphics[width=\linewidth]{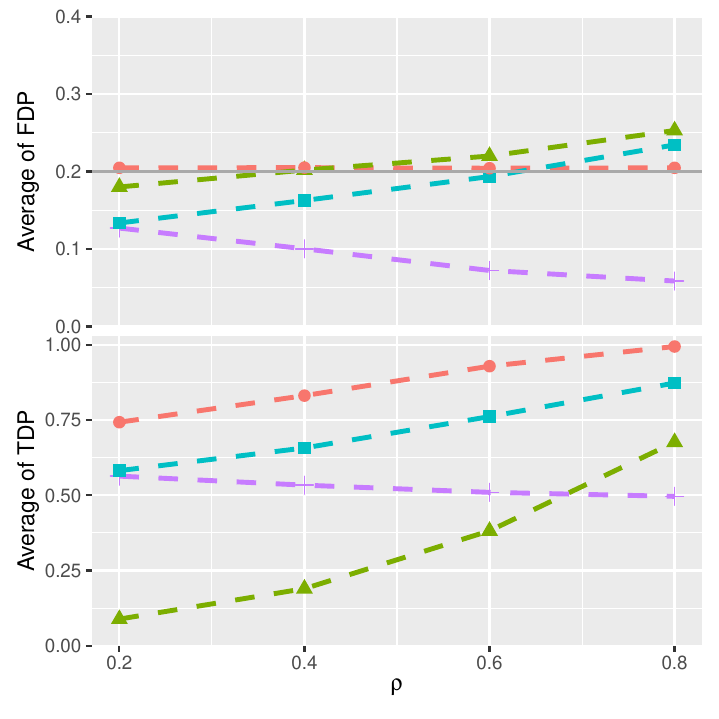}
    \end{minipage}\hfill
    \begin{minipage}{0.49\linewidth}
        \centering
        \includegraphics[width=\linewidth]{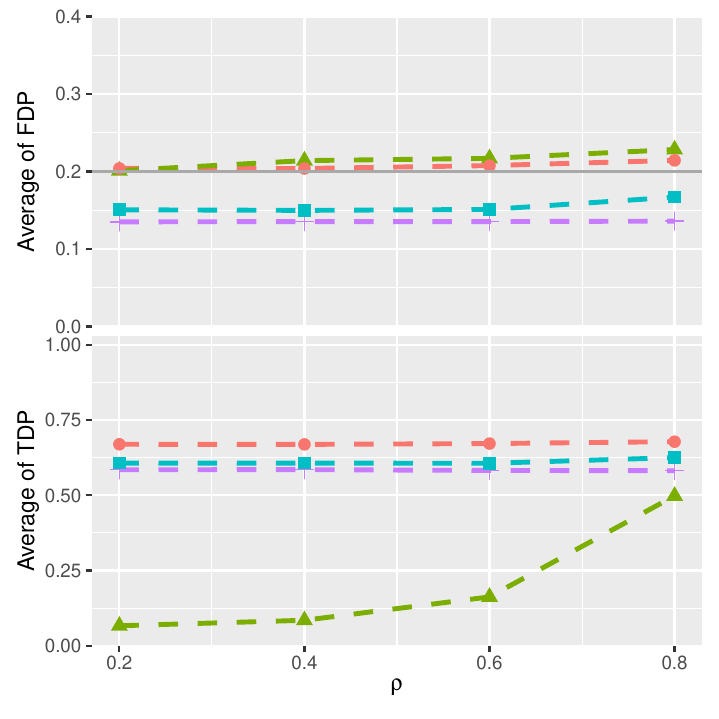}
    \end{minipage}

    \vspace{0.4em}

    \includegraphics[width=0.7\linewidth]{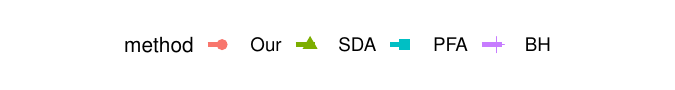}
    \vspace{0.3em}
    \caption{Results of Studies 1--4: top left panel: Study 1; top right panel: Study 2; bottom left panel: Study 3; bottom right panel: Study 4.}
    \label{Figure-Examples-1--4}
\end{figure}

\begin{figure}[!htb] 
    \centering
    \begin{minipage}{0.48\linewidth}
        \centering
        \includegraphics[
            width=\linewidth,
            height=1.25\linewidth
        ]{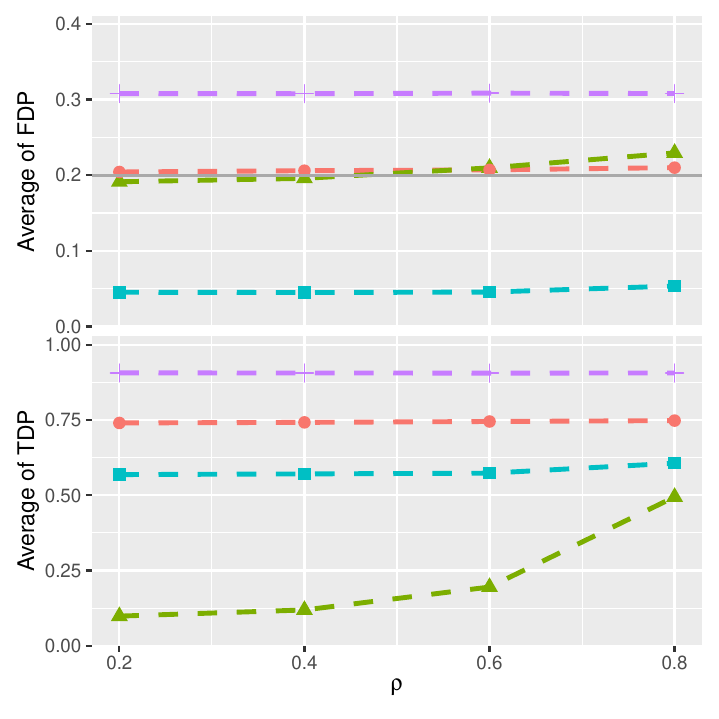}
    \end{minipage}\hfill
    \begin{minipage}{0.48\linewidth}
        \centering
        \includegraphics[
            width=\linewidth,
            height=1.25\linewidth
        ]{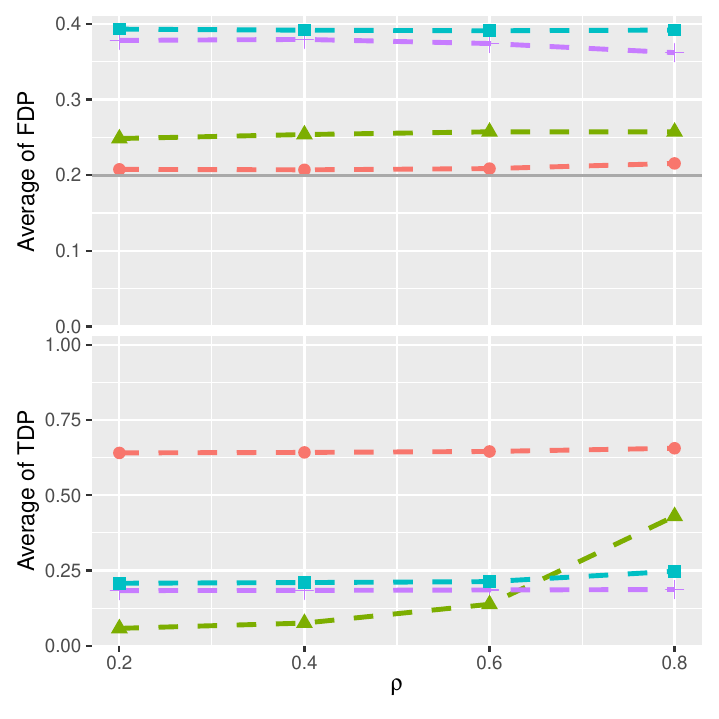}
    \end{minipage}
     \vspace{0.4em}

    \includegraphics[width=0.7\linewidth]{Paper_Figures/legend_only_v5.pdf}
    \vspace{0.3em}
  \caption{Results of Studies 5--6: left panel: Study 5;  right panel: Study 6. } \label{Figure-Examples-5--6}
\end{figure}

In Study~5, the PFA method is overly conservative, with FDPs substantially below the nominal level, 
leading to a noticeable loss of power. 
The BH procedure fails to control the FDP and exhibits inflated FDP values. 
Although the BH method achieves relatively high TDP, this gain in power is attained at the expense of liberal FDP control. 
The SDA method shows mixed behavior: while its FDP is close to the nominal level for small to moderate dependence, 
it becomes slightly inflated under strong dependence, and its TDP is substantially lower than that of our method, 
particularly when the dependence is weak.

The discrepancies among methods become more pronounced in Study~6. 
While our method continues to provide reliable FDP control with stable power, 
both the PFA and BH procedures exhibit substantially inflated FDPs across all dependence settings, 
indicating a clear failure of error control under this more severe mixture contamination. 
The SDA method again shows unstable performance, with FDPs exceeding the nominal level and TDPs remaining low 
for weak to moderate dependence, improving only under very strong dependence. 
Overall, none of the competing methods achieves a satisfactory balance between FDP control and power in this example.

Taken together, the results in Studies~5--6 highlight the intrinsic difficulty of FDP control in the presence of both dependence 
and mixture-induced distributional heterogeneity. 
In these practically motivated scenarios, our method stands out by maintaining accurate FDP control 
while achieving consistently higher power than existing alternatives.
 
Additional simulation examples based on one-sided hypothesis tests are presented in Section~3 of the supplementary material.

\section{Real Data Analysis} \label{section-real-data}
The dataset GSE25066 is a GEO SuperSeries containing genome-wide expression profiles measured on the Affymetrix Human Genome U133A platform for a total of 508 pretreatment breast-tumor biopsy samples obtained from patients receiving neoadjuvant chemotherapy. Among these samples, 99 correspond to patients who achieved a pathologic complete response (pCR), 389 correspond to those with residual disease (RD), and 20 samples lack response annotation. We retrieved the dataset using the \texttt{GEOquery} package in \textsf{R}. For the purpose of identifying differentially expressed genes, we removed samples with missing response information and thus the data of interest is composed of  pCR Group with $n_1=99$ samples and the RD group with $n_2=389$ samples. 

The data set measures the expressions of 22283 genes. We first take a variance-filtering step to conduct a preliminary screening. Specifically, we computed the variance of each gene across all samples, ranked the genes in descending order of variance, and retained the top $5\%$ (i.e., 1,114 genes) for subsequent analysis. Genes exhibiting higher overall variance are more likely to capture biologically meaningful differences between groups.

To apply our method, we follow the data processing procedure described in Section~\ref{Sim-Data-Processing} with $K=5$ and construct test statistics $Z_1,\ldots,Z_{1114}$. Since the underlying differential expression signals may be either positive or negative, we consider two-sided tests and compute the corresponding $p$-values by treating the test statistics as Welch $t$-statistics. These $p$-values are then used as the input to our algorithm.

The numbers of identified genes across different methods and nominal FDR levels are reported in Table~\ref{Table-real-data-detect}. As seen from the table, the competing methods exhibit markedly different discovery behaviors. The BH procedure identifies a substantially larger number of genes across all values of $\alpha$, while the PFA and SDA methods are considerably more conservative, with the SDA method yielding the smallest number of discoveries at lower FDR levels. Our method produces a moderate number of discoveries, lying between the BH procedure and the more conservative PFA and SDA methods.

\begin{table}[htbp]
\centering
\caption{Number of discoveries based on different $\alpha$; real data}
\vspace{0.15in}
\begin{tabular}{@{}lcccc@{}}
\toprule
& Our & PFA & SDA & BH \\
\midrule
$\alpha=0.05$ & 138 & 73  & 20 & 442 \\
$\alpha=0.10$ & 192 & 93  & 29 & 487 \\
$\alpha=0.20$ & 269 & 138 & 72 & 582 \\
\bottomrule
\end{tabular} \label{Table-real-data-detect}
\end{table}

To further assess the plausibility of these results, we estimate the proportion of significant genes among the 1,114 screened genes. The estimated proportions of the significant genes based on our method and the PFA method are 0.223 and 0.387, respectively. Combined with the simulation findings—particularly those from Example~5, which was designed to mimic realistic gene expression settings involving dependence and mixture contamination—these estimates suggest that our method may have achieved a reasonable balance between false discovery control and detection power in this real data analysis. Although the true set of differentially expressed genes is unknown, the moderate discovery rate of our method appears consistent with stable FDP control and nontrivial power.

Additional insight is provided by examining the overlap of identified genes across methods at $\alpha=0.2$, as summarized in Figure~\ref{Upset-alpha-0.2-final}. Nearly all genes identified by the PFA method are also detected by our method, indicating that our method captures the conservative discoveries made by PFA. A substantial proportion of the genes identified by the SDA method also overlap with our discoveries, though the SDA method identifies far fewer genes overall. In contrast, while the BH procedure identifies many more genes, only a subset of these overlap with those detected by our method, suggesting that the additional discoveries made by BH may include a larger number of false positives.

Taken together, these results reinforce the conclusions drawn from the simulation studies. In complex gene expression data characterized by dependence and distributional heterogeneity, our method appears to provide a stable and interpretable set of discoveries, avoiding the excessive conservativeness of PFA and SDA while mitigating the liberal behavior of the BH procedure.



\begin{figure}[htbp]
    \centering
    \includegraphics[width=0.7\linewidth]{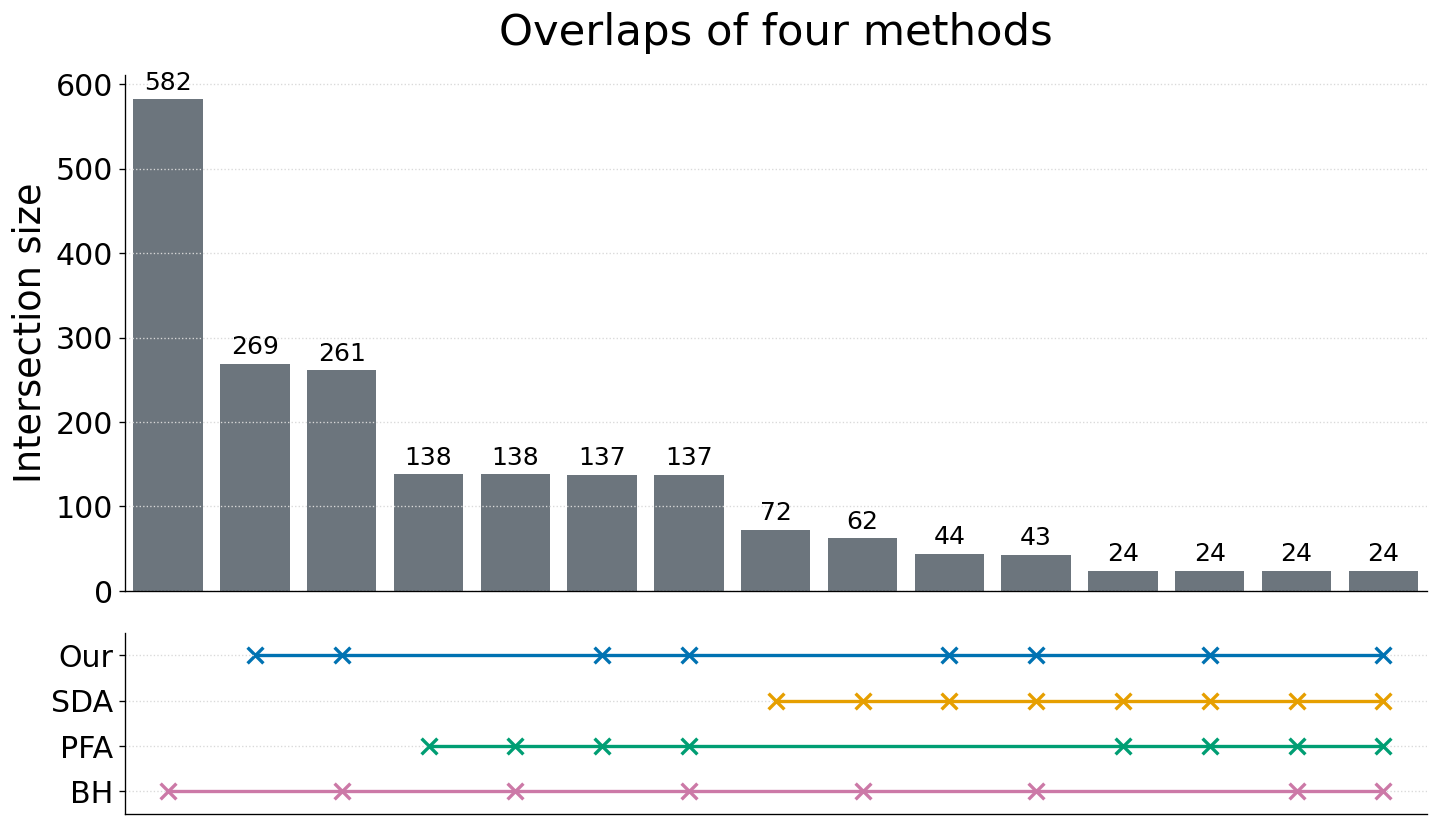}
    \caption{UpSet plot showing overlaps among genes identified by the four methods; real data}
    \label{Upset-alpha-0.2-final}

\end{figure}

\section{Discussion} \label{section-discussion}

In this paper, we have focused on the role of the empirical null distribution in controlling the FDP in large-scale multiple testing problems. Our work highlights a shift in perspective: in large-scale testing, accurate estimation of the empirical distribution of null test statistics may be more consequential for FDP control than explicitly modeling dependence among tests. Our contributions are as follows. First, we show that under an oracle scenario where the empirical null is known, FDP control can be achieved with optimal efficiency regardless of the dependence structure among the test statistics, suggesting that modeling dependence, while common in the literature, may be less critical than accurate empirical null estimation. Second, we establish nonparametric estimation of the empirical c.d.f. of null test statistics within a multivariate mixture model framework, together with a proof of asymptotic convergence. Finally, we establish an eFDP procedure that asymptotically controls the FDP and demonstrates robust performance under a broad class of dependence structures.

The simulation studies and real-data analysis illustrate that our eFDP method achieves a favorable balance between false discovery control and detection power. In particular, our approach remains reliable under strong dependence structures such as compound symmetry, where methods explicitly modeling dependence may fail. The real data example further confirms that the empirical null-based procedure identifies a stable and interpretable set of discoveries, capturing conservative findings while avoiding overly liberal inferences.

Beyond the immediate application to FDP control, our work is also related to the literature on statistical inference for mixture models. One central focus in the literature is the estimation of component distributions and mixing proportions under identifiability conditions; see, for example, \citet{Hall2003, Hall2005, Allman2009, levine2011, YuQinLi2026}. Much of the existing work in the mixture model context assumes independence of the observations; see, for example, \citet{Benaglia2009, levine2011, Yu2019, Zheng2020, YuQinLi2026}. While these approaches have been successful in many settings, the problem of estimating empirical distribution functions of mixture components—particularly in nonparametric settings and under dependence—has received relatively less attention. By directly targeting the empirical c.d.f. without imposing independence, our framework provides a complementary perspective that may be useful for extending mixture model methodology to dependent data settings.

Several directions for future research arise naturally from our study. In this paper, we aggregated the empirical p-values across dimensions by computing their median; however, alternative aggregation strategies could be explored, and their theoretical and numerical properties warrant investigation. In addition, while we have established the consistency of the empirical c.d.f. estimators, we have not characterized the asymptotic variance or distribution of the FDP. Quantifying the variability of the FDP is important for providing more precise inferential statements; see \citet{mei2024asymptotic} and the references therein. Moreover, extending the theoretical development to enable more sophisticated inference based on the estimated empirical c.d.f.s represents a promising direction for future work, which would not only enhance the practical and theoretical utility of empirical null-based procedures in large-scale multiple testing but also broaden their applicability for inference in mixture models.

\section*{Declaration of competing interest}
The authors declare that they have no competing interests.

\section*{Acknowledgments}
Dr. Li's work is supported in part by the Natural Sciences and Engineering Research Council of Canada (RGPIN-2020-04964) and the Faculty of Mathematics Research Chair Funding of the
University of Waterloo. Dr. Jiang's work is supported in part by the U.S. National Science Foundation (Award Number: 2515805). Dr. Yu's work is supported in part by the Singapore Ministry of Education Academic Research Tier 1 Fund: A-8000413-00-00.

\bibliographystyle{abbrvnat}
\bibliography{ref}

@article{jin2007estimating,
  author  = {Jin, Jiashun and Cai, T. Tony},
  title   = {Estimating the Null and the Proportion of Nonnull Effects in Large-Scale Multiple Comparisons},
  journal = {Journal of the American Statistical Association},
  year    = {2007},
  volume  = {102},
  pages   = {495--506},
  doi     = {10.1198/016214507000000167}
}

@article{wang2010slim,
  title={SLIM: a sliding linear model for estimating the proportion of true null hypotheses in datasets with dependence structures},
  author={Wang, Hong-Qiang and Tuominen, Lindsey K and Tsai, Chung-Jui},
  journal={Bioinformatics},
  volume={27},
  pages={225--231},
  year={2010},
  publisher={Oxford University Press}
}

@article{langaas2005estimating,
  title={Estimating the proportion of true null hypotheses, with application to DNA microarray data},
  author={Langaas, Mette and Lindqvist, Bo Henry and Ferkingstad, Egil},
  journal={Journal of the Royal Statistical Society: Series B (Statistical Methodology)},
  volume={67},
  pages={555--572},
  year={2005},
  publisher={Wiley Online Library}
}

@article{storey2003statistical,
  title={Statistical significance for genomewide studies},
  author={Storey, John D and Tibshirani, Robert},
  journal={Proceedings of the National Academy of Sciences},
  volume={100},
  pages={9440--9445},
  year={2003},
  publisher={National Acad Sciences}
}

@article{Bradley2005,
  author = {Bradley, Richard C.},
  title = {Basic Properties of Strong Mixing Conditions. A Survey and Some Open Questions},
  journal = {Probability Surveys},
  volume = {2},
  pages = {107--144},
  year = {2005}
}

@book{Doukhan1994,
  author = {Doukhan, Paul},
  title = {Mixing: Properties and Examples},
  publisher = {Springer},
    address   = {New York},
  year = {1994}
}

@book{FanYao2003,
  author = {Fan, Jianqing and Yao, Qiwei},
  title = {Nonlinear Time Series: Nonparametric and Parametric Methods},
  publisher = {Springer},
   address   = {New York},
  year = {2003}
}

@article{Zheng2020,
	author = {Chaowen Zheng and Yichao Wu},
	journal = {Journal of the American Statistical Association},
	pages = {1456-1471},
	title = {Nonparametric estimation of multivariate mixtures},
	volume = {115},
	year = {2020}}

@article{levine2011,
	author = {M. Levine and D. R. Hunter and D. Chauveau},
	journal = {Biometrika},
	pages = {403-416},
	title = {Maximum smoothed likelihood for multivariate mixtures},
	volume = {98},
	year = {2011}}

@article{Benaglia2009,
	author = {Tatiana Benaglia and Didier Chauveau and David R. Hunter},
	journal = {Journal of Computational and Graphical Statistics},
	pages = {505-526},
	title = {An {EM}-like algorithm for semi- and nonparametric estimation in multivariate mixtures},
	volume = {18},
	year = {2009}}

@article{Hall2003,
	author = {Peter Hall and Xiao-Hua Zhou},
	journal = {The Annals of Statistics},
	pages = {201-224},
	title = {{Nonparametric estimation of component distributions in a multivariate mixture}},
	volume = {31},
	year = {2003}}

@article{Hall2005,
	author = {Hall, Peter and Neeman, Amnon and Pakyari, Reza and Elmore, Ryan},
	journal = {Biometrika},
	pages = {667-678},
	title = {{Nonparametric inference in multivariate mixtures}},
	volume = {92},
	year = {2005}}

@article{Allman2009,
	author = {Elizabeth S. Allman and Catherine Matias and John A. Rhodes},
	journal = {The Annals of Statistics},
	pages = {3099-3132},
	title = {Identifiability of parameters in latent structure models with many observed variables},
	volume = {37},
	year = {2009}}

@article{Yu2019,
	author = {Tao Yu and Pengfei Li and Jing Qin},
	journal = {Electronic Journal of Statistics},
	pages = {4035-4078},
	title = {{Maximum smoothed likelihood component density estimation in mixture models with known mixing proportions}},
	volume = {13},
	year = {2019}}

@article{StoreyTaylorSiegmund2004,
  author       = {Storey, John D. and Taylor, Jonathan E. and Siegmund, David},
  title        = {Strong control, conservative point estimation and simultaneous conservative consistency of false discovery rates: A unified approach},
  journal      = {Journal of the Royal Statistical Society: Series B (Statistical Methodology)},
  volume       = {66},
  number       = {1},
  pages        = {187--205},
  year         = {2004},
  doi          = {10.1111/j.1467-9868.2004.00439.x},
  url          = {https://doi.org/10.1111/j.1467-9868.2004.00439.x},
}

@article{de2009,
	author = {de Leeuw, Jan and Hornik, Kurt and Mair, Patrick},
	journal = {Journal of Statistical Software},
	pages = {1-24},
	title = {Isotone optimization in R: Pool-adjacent-violators algorithm (PAVA) and active set methods},
	volume = {32},
	year = {2009}}

@article{Ayer1955,
	author = {Miriam Ayer and H. D. Brunk and G. M. Ewing and W. T. Reid and Edward Silverman},
	journal = {The Annals of Mathematical Statistics},
	pages = {641-647},
	title = {{An empirical distribution function for sampling with incomplete information}},
	volume = {26},
	year = {1955}}

@article{Kwonsang2023,
	author = {Kwonsang Lee and Bhattacharya, {Bhaswar B.} and Jing Qin and Small, {Dylan S.}},
	journal = {Journal of the Korean Statistical Society},
	pages = {1055-1077},
	title = {A nonparametric binomial likelihood approach for causal inference in instrumental variable models},
	volume = {52},
	year = {2023}}

@article{Cristiano2011,
	author = {Cristiano Varin and Nancy Reid and David Firth},
	journal = {Statistica Sinica},
	pages = {5-42},
	title = {An overview of composite likelihood methods},
	volume = {21},
	year = {2011}}

@article{YuQinLi2026,
  author  = {Yu, Tao and Qin, Jing and Li, Pengfei},
  title   = {Maximum binomial likelihood method for multivariate mixture data},
  journal = {Journal of the American Statistical Association},
  year    = {2026},
  note    = {Advance online publication},
  doi     = {10.1080/01621459.2026.2624135}
}

@article{Qin2014,
	author = {Qin, Jing and Garcia, Tanya and Ma, Yanyuan and Tang, Ming-Xin and Marder, Karen and Wang, Yuanjia},
	journal = {The Annals of Applied Statistics},
	pages = {1182-1208},
	title = {Combining isotonic regression and EM algorithm to predict genetic risk under monotonicity constraint},
	volume = {8},
	year = {2014}}

@article{YU2023,
	author = {Tao Yu and Pengfei Li and Baojiang Chen and Ao Yuan and Jing Qin},
	journal = {Journal of Econometrics},
	pages = {454-469},
	title = {Maximum pairwise-rank-likelihood-based inference for the semiparametric transformation model},
	volume = {235},
	year = {2023}}

@article{BenjaminiHochberg1995,
  author    = {Yoav Benjamini and Yosef Hochberg},
  title     = {Controlling the false discovery rate: a practical and powerful approach to multiple testing},
  journal   = {Journal of the Royal Statistical Society: Series B (Methodological)},
  volume    = {57},
  pages     = {289--300},
  year      = {1995},
  doi       = {10.1111/j.2517-6161.1995.tb02031.x}
}

@article{DuGuoSunZou2023SDA,
  author    = {Du, Lilun and Guo, Xu and Sun, Wenguang and Zou, Changliang},
  title     = {False Discovery Rate Control Under General Dependence By Symmetrized Data Aggregation},
  journal   = {Journal of the American Statistical Association},
  volume    = {118},
  pages     = {607--621},
  year      = {2023},
  doi       = {10.1080/01621459.2021.1945459}
}

@article{FanHan2017FDP,
  author  = {Fan, Jianqing and Han, Xu},
  title   = {Estimation of the False Discovery Proportion with Unknown Dependence},
  journal = {Journal of the Royal Statistical Society: Series B (Statistical Methodology)},
  volume  = {79},
  pages   = {1143--1164},
  year    = {2017},
  doi     = {10.1111/rssb.12204}
}

@article{Mokkadem1988,
  author       = {Abdelkader Mokkadem},
  title        = {Mixing properties of {ARMA} processes},
  journal      = {Stochastic Processes and their Applications},
  volume       = {29},
  number       = {2},
  pages        = {309--315},
  year         = {1988},
  month        = sep,
  doi          = {10.1016/0304-4149(88)90045-2},
}

@incollection{MerlevedePeligradRio2009,
  author       = {Merlev\`ede, Florence and Peligrad, Magda and Rio, Emmanuel},
  title        = {Bernstein inequality and moderate deviations under strong mixing conditions},
  booktitle    = {High Dimensional Probability V: The Luminy Volume},
  editor       = {Lugosi, G\'abor and Talagrand, Michel},
  publisher    = {Institute of Mathematical Statistics},
  year         = {2009},
  pages        = {273--292},
  doi          = {10.1214/09-IMSCOLL518},
}

@article{Naaman2021,
  author       = {Naaman, Michael},
  title        = {On the tight constant in the multivariate Dvoretzky–Kiefer–Wolfowitz inequality},
  journal      = {Statistics \& Probability Letters},
  volume       = {173},
  year         = {2021},
  pages        = {109088},
  doi          = {10.1016/j.spl.2021.109088},
  publisher    = {Elsevier}
}

@article{zhang2011multiple,
  title={Multiple testing via {${\rm FDR}_L$} for large-scale imaging data},
  author={Zhang, Chi and Fan, Jianqing and Yu, Tianwei},
  journal={The Annals of Statistics},
  volume={39},
  pages={613--642},
  year={2011},
  publisher={Institute of Mathematical Statistics},
  doi={10.1214/10-AOS848}
}

@article{fan2012estimating,
  title={Estimating false discovery proportion under arbitrary covariance dependence},
  author={Fan, Jianqing and Han, Xu and Gu, Weijie},
  journal={Journal of the American Statistical Association},
  volume={107},
  pages={1019--1035},
  year={2012},
  publisher={Taylor \& Francis}
}

@article{bonferroni1936teoria,
  title={Teoria statistica delle classi e calcolo delle probabilita},
  author={Bonferroni, Carlo},
  journal={Pubblicazioni del R Istituto Superiore di Scienze Economiche e Commericiali di Firenze},
  volume={8},
  pages={3--62},
  year={1936}
}

@article{vsidak1967rectangular,
  title={Rectangular confidence regions for the means of multivariate normal distributions},
  author={{\v{S}}id{\'a}k, Zbyn{\v{e}}k},
  journal={Journal of the American Statistical Association},
  volume={62},
  pages={626--633},
  year={1967},
  publisher={Taylor \& Francis}
}

@article{holm1979simple,
  title={A simple sequentially rejective multiple test procedure},
  author={Holm, Sture},
  journal={Scandinavian Journal of Statistics},
  volume={6},
  pages={65--70},
  year={1979},
  publisher={JSTOR}
}

@article{holland1987improved,
  title={An improved sequentially rejective Bonferroni test procedure},
  author={Holland, Burt S and Copenhaver, Margaret DiPonzio},
  journal={Biometrics},
  volume={43},
  pages={417--423},
  year={1987},
  publisher={JSTOR}
}

@article{simes1986improved,
  title={An improved Bonferroni procedure for multiple tests of significance},
  author={Simes, R John},
  journal={Biometrika},
  volume={73},
  pages={751--754},
  year={1986},
  publisher={Oxford University Press}
}

@article{hochberg1988sharper,
  title={A sharper Bonferroni procedure for multiple tests of significance},
  author={Hochberg, Yosef},
  journal={Biometrika},
  volume={75},
  pages={800--802},
  year={1988},
  publisher={Oxford University Press}
}

@article{rom1990sequentially,
  title={A sequentially rejective test procedure based on a modified Bonferroni inequality},
  author={Rom, Dror M},
  journal={Biometrika},
  volume={77},
  pages={663--665},
  year={1990},
  publisher={Oxford University Press}
}

@article{dudoit2004multiple,
  title={Multiple testing. Part I. Single-step procedures for control of general type I error rates},
  author={Dudoit, Sandrine and van der Laan, Mark J and Pollard, Katherine S},
  journal={Statistical Applications in Genetics and Molecular Biology},
  volume={3},
  pages={1--69},
  year={2004},
  publisher={De Gruyter}
}

@article{pollard2004choice,
  title={Choice of a null distribution in resampling-based multiple testing},
  author={Pollard, Katherine S and van der Laan, Mark J},
  journal={Journal of Statistical Planning and Inference},
  volume={125},
  pages={85--100},
  year={2004},
  publisher={Elsevier}
}

@article{lehmann2012generalizations,
  title={Generalizations of the familywise error rate},
  author={Lehmann, Erich L. and Romano, Joseph P.},
  journal={The Annals of Statistics},
  volume={33},
  pages={1138--1154},
  year={2005},
  publisher={Institute of Mathematical Statistics}
}

@article{benjamini1995controlling,
  title={Controlling the false discovery rate: a practical and powerful approach to multiple testing},
  author={Benjamini, Yoav and Hochberg, Yosef},
  journal={Journal of the Royal Statistical Society: Series B (Methodological)},
  volume={57},
  pages={289--300},
  year={1995},
  publisher={Wiley Online Library}
}

@article{benjamini2000adaptive,
  title={On the adaptive control of the false discovery rate in multiple testing with independent statistics},
  author={Benjamini, Yoav and Hochberg, Yosef},
  journal={Journal of Educational and Behavioral Statistics},
  volume={25},
  pages={60--83},
  year={2000},
  publisher={Sage Publications Sage CA: Los Angeles, CA}
}

@article{storey2002direct,
  title={A direct approach to false discovery rates},
  author={Storey, John D},
  journal={Journal of the Royal Statistical Society: Series B (Statistical Methodology)},
  volume={64},
  pages={479--498},
  year={2002},
  publisher={Wiley Online Library}
}

@article{efron2001empirical,
  title={Empirical Bayes analysis of a microarray experiment},
  author={Efron, Bradley and Tibshirani, Robert and Storey, John D and Tusher, Virginia},
  journal={Journal of the American Statistical Association},
  volume={96},
  pages={1151--1160},
  year={2001},
  publisher={Taylor \& Francis}
}

@article{benjamini2001control,
  title={The control of the false discovery rate in multiple testing under dependency},
  author={Benjamini, Yoav and Yekutieli, Daniel},
  journal={The Annals of Statistics},
  volume={29},
  pages={1165--1188},
  year={2001},
  publisher={Institute of Mathematical Statistics}
}

@article{efron2007correlation,
  title={Correlation and large-scale simultaneous significance testing},
  author={Efron, Bradley},
  journal={Journal of the American Statistical Association},
  volume={102},
  pages={93--103},
  year={2007},
  publisher={Taylor \& Francis}
}

@article{sun2009large,
  title={Large-scale multiple testing under dependence},
  author={Sun, Wenguang and Cai, T},
  journal={Journal of the Royal Statistical Society: Series B (Statistical Methodology)},
  volume={71},
  pages={393--424},
  year={2009},
  publisher={Wiley Online Library}
}

@article{fithian2022conditional,
  title={Conditional calibration for false discovery rate control under dependence},
  author={Fithian, William and Lei, Lihua},
  journal={The Annals of Statistics},
  volume={50},
  pages={3091--3118},
  year={2022},
  publisher={Institute of Mathematical Statistics}
}

@article{tansey2018false,
  title={False discovery rate smoothing},
  author={Tansey, Wesley and Koyejo, Oluwasanmi and Poldrack, Russell A and Scott, James G},
  journal={Journal of the American Statistical Association},
  volume={113},
  pages={1156--1171},
  year={2018},
  publisher={Taylor \& Francis}
}

@article{xie2011optimal,
  title={Optimal false discovery rate control for dependent data},
  author={Xie, Jichun and Cai, T Tony and Maris, John and Li, Hongzhe},
  journal={Statistics and Its Interface},
  volume={4},
  pages={417--430},
  year={2011}
}

@article{heller2021optimal,
  title={Optimal control of false discovery criteria in the two-group model},
  author={Heller, Ruth and Rosset, Saharon},
  journal={Journal of the Royal Statistical Society:
  Series B (Statistical Methodology)},
  volume={83},
  pages={133--155},
  year={2021},
  publisher={Oxford University Press}
}

@article{friguet2009factor,
  title={A factor model approach to multiple testing under dependence},
  author={Friguet, Chlo{\'e} and Kloareg, Maela and Causeur, David},
  journal={Journal of the American Statistical Association},
  volume={104},
  pages={1406--1415},
  year={2009},
  publisher={Taylor \& Francis}
}

@article{wei2008incorporating,
  title={Incorporating gene networks into statistical tests for genomic data via a spatially correlated mixture model},
  author={Wei, Peng and Pan, Wei},
  journal={Bioinformatics},
  volume={24},
  pages={404--411},
  year={2008},
  publisher={Oxford University Press}
}

@article{schwartzman2011effect,
  title={The effect of correlation in false discovery rate estimation},
  author={Schwartzman, Armin and Lin, Xihong},
  journal={Biometrika},
  volume={98},
  pages={199--214},
  year={2011},
  publisher={Oxford University Press}
}

@article{heesen2015inequalities,
  title={Inequalities for the false discovery rate ({FDR}) under dependence},
  author={Heesen, Philipp and Janssen, Arnold},
  journal={Electronic Journal of Statistics},
  volume={9},
  pages={679--716},
  year={2015}
}

@article{mei2024asymptotic,
  title={Asymptotic uncertainty of false discovery proportion},
  author={Mei, Meng and Yu, Tao and Jiang, Yuan},
  journal={Biometrics},
  volume={80},
  pages={ujae015},
  year={2024},
  publisher={Oxford University Press}
}

@article{ghosal2011predicting,
  title={Predicting false discovery proportion under dependence},
  author={Ghosal, Subhashis and Roy, Anindya},
  journal={Journal of the American Statistical Association},
  volume={106},
  pages={1208--1218},
  year={2011},
  publisher={Taylor \& Francis}
}

@article{sun2015false,
  title={False discovery control in large-scale spatial multiple testing},
  author={Sun, Wenguang and Reich, Brian J and Tony Cai, T and Guindani, Michele and Schwartzman, Armin},
  journal={Journal of the Royal Statistical Society:
  Series B (Statistical Methodology)},
  volume={77},
  pages={59--83},
  year={2015},
  publisher={Oxford University Press}
}

@article{efron2004large,
  title={Large-scale simultaneous hypothesis testing: the choice of a null hypothesis},
  author={Efron, Bradley},
  journal={Journal of the American Statistical Association},
  volume={99},
  pages={96--104},
  year={2004},
  publisher={Taylor \& Francis}
}

@Manual{RpackagePFA,
  title = {pfa: Estimates False Discovery Proportion Under Arbitrary Covariance
Dependence},
  author = {Jianqing Fan and Tracy Ke and Sydney Li and Lucy Xia},
  year = {2016},
  note = {R package version 1.1},
  url = {https://CRAN.R-project.org/package=pfa},
}

@article{efron2007size,
  title={Size, Power and False Discovery Rates},
  author={Efron, Bradley},
  journal={The Annals of Statistics},
  pages={1351--1377},
  volume  = {35},
  year={2007},
  publisher={JSTOR}
}

@article{schwartzman2008empirical,
  title={Empirical Null and False Discovery Rate Inference for Exponential Families},
  author={Schwartzman, Armin},
  journal={The Annals of Applied Statistics},
  volume={2},
  pages={1332--1359},
  year={2008},
  publisher={JSTOR}
}

@article{park2011estimation,
  title={Estimation of empirical null using a mixture of normals and its use in local false discovery rate},
  author={Park, DoHwan and Park, Junyong and Zhong, Xiaosong and Sadelain, Michel},
  journal={Computational Statistics \& Data Analysis},
  volume={55},
  pages={2421--2432},
  year={2011},
  publisher={Elsevier}
}

@article{gauran2018empirical,
  title={Empirical null estimation using zero-inflated discrete mixture distributions and its application to protein domain data},
  author={Gauran, Iris Ivy M and Park, Junyong and Lim, Johan and Park, DoHwan and Zylstra, John and Peterson, Thomas and Kann, Maricel and Spouge, John L},
  journal={Biometrics},
  volume={74},
  pages={458--471},
  year={2018},
  publisher={Oxford University Press}
}

\end{document}


\def\spacingset#1{\renewcommand{\baselinestretch}{#1}\small\normalsize}
\maketitle

\begin{abstract}
    This supplementary material accompanies the paper ``Revisiting dependence in multiple testing: empirical distribution approaches for FDP control.'' It is organized into three sections: Section~\ref{Section-technical details} provides the technical details for the theoretical results of Sections 3 and 4 of the main article; Section~\ref{Section-Details Main 3.2} further elaborates on the developments in Section~3.2 of the main article and presents theoretical results on the convergence of the algorithm; and Section~\ref{Simulation-one-sided} includes additional simulation studies focusing on one-sided hypothesis tests.
\end{abstract}

\section{Technical Details of the Examples and Theorems in the Main Article} \label{Section-technical details}

\subsection{Review of the theoretical results and technical conditions in the main article} 

In Sections 3.3 and 3.4 of the main article, we have imposed the following regularity conditions. 

\begin{Condition}\label{Condition-1}
For $m=0, 1$, $\FF_m(\cdot)$ and $\FF_{m,\theta_0}(\cdot)$ satisfies, as $p\to \infty$,
\begin{eqnarray*}
\sup_{z\in \mathbb{R}^K}|\FF_m(z) - \FF_{m,\theta_0}(z)| = o\left\{ (\log p)^{-1} \right\}, \quad a.s.
\end{eqnarray*}
\end{Condition}

\begin{Condition}\label{Condition-2}
There exist multivariate c.d.f.s $F_0^*(z)$ and $F_1^*(z)$ (possibly random; see Example \ref{Example-2}) defined on $z\in \mathbb{R}^K$, such that:
\begin{eqnarray*}
    \sup_{z} |\FF_0(z) - F_0^*(z)| = o(1) \quad \text{and} \quad \sup_{z} |\FF_1(z) - F_1^*(z)| = o(1), \quad a.s.,
\end{eqnarray*}
as $p\to \infty$.
\end{Condition}

\begin{Condition}\label{Condition-3}
    $\lambda_{0,0} = p_0/p$ satisfies
    \begin{eqnarray*}
        \lambda_{0,0} = \lambda_0^* + o(1), \quad \text{as }p\to \infty,
    \end{eqnarray*}
    with $0<\lambda_0^*<1$ being a constant. Denote $\lambda_1^* = 1- \lambda_0^*$.
\end{Condition}

\begin{Condition}\label{Condition-4}
Consider $F_m^*(z)$ given in Condition \ref{Condition-2}, and $\lambda_m^*$ given in Condition \ref{Condition-3}, $m=0,1$. Let 
\begin{eqnarray*}
F^*(z) = \lambda_0^* F_0^*(z) + \lambda_1^* F_1^*(z), 
\end{eqnarray*}
and for $m=0,1$, denote by $F_{m,k}^*(\cdot)$ the $k$th marginal c.d.f. of $F_m^*(\cdot)$. For any $\theta_1, \theta_2 \in \Theta$, where we denote $\theta_r = \{\lambda_{0,r}, F_{0,1,r}, \ldots, F_{1,K,r}\}$, $r = 1,2$. For each $\omega$ in the sample space, if 
\begin{eqnarray*}
    \int \left\{ \FF_{\theta_1}^\omega (z) - \FF_{\theta_2}^\omega(z) \right\}^2 d F^{*\omega} (z) =0, 
\end{eqnarray*}
then $\lambda_{0,1} = \lambda_{0,2}$ and $F_{m_1, k, 1}^{\omega}(\cdot) = F_{m_1,k,2}^{\omega}(\cdot)$ almost surely in $F_{m_2,k}^{*\omega}(\cdot)$ for every $m_1, m_2 \in \{0,1\}, k=1,\ldots, K$.
\end{Condition}

\begin{Condition} \label{Condition-5}
    The empirical c.d.f.s of the empirical $p$-values $\FF_{0,j}(Y_{i,j}), i = 1, \ldots, p_1$ for Group 1, i.e., 
    \begin{eqnarray*}
        G_j(t) = \frac{1}{p_1} \sum_{i=1}^{p_1} I \left\{\FF_{0,j}(Y_{i,j}) \leq t \right\}, \quad j = 1,\ldots, K,
    \end{eqnarray*}
    satisfies, as $p\to \infty$,
    \begin{eqnarray*}
        \sup_{t\in [0,1]} |G_j(t) - G_j^*(t)| = o(1),  \quad a.s.,
    \end{eqnarray*}
    for $G_j^*(t), j=1,\ldots, K$ being  Lipschitz continuous functions. Furthermore, for the specified $\alpha > 0$, there exists a $t_\alpha^\infty > 0$, such that 
    \begin{eqnarray*}
        \frac{\lambda_0^*}{1-\lambda_0^*} \cdot \frac{1-\alpha}{\alpha} M(t_\alpha^\infty) <  G^*(t_\alpha^\infty),
    \end{eqnarray*}
    where 
    \begin{eqnarray*}
        G^*(t) = \sum\limits_{k=\left\lceil \frac{K}{2} \right\rceil }^K\sum\limits_{|S|=k} \prod_{j\in S} G_j^*(t) \prod_{j \notin S}\left\{1-G_j^*(t) \right\},
    \end{eqnarray*}
    and $S$ denotes a subset of $\{1, \ldots, K\}$, $|S|$ is the number of elements in set $S$. 
\end{Condition}

In section 3.3 of the main article, we have given the following examples, such that Condition \ref{Condition-1} is satisfied. 

\begin{Example} \label{Example-1}
    Assume that there exist multivariate c.d.f.s $F_0^*(z)$ and $F_1^*(z)$ (possibly random) defined on $z\in \mathbb{R}^K$, such that:
\begin{equation}
    \sup_{z} |\FF_0(z) - F_0^*(z)| = o\left\{ (\log p)^{-1} \right\} \quad \text{and} \quad \sup_{z} |\FF_1(z) - F_1^*(z)| = o\left\{ (\log p)^{-1} \right\}, \quad a.s., \label{eq-example-1-1}
\end{equation}
as $p\to \infty$. Furthermore, $F_0^*(z)$ and $F_1^*(z)$ have the form: 
\begin{eqnarray*}
    F_0^*(z) = \prod_{k=1}^K F_{0,k}^*(z_k) \quad \text{and} \quad  F_1^*(z) = \prod_{k=1}^K F_{1,k}^*(z_k),
\end{eqnarray*}
for $F_{m,k}^*(\cdot), m=0,1; k=1,\ldots, K$ being c.d.f.s defined on $\mathbb{R}$. Then, Condition \ref{Condition-1} is satisfied. 
\end{Example}

\begin{Example} \label{Example-2}
    Let $U_1,\ldots, U_K$ be i.i.d. random variables, such that $E(U_k) = \mu_k, \text{var}(U_k) = 1$; let $\epsilon_{i,k}, i=1,\ldots,\widetilde p; k=1,\ldots, K$ be i.i.d. random variables such that $E(\epsilon_{i,k}) = 0, \text{var}(\epsilon_{i,k}) = 1$. Suppose that $\{U_k, k=1,\ldots, K\}$ and $\{\epsilon_{i,k}, i=1,\ldots, \widetilde p; k=1,\ldots, K\}$ are independent. Let $V_i = (V_{i,1},\ldots, V_{i,K})^T, i=1,\ldots, \widetilde p$, where 
    \begin{eqnarray}
    V_{i,k} = \sqrt{\rho} U_k + \sqrt{1-\rho} \epsilon_{i,k}, \label{eq-Example-2-1}
    \end{eqnarray}
    with $\rho\in [0,1)$ being a given constant. Then for each dimension $k = 1,\ldots, K$, $V_{1,k},\ldots, V_{\widetilde p, k}$ are dependent random variables, with the variance-covariance matrix having the {\it compound symmetry structure}. Denote by $\FF_V(\cdot)$ and $\FF_{V,k}(\cdot)$ the empirical c.d.f.s based on $\{V_i, i=1,\ldots, \widetilde p\}$ and $\{V_{i,k},\ldots, V_{\widetilde{p},k}\}$ respectively. We have, as $\widetilde p \to \infty$,
    \begin{eqnarray}
        \sup_{z\in \mathbb{R}^K} \left| \FF_V(z) - \prod_{k=1}^K \FF_{V,k}(z_k)\right| = O\left\{\sqrt{\frac{\log(\widetilde p)}{\widetilde p}} \right\}, \quad a.s. \label{eq-Example-2-1-1}
    \end{eqnarray}
The above conclusion implies that if the test statistics in $\mathX$ and in $\mathY$ respectively satisfy the compound symmetry structure \eqref{eq-Example-2-1}, then Condition \ref{Condition-1} is satisfied. 
\end{Example}

\begin{Example}  \label{Example-3}
The $\alpha$--mixing coefficient of two $\sigma$--algebras $\mathA$ and $\mathB$ is defined as
\begin{eqnarray}
    \alpha(\mathA, \mathB) = \sup_{A\in \mathA, B\in \mathB}\left| P(A\cap B) - P(A)P(B)\right|. \nonumber 
\end{eqnarray}
Let $\mathV = \{V_t\}_{t=1,\ldots,\widetilde p}$ be a sequence of $K$-dimensional stationary time series with c.d.f. $F_V(\cdot)$ for $V_t$. Define the $\sigma$-algebras:
\[
\mathA_i = \sigma(V_t, t\leq i), \qquad 
\mathB_j = \sigma(V_t, t \geq j),
\]
and denote
\begin{eqnarray}
\alpha_V(n) = \sup_{k\geq 1}\alpha\left(\mathA_k, \mathB_{k+n} \right). \label{eq-Example-3-def-alpha-V-n} 
\end{eqnarray}
Let $\FF_V(v)$ be the empirical c.d.f. of $\{V_t\}_{t=1}^{\widetilde p}$.  
If $\mathV$ satisfies
\begin{eqnarray}
    \alpha_{V}(n) &\leq& \exp(-cn), \label{eq-Condition-3-1} \\
    \xi^2 &<&  \infty, \label{eq-Condition-3-2} \\
    V_{t,1},\ldots,V_{t,K} &&\text{are independent for each } t, \label{eq-Condition-3-3} 
\end{eqnarray}
where $c>0$ is a universal constant, 
\begin{eqnarray}
\xi^2 = \sup_{v\in \mathbb{R}^K} \sup_{t\geq 1}\left[\text{var}\{W_t(v)\} + 2 \sum_{j > t}\left|\text{cov}\{W_t(v), W_j(v)\}\right| \right], \label{eq-Example-3-def-xi}
\end{eqnarray}
and $W_t(v)=I(V_t\le v)-F_V(v)$, then
\begin{equation}
\sup_{v\in \mathbb{R}^K}
\left|\FF_{V}(v) - \prod_{k=1}^K \FF_{V,k}(v_k)\right|
=
O\!\left(\sqrt{\frac{\log \widetilde p}{\widetilde p}}\right),\quad a.s. \nonumber 
\end{equation}
Consequently, if both $\mathX$ and $\mathY$ satisfy \eqref{eq-Condition-3-1}--\eqref{eq-Condition-3-3}, then Condition \ref{Condition-1} holds.
\end{Example}

\begin{Example} \label{Example-4}
    Let $V_t, t = 1,\ldots, p$, be a $K$-dimensional stationary $\text{ARMA}(c_1, c_2)$ time series, with error terms being i.i.d. continuous random vectors. For each $t$, $V_{t,1}, \ldots, V_{t,K}$ are independent. Recall $\mathH_0, \mathH_1$ defined in the main article. Let $Z_t = V_t$, if $t \in \mathH_0$, and let $Z_t = V_t + \mu$ if $t \in \mathH_1$, where $\mu \neq 0$ is a constant. Then Condition \ref{Condition-1} is satisfied.  
\end{Example}

Section 3.3 of the main article presents the following theorem (Theorem 2), which, under Conditions \ref{Condition-1}--\ref{Condition-4}, establishes the asymptotic properties of our estimators for the empirical c.d.f.s.

\setcounter{theorem}{1}

\begin{theorem} \label{theorem-2}
Assume Conditions \ref{Condition-1}--\ref{Condition-4}. We have, as $p\to \infty$,
\begin{itemize}
    \item[(a).] $\widehat \lambda_0 = p_0/p + o(1)$, a.s. 

    \item[(b).] For $m_1, m_2 \in \{0,1\}$, and $k=1,\ldots, K$, we have 
    \begin{eqnarray*}
        \int \left\{\widehat \FF_{m_1,k}(t) - \FF_{m_1, k}(t) \right\}^2 d \FF_{m_2,k}(t) = o(1), \quad a.s.
    \end{eqnarray*}
    
\end{itemize}
    
\end{theorem}

In Section 4 of the main article, we have defined: 
\begin{eqnarray}
\widehat{\FDP}(t)=\frac{\widehat \lambda_0\cdot M(t)}{\frac{1}{p}\cdot \max\left\{\sum\limits_{i=1}^pI(\widehat M_i\leq t),1\right\}}, \label{eq-def-hat-FDP-t-median} 
\end{eqnarray}
where $M(t)$ denotes the c.d.f. of the median of $K$ i.i.d. $\text{Uniform}(0,1)$ random variables, 
\begin{eqnarray}
\text{FDP}(t)=\frac{\sum\limits_{i \in \mathcal{H}_0}I(\hat{M}_i\leq t)}{\max\left\{\sum\limits_{i=1}^p I(\widehat{M}_i\leq t),1\right\}}, \label{eq-def-FDP-t-median}
\end{eqnarray}
and
\begin{eqnarray}
t_\alpha=\sup\limits_{t}\left\{t:\widehat{\FDP}(t)\leq \alpha\right\}. \label{eq-t-alpha-median}
\end{eqnarray}

Under Conditions \ref{Condition-1}--\ref{Condition-5}, the following theorem (Theorem 3 in Section 4 of the main article) guarantees FDP control for our method.

\begin{theorem}\label{median fdr theorem}
     Assume Conditions \ref{Condition-1}--\ref{Condition-5}. For any $\alpha\in(0,1)$,  $t_\alpha$ is defined by \eqref{eq-t-alpha-median}. When $p\to \infty$, we have 
     \begin{eqnarray*}
     {\rm{FDP}}(t_\alpha) \leq \alpha +o(1),\quad a.s. 
     \end{eqnarray*}
\end{theorem}

\subsection{Notations} \label{section-notations}

Recall that we have defined
\begin{eqnarray*}
     \FF_{\theta}(z) = \lambda_0  \prod_{k=1}^K F_{0,k}(z_k) + \lambda_1 \prod_{k=1}^K F_{1,k}(z_k),
\end{eqnarray*}
where $\lambda_0 + \lambda_1 = 1$,
\begin{eqnarray*}
    \theta = \left\{\lambda_0, F_{0,1}, \ldots, F_{0,K}, F_{1,1},\ldots, F_{1,K} \right\},
\end{eqnarray*}
which is the set of parameters to be estimated. We denote
\begin{eqnarray*}
    \theta_0 = \left\{\lambda_{0,0}, \FF_{0,1}, \ldots, \FF_{0,K}, \FF_{1,1}, \ldots,  \FF_{1,K} \right\},
\end{eqnarray*}
where $\lambda_{0,0} = p_0/p$, for each $k=1,\ldots, K$, $\FF_{0,k}(\cdot)$ is the empirical c.d.f. of $\{X_{i,k}, i=1,\ldots, p_0\}$, $\FF_{1,k}(\cdot)$ is the empirical c.d.f. of $\{Y_{i,k}, i=1,\ldots, p_1\}$. We have
\begin{eqnarray*}
\FF_{0,\theta_0}(z) =   \prod_{k=1}^K \FF_{0,k}(z_k), \quad  \FF_{1,\theta_0}(z) =   \prod_{k=1}^K \FF_{1,k}(z_k), 
\end{eqnarray*}
and denote $\FF_0(\cdot)$ and $\FF_1(\cdot)$ the empirical c.d.f.s based on $\mathX$ and $\mathY$, respectively. 
For a $Z\sim \FF_{\theta}(z)$, define
\begin{eqnarray*}
    f_\theta(s, t) = P\left(I(Z \leq t) = I(s \leq t)\right),
\end{eqnarray*}
for $s = (s_1,\ldots, s_K)^T, t = (t_1, \ldots, t_K)^T$; then we can verify 
\begin{eqnarray}
    f_\theta(s, t) = \sum_{m=0}^1 \lambda_m \prod_{k=1}^K \left\{ F_{m,k}^{I(s_k \leq t_k)}(t_k) \, \bar{F}_{m,k}^{I(s_k > t_k)}(t_k) \right\}, \label{def-f-theta}
\end{eqnarray}
For any given \( t \in \mathbb{R}^K \), let \( A_h(t) = \{s : I(s \leq t) = \Delta_h\} \), where \( \Delta_h \) is a length \( K \) vector with elements being 0 or 1. Clearly, there are \( Q = 2^K \) different possible values for \( \Delta_h \), and thus we shall write \( h = 1, \ldots, Q \). Furthermore, \( A_h(t) \) partition the \( \mathbb{R}^K \) space into \( Q \) disjoint regions. Based on the definition of \( f_\theta(s, t) \) given in \eqref{def-f-theta}, when \( s \in A_h(t) \), we have

\begin{itemize}
    \item when \( Z \sim \FF_{\theta}(z) \),
    \begin{eqnarray}
        f_\theta(s,t) = P\left(I(Z \leq t) = I(s \leq t)\right) = P(Z \in A_h(t)) = \int_{s\in A_h(t)} d \FF_\theta(s) \equiv \FF_{\theta,h}(t); \label{def-F-theta-h-t}
    \end{eqnarray}

    \item when \( Z \sim \FF_{\theta_0}(z) \), 
    \begin{eqnarray}
        f_{\theta_0}(s,t)  = \int_{s\in A_h(t)} d \FF_{\theta_0}(s) \equiv \FF_{\theta_0,h}(t); \label{def-F-theta-h-t-0}
    \end{eqnarray}

    \item when \( Z \sim \FF(z) \), we denote
    \begin{eqnarray}
        \FF_{h}(t) = \int_{s \in A_h(t)} d\FF(s). \label{def-F-h-t}
    \end{eqnarray}
\end{itemize}
For $m=0,1$, we can define $\FF_{m,\theta,h}, \FF_{m,\theta_0,h}$, and $\FF_{m,h}$ in the same manner as \eqref{def-F-theta-h-t}--\eqref{def-F-h-t}; clearly, we have 
\begin{eqnarray}
    \FF_{\theta, h} &=&  \lambda_0 \FF_{0,\theta,h} + (1-\lambda_0) \FF_{1,\theta,h} \nonumber \\
    \FF_{\theta_0,h} &=& \lambda_{0,0} \FF_{0,\theta_0,h} + (1-\lambda_{0,0}) \FF_{1,\theta_0,h} \nonumber \\
    \FF_h &=& \lambda_{0,0} \FF_{0,h} + (1-\lambda_{0,0}) \FF_{1,h}. \label{F-F-h-relation}
\end{eqnarray}

\subsection{Proof of Examples \ref{Example-1}-\ref{Example-4}}

\subsubsection{Proof of Example \ref{Example-1}} For each $m= 0,1; k=1,\ldots, K$, we have 
\begin{eqnarray*}
    \sup_{z_k \in \mathbb{R}}\left| \FF_{m,k}(z_k) - F_{m,k}^*(z_k)\right| &=& \sup_{z_k \in \mathbb{R}}\left| \FF_m(\infty, \ldots, z_k, \ldots, \infty) - F_m^*(\infty, \ldots, z_k, \ldots, \infty)\right| \\
    &\leq & \sup_{z \in \mathbb{R}^K} |\FF_m(z) - F_m^*(z)| = o\left\{ (\log p)^{-1} \right\}, \quad a.s., 
\end{eqnarray*}
based on \eqref{eq-example-1-1}, which leads to  
\begin{eqnarray}
    \sup_{z\in \mathbb{R}^K}\left|\FF_{m, \theta_0}(z) - F_m^*(z)\right| = \sup_{z\in \mathbb{R}^K}\left|\prod_{k=1}^K \FF_{m,k}(z_k) - \prod_{k=1}^K F_{m,k}^*(z_k) \right| = o\left\{ (\log p)^{-1} \right\}, \quad a.s. \label{eq-example-1-2}
\end{eqnarray}
Combining \eqref{eq-example-1-1} and \eqref{eq-example-1-2}, we conclude 
\begin{eqnarray*}
     \sup_{z\in \mathbb{R}^K}\left|\FF_m(z) - \FF_{m, \theta_0}(z)\right| = o\left\{ (\log p)^{-1} \right\}, \quad a.s.
\end{eqnarray*}
which verifies Condition \ref{Condition-1}. \epf


\subsubsection{Proof of Example \ref{Example-2}} We need the following lemma, which gives the multivariate Dvoretzky–Kiefer–Wolfowitz (DKW) inequality adopted from Lemma 4.1 in \citet{Naaman2021}. 
\begin{lemma} \label{lemma-DKW}
    Let $\epsilon_i, i=1,\ldots, \widetilde p$ be $K$ dimensional i.i.d. random vectors. We have 
    \begin{eqnarray*}
        P\left( \sup_{z\in \mathbb{R}^K} \left|\FF_{\epsilon}(z) - F_\epsilon(z) \right| > t \right) \leq K(\widetilde p + 1) e^{-2\widetilde p t^2},
    \end{eqnarray*}
    for every $t\geq 0$ and $\widetilde p$, where $\FF_\epsilon(\cdot)$ and $F_\epsilon(\cdot)$ denote the empirical c.d.f. of $\{\epsilon_i, i=1,\ldots, \widetilde p\}$ and the c.d.f. of $\epsilon_i$ respectively.  
\end{lemma}
We proceed to show Example \ref{Example-2}. Note that 
\begin{eqnarray*}
    \FF_V(z) = \frac{1}{\widetilde p}  \sum_{i=1}^{\widetilde p} \prod_{k=1}^K I(V_{i,k} \leq z_k) = \frac{1}{\widetilde p}  \sum_{i=1}^{\widetilde p} \prod_{k=1}^K I\left(\epsilon_{i,k} \leq \frac{z_k - \sqrt{\rho}U_k}{\sqrt{1-\rho}}\right) = \FF_{\epsilon}(W), \label{eq-Example-2-2}
\end{eqnarray*}
where $W = \left(\frac{z_1 - \sqrt{\rho}U_1}{\sqrt{1-\rho}}, \ldots,   \frac{z_K - \sqrt{\rho}U_K}{\sqrt{1-\rho}}\right)^T \equiv (W_1, \ldots, W_K)^T$, and thus
\begin{eqnarray}
    E\left\{ \FF_V(z)| U_1,\ldots, U_K \right\} = F_\epsilon(W) = \prod_{k=1}^K F_{\epsilon_k}(W_k),\label{eq-Example-2-3}
\end{eqnarray}
where $F_{\epsilon_k}(\cdot)$ denotes the c.d.f. of $\epsilon_{i,k}$. 
Based on Lemma \ref{lemma-DKW}, for any $t\geq 0$, we have
\begin{eqnarray}
    &&P\left( \sup_{z\in \mathbb{R}^K} \left|\FF_V(z) - E\left\{ \FF_V(z)| U_1,\ldots, U_K \right\} \right| > t \right) \nonumber \\
    &=& P\left(\sup_{z\in \mathbb{R}^K} \left| \FF_\epsilon(W) - F_\epsilon(W) \right| > t \right) \nonumber  \\
    &\leq &P\left(\sup_{z\in \mathbb{R}^K} \left| \FF_\epsilon(z) - F_\epsilon(z) \right| > t \right) \nonumber \\
    &\leq & K(\widetilde p + 1) e^{-2 \widetilde p t^2}. \label{eq-Example-2-4}
\end{eqnarray}
In \eqref{eq-Example-2-4}, set $t = C\sqrt{\log \widetilde p / \widetilde p}$, for a constant $C>1$, we have 
\begin{eqnarray*}
   && P\left( \sup_{z\in \mathbb{R}^K} \left|\FF_V(z) - E\left\{ \FF_V(z)| U_1,\ldots, U_K \right\} \right| > C \sqrt{\frac{\log \widetilde p }{\widetilde p }} \right)  \\ &\leq&  K(\widetilde p + 1) \widetilde p^{-2C^2} \leq 2 K \widetilde p^{1-2C^2}, \label{eq-Example-2-5}
\end{eqnarray*}
which implies 
\begin{eqnarray*}
   \sum_{\widetilde p =1}^\infty  P\left( \sup_{z\in \mathbb{R}^K} \left|\FF_V(z) - E\left\{ \FF_V(z)| U_1,\ldots, U_K \right\} \right| > C \sqrt{\frac{\log \widetilde p }{\widetilde p }} \right) < \infty. \label{eq-Example-2-6}
\end{eqnarray*}
Based on the Borel-Cantelli lemma, we have 
\begin{eqnarray*}
    P\left( \sup_{z\in \mathbb{R}^K} \left|\FF_V(z) - E\left\{ \FF_V(z)| U_1,\ldots, U_K \right\} \right| > C \sqrt{\frac{\log \widetilde p }{\widetilde p }}, i.o. \right) = 0, \label{eq-Example-2-7}
\end{eqnarray*}
or equivalently 
\begin{eqnarray*}
    \sup_{z\in \mathbb{R}^K} \left|\FF_V(z) - E\left\{ \FF_V(z)| U_1,\ldots, U_K \right\} \right| = O\left( \sqrt{\frac{\log \widetilde p }{\widetilde p }}\right), \quad a.s. \label{eq-Example-2-8}
\end{eqnarray*}
which together with \eqref{eq-Example-2-3} leads to 
\begin{eqnarray}
    \sup_{z\in \mathbb{R}^K} \left|\FF_V(z) - \prod_{k=1}^K F_{\epsilon_k}(W_k) \right| = O\left( \sqrt{\frac{\log \widetilde p }{\widetilde p }}\right). \quad a.s. \label{eq-Example-2-9}
\end{eqnarray}
With similar developments, for each $k=1,\ldots, K$, we have 
\begin{eqnarray}
    \sup_{z_k\in \mathbb{R}} \left|\FF_{V,k}(z_k) -  F_{\epsilon_k}(W_k) \right| = O\left( \sqrt{\frac{\log \widetilde p }{\widetilde p }}\right). \quad a.s. \label{eq-Example-2-10}
\end{eqnarray}
Combining \eqref{eq-Example-2-9} and \eqref{eq-Example-2-10} leads to \eqref{eq-Example-2-1-1}. We complete the proof of this example. \epf


\subsubsection{Proof of Example \ref{Example-3}}

Let
\[
W_t(v) = I(V_t \leq v) - F_V(v),
\]
and define
\[
S_{\widetilde p}(v) = \sum_{t = 1}^{\widetilde p} W_t(v) 
      = \widetilde p \{\FF_V(v) - F_V(v)\}.
\]
Since $W_t(v)$ is a measurable function of $V_t$, thus for any $v \in \mathbb{R}^K$, and $i,j$ being positive integers, 
\[
\sigma(W_t(v), t\leq i) \subset \sigma(V_t, t\leq i), \quad \text{and} \quad  \sigma(W_t(v), t\geq j) \subset \sigma(V_t, t\geq j). 
\]
This implies, 
\begin{eqnarray}
    \alpha_W(n) \leq \alpha_V(n) \leq \exp(-cn), \label{eq-Condition-3-revised-1}
\end{eqnarray}
for $c>0$ being a universal constant. 

Based on the conditions $\xi^2 < \infty$ and \eqref{eq-Condition-3-revised-1}, and applying Theorem 2 in \citet{MerlevedePeligradRio2009}, we have, for $n\geq 2$ and any $\epsilon>0$, there exists a constant $C>0$ depending on $c$ only, such that 
\[
P(|S_{\widetilde p}(v)|\ge \epsilon)
\le
\exp\!\left\{
-\frac{C\epsilon^2}{\xi^2 \widetilde p +1+\epsilon(\log \widetilde p)^2}
\right\},
\]
which further leads to 
\begin{eqnarray}
P\!\left(
\max_{v\in\widetilde\mathV}|S_{\widetilde p}(v)|\ge \epsilon
\right)
&\le&
(\widetilde p+1)^K
\exp\!\left\{
-\frac{C\epsilon^2}{\xi^2 \widetilde p +1+\epsilon(\log \widetilde p)^2}
\right\}  \nonumber \\  
&\leq & 2^K \widetilde p^K
\exp\!\left\{
-\frac{C\epsilon^2}{\xi^2 \widetilde p +1+\epsilon(\log \widetilde p)^2}
\right\}  \label{eq-Condition-3-revised-2}
\end{eqnarray}
where 
\[
\widetilde \mathV =
\Big\{
v=(v_1,\ldots,v_K)^T:
v_k\in\{V_{1,k},\ldots,V_{p,k}, +\infty \}, k =1,\ldots, K
\Big\}.
\]
Since $\epsilon >0$ is arbitrary, we can set $\epsilon$ to be $\epsilon_{\widetilde p}=A\sqrt{\widetilde p\log \widetilde p}$ for $A$ being a universal constant satisfying $K - CA^2/(\xi^2+2) < -1$.  Clearly, when $\widetilde p$ is sufficiently large, $1+\epsilon_{\widetilde p}(\log \widetilde p)^2 < 2\widetilde p$, and thus
\[
\exp\!\left\{
-\frac{C\epsilon_{\widetilde p}^2}{\xi^2 \widetilde p+1+\epsilon_{\widetilde p}(\log \widetilde p)^2}
\right\}
\le
\exp\!\left\{
-\frac{C\epsilon_{\widetilde p}^2}{(\xi^2+2)\widetilde p}
\right\},
\]
which together with \eqref{eq-Condition-3-revised-2} leads to 
\begin{eqnarray*}
&&\sum_{\widetilde p=1}^\infty
P\!\left(
\max_{v\in\widetilde\mathV}|S_{\widetilde p}(v)|\ge \epsilon_{\widetilde p}
\right) \lesssim \sum_{\widetilde p=1}^\infty  \widetilde p^K
\exp\!\left\{
-\frac{C\epsilon_{\widetilde p}^2}{(\xi^2+2)\widetilde p}
\right\} \\
&=&
\sum_{\widetilde p=1}^\infty
\widetilde p^{K-CA^2/(\xi^2+2)}<\infty.
\end{eqnarray*}
By the Borel–Cantelli lemma,
\[
\widetilde p^{-1/2}(\log \widetilde p)^{-1/2}
\max_{v\in\widetilde\mathV}|S_{\widetilde p}(v)|
=O(1)\quad a.s.,
\]
which yields
\begin{eqnarray}
\max_{v\in\widetilde\mathV}
|\FF_V(v)-F_V(v)|
=
O\!\left(\sqrt{\frac{\log \widetilde p}{\widetilde p}}\right),\quad a.s. \label{eq-Condition-3-revised-3}
\end{eqnarray}
If \eqref{eq-Condition-3-3} holds, then 
\[
F_V(v)=\prod_{k=1}^K F_{V,k}(v_k),
\]
which together with \eqref{eq-Condition-3-revised-3} leads to 
\begin{eqnarray}
\max_{v_k}
|\FF_{V,k}(v_k)-F_{V,k}(v_k)|
=
O\!\left(\sqrt{\frac{\log \widetilde p}{\widetilde p}}\right),\quad a.s. \label{eq-Condition-3-revised-4}
\end{eqnarray}
where $\FF_{V,k}(v)$ is the empirical c.d.f. of $\{V_{t,k}\}_{t=1}^{\widetilde p}$, and $F_{V,k}(v_k)$ is the c.d.f. of $V_{t,k}$. 

Combining \eqref{eq-Condition-3-revised-3} and \eqref{eq-Condition-3-revised-4}, we have 
\begin{equation*}
    \sup_{v\in \mathbb{R}^K} \left|\FF_{V}(v) - \prod_{k=1}^K \FF_{V,k}(v_k)\right| = \max_{v\in \widetilde \mathV} \left|\FF_{V}(v) - \prod_{k=1}^K \FF_{V,k}(v_k)\right|= O\left(\sqrt{\frac{\log \widetilde p}{ \widetilde p }} \right), \quad a.s. 
\end{equation*}
completing the proof. 
\epf

\subsubsection{Proof of Example \ref{Example-4}} It suffices to show that $\{V_t\}_{t=1}^p$ satisfies \eqref{eq-Condition-3-1} and \eqref{eq-Condition-3-2}. In fact, based on Theorem 1 in \citet{Mokkadem1988}, \eqref{eq-Condition-3-1} is valid. We proceed to show \eqref{eq-Condition-3-2}. Recall that $W_t(v) = I(V_t \leq v) - F_V(v)$ and thus $|W_t(v)|\leq 1$; we have 
\begin{eqnarray}
    \sup_{v \in \mathbb{R}^K} \sup_{t\geq 1} \text{var}\{W_t(v)\} \leq 1. \label{eq-Example-4-0}
\end{eqnarray}
The covariance terms in $\xi^2$ is considered as follows. Note that 
\begin{eqnarray*}
    W_t(v) = \int_0^1 \left[ I\left\{ W_t(v) > s\right\} - I\left\{ W_t(v) < -s\right\} \right] ds 
\end{eqnarray*}
For $s\in (0,1)$ and $v\in \mathbb{R}^K$, let $A(s,t) = \left\{\omega: W_t^\omega(v) > s \right\}$ and $B(s,t) = \left\{\omega: W_t^\omega(v) < -s \right\}$ be events in the sample space; we have 
\begin{eqnarray*}
    && E\left\{W_{t_1}(v) W_{t_2}(v) \right\} \\
    &=& \int_{(s_1,s_2)\in[0,1]^2} E \Big([ I\{ W_{t_1}(v) > s_1\} - I\left\{ W_{t_1}(v) < -s_1\right\}] \\
    &&\times\left[ I\left\{ W_{t_2}(v) > s_2\right\} - I\left\{ W_{t_2}(v) < -s_2\right\} \right] \Big) d s_1 d s_2 \\
    &=& \int_{(s_1,s_2) \in [0,1]^2} E \left[I\left\{ W_{t_1}(v) > s_1\right\}I\left\{ W_{t_2}(v) > s_2\right\}\right] d s_1 d s_2 \\
    && - \int_{(s_1,s_2) \in [0,1]^2} E \left[I\left\{ W_{t_1}(v) > s_1\right\}I\left\{ W_{t_2}(v) < -s_2\right\}\right] d s_1 d s_2\\
    && - \int_{(s_1,s_2) \in [0,1]^2} E \left[I\left\{ W_{t_1}(v) < -s_1\right\}I\left\{ W_{t_2}(v) > s_2\right\}\right] d s_1 d s_2\\
    && +\int_{(s_1,s_2) \in [0,1]^2} E \left[I\left\{ W_{t_1}(v) < -s_1\right\}I\left\{ W_{t_2}(v) < -s_2\right\}\right] d s_1 d s_2 \\
    &=& \int_{(s_1,s_2) \in [0,1]^2} P(A(s_1,t_1) \cap A(s_2,t_2)) d s_1 d s_2 \\
    &&                        - \int_{(s_1,s_2) \in [0,1]^2} P(A(s_1,t_1) \cap B(s_2,t_2)) d s_1 d s_2 \\
    && - \int_{(s_1,s_2) \in [0,1]^2} P(B(s_1,t_1) \cap A(s_2,t_2)) d s_1 d s_2  \\
    && + \int_{(s_1,s_2) \in [0,1]^2} P(B(s_1,t_1) \cap B(s_2,t_2)) d s_1 d s_2. 
\end{eqnarray*}
Similarly
\begin{eqnarray*}
    && E\left\{W_{t_1}(v)\right\} \cdot E\left\{W_{t_2}(v) \right\} \\ &=& \int_0^1 \left\{P(A(s_1, t_1))- P(B(s_1,t_1))\right\} d s_1 \int_0^1 \left\{P(A(s_2, t_2))- P(B(s_2,t_2)) \right\} d s_2 \\
    &=& \int P(A(s_1,t_1)) d s_1 \int P(A(s_2, t_2)) d s_2 - \int P(A(s_1,t_1)) d s_1 \int P(B(s_2, t_2)) d s_2 \\
    && - \int P(B(s_1,t_1)) d s_1 \int P(A(s_2, t_2)) d s_2 + \int P(B(s_1,t_1)) d s_1 \int P(B(s_2, t_2)) d s_2. 
\end{eqnarray*}
Thus, 
\begin{eqnarray*}
    && \text{cov}\left\{W_{t_1}(v), W_{t_2}(v)\right\} = E\left\{W_{t_1}(v) W_{t_2}(v) \right\} - E\left\{W_{t_1}(v)\right\} \cdot E\left\{W_{t_2}(v) \right\}  \\
    &=& \int_{(s_1,s_2) \in [0,1]^2} P(A(s_1,t_1) \cap A(s_2,t_2)) d s_1 d s_2 - \int P(A(s_1,t_1)) d s_1 \int P(A(s_2, t_2)) d s_2 \\
    &&- \left\{\int_{(s_1,s_2) \in [0,1]^2} P(A(s_1,t_1) \cap B(s_2,t_2)) d s_1 d s_2 - \int P(A(s_1,t_1)) d s_1 \int P(B(s_2, t_2)) d s_2\right\} \\
    &&- \left\{\int_{(s_1,s_2) \in [0,1]^2} P(B(s_1,t_1) \cap A(s_2,t_2)) d s_1 d s_2 - \int P(B(s_1,t_1)) d s_1 \int P(A(s_2, t_2)) d s_2\right\} \\
    && + \int_{(s_1,s_2) \in [0,1]^2} P(B(s_1,t_1) \cap B(s_2,t_2)) d s_1 d s_2  - \int P(B(s_1,t_1)) d s_1 \int P(B(s_2, t_2)) d s_2, 
\end{eqnarray*}
which together with the definition of ``$\alpha_V(n)$" given in \eqref{eq-Example-3-def-alpha-V-n} leads to 
\begin{eqnarray}
    \left| \text{cov}\left\{W_{t_1}(v), W_{t_2}(v)\right\} \right| \leq 4\alpha_W(|t_2 - t_1|). \label{eq-Example-4-1}
\end{eqnarray}
Furthermore, we have concluded at the beginning of the proof for this example that \eqref{eq-Condition-3-1} is valid, and based on \eqref{eq-Condition-3-revised-1} in the proof of Example \ref{Example-3}, \eqref{eq-Example-4-1} implies 
\begin{eqnarray*}
    \left| \text{cov}\left\{W_{t_1}(v), W_{t_2}(v)\right\} \right| \leq 4\exp(-c|t_2 - t_1|),
\end{eqnarray*}
which, together with \eqref{eq-Example-4-0} and the definition of $\xi^2$ given by \eqref{eq-Example-3-def-xi}, verifies \eqref{eq-Condition-3-2}. We complete the proof of this example. \epf

\subsection{Preliminaries for technical developments of Theorems \ref{theorem-2} and \ref{median fdr theorem}}

\begin{lemma} \label{lemma-1}
    Assume Condition \ref{Condition-1}. Recalling the definition of $\FF_{\theta_0, h}(\cdot), \FF_h(\cdot), \FF_{m,\theta_0,h}(\cdot)$, and $\FF_{m,h}(\cdot)$ defined in Section \ref{section-notations}, we have 
    \begin{eqnarray}
       \sup_{z\in \mathbb{R}^K} \left| \FF_{m,h}(z) - \FF_{m,\theta_0,h}(z)\right| &=& o\left\{ (\log p)^{-1} \right\}, \quad a.s. \label{eq-lemma-1-0} \\
        \sup_{z\in \mathbb{R}^K}|\FF_h(z) - \FF_{\theta_0, h}(z)| &=& o\left\{ (\log p)^{-1} \right\}, \quad a.s., \label{eq-lemma-1-1}
    \end{eqnarray}
    for any $h=1,\ldots, Q; m =0,1$. 
\end{lemma}
\proof We only need to show for $m=0$, and each $h=1,\ldots, Q$, \eqref{eq-lemma-1-0} is  valid, since the cases for $m=1$ can be established similarly; and thus, \eqref{eq-lemma-1-1} is valid based on \eqref{eq-lemma-1-0} and \eqref{F-F-h-relation}. For presentational clarity, we consider one $h$ only, the other cases are established in the same manner. Specifically, we consider the $h$, such that $\Delta_h = (0,1,\ldots, 1)^T$. We have,
\begin{eqnarray*}
  \FF_{0, \theta_0, h}(z) &=& \int_{s\in A_h(z)} d \FF_{0,\theta_0}(s) = \left\{1-\FF_{0,1}(z_1)\right\}\prod_{k=2}^K  \FF_{0,k}(z_k)  \\
  &=& \prod_{k=2}^K  \FF_{0,k}(z_k) - \FF_{0,\theta_0}(z) = \FF_{0,\theta_0}(\infty, z_2, \ldots, z_K) -  \FF_{0, \theta_0}(z)\\
  \FF_{0,h}(z) &=& \int_{s\in A_h(z) } d \FF_0(s) =\frac{1}{p_0} \sum_{j=1}^{p_0} I(X_{j,1}>z_1)\prod_{k=2}^KI(X_{j,k}\leq z_k)\\
  &=& \frac{1}{p_0} \sum_{j=1}^{p_0} \prod_{k=2}^KI(X_{j,k}\leq z_k) - \FF_0(z) \\
  &=& \FF_0(\infty, z_2, \ldots, z_K) - \FF_0(z). 
\end{eqnarray*}
Thus, 
\begin{eqnarray*}
    \sup_{z\in \mathbb{R}^K}|\FF_{0,h}(z) - \FF_{0,\theta_0, h}(z)| &\leq& \sup_{z\in \mathbb{R}^K}\left|\FF_0(\infty, z_2, \ldots, z_K)  - \FF_{0, \theta_0}(\infty, z_2, \ldots, z_K) \right| \\
    && +\sup_{z\in \mathbb{R}^K} \left|\FF_0(z) - \FF_{0, \theta_0}(z)  \right| \\
    &\leq &  2 \sup_{z\in \mathbb{R}^K} \left|\FF_0(z) - \FF_{0, \theta_0}(z)  \right| = o\left\{ (\log p)^{-1} \right\}, \quad a.s.,
\end{eqnarray*}
based on Condition \ref{Condition-1}. \epf

\begin{lemma} \label{lemma-2}
    For $p=1,2,\ldots$, let $H_p(z)$ be an arbitrary sequences of c.d.f.s defined on $z\in \mathbb{R}^K$, and let $P_p(z)$, $Q_p(z)$ be sequences of c.d.f.s defined on $z\in \mathbb{R}^K$, such that as $p\to \infty$,
    \begin{eqnarray}
    \sup_{z}|P_p(z) - Q_p(z)| = o(1), \quad a.s. \label{eq-lemma-2-0}
    \end{eqnarray}
    Then, we have
    \begin{eqnarray*}
    \int H_p(z) d\left\{P_p(z) - Q_p(z) \right\} = o(1), \quad a.s.
    \end{eqnarray*}
\end{lemma}

\proof Consider
\begin{eqnarray}
        \nonumber \int H_p(z)dP_p(z)&=&\iint I(s_1\leq z_1,\ldots,s_K\leq z_K)d H_p(s) dP_p(z)\\
       \nonumber &=&\iint \left\{\prod^K_{k=1} I(z_k\geq s_k)\right\} dP_p(z)dH_p(s)\\
        &=&\int P(Z\geq s)dH_p(s),\quad \text{with $Z\sim P_p(\cdot)$,}\label{lemma1-a-proof-eq-1}
    \end{eqnarray}
    where for $a = (a_1, \ldots, a_K)^T , b = (b_1, \ldots, b_K)^T \in \mathbb{R}^K$, $a\leq b$ means $a_k \leq b_k$, for every $k=1,\ldots,K$. On the other hand, applying the inclusive-exclusive principle for probability events, we have 
    \begin{eqnarray}
        P(Z\geq s) &=& 1 - P\left(\cup_{k=1}^K\{Z_k < s_k\} \right) \nonumber \\
        &=& 1- \sum_{S\subset \{1,\ldots, K\}} (-1)^{|S|-1} P\left(\cap_{k\in S} \{Z_k < s_k\} \right) \nonumber \\
        &=& 1+ \sum_{S\subset \{1,\ldots, K\}} (-1)^{|S|} P_p(\widetilde{s}_S), \label{eq-lemma-2-1}
    \end{eqnarray}
    where $|S|$ is denoting the number of elements contained in the nonempty subset $S$ of $\{1,\ldots, K\}$,
    \begin{equation*}\label{lemma1-a-proof-eq-3}
        \tilde{s}_{S}=\left(\tilde{s}_{1},\ldots,\tilde{s}_K\right),\quad \text{with}\quad \tilde{s}_k=\begin{cases}
        s_{k}-& k \in S\\
        +\infty& k \notin S
    \end{cases}
    \end{equation*}
    and thus, for $Z\sim P_p(\cdot)$,
    \begin{eqnarray*}
        P\left(\cap_{k\in S} \{Z_k < s_k\} \right) = P\left(\cap_{k\in S} \{Z_k < s_k\} \cap_{k\notin S}\{Z_k \leq \infty \}\right) = P_p(\widetilde s_S). 
    \end{eqnarray*}

Combining \eqref{lemma1-a-proof-eq-1} and \eqref{eq-lemma-2-1}, we have 
\begin{eqnarray}
    \int H_p(z)dP_p(z) = 1+ \int \sum_{S\subset \{1,\ldots, K\}} (-1)^{|S|} P_p(\widetilde{s}_S) d H_p(s). \label{eq-lemma-2-2} 
\end{eqnarray}
Similarly, 
\begin{eqnarray}
        \int H_p(z)dQ_p(z) = 1+ \int \sum_{S\subset \{1,\ldots, K\}} (-1)^{|S|} Q_p(\widetilde{s}_S) d H_p(s). \label{eq-lemma-2-3}
\end{eqnarray}
Combining \eqref{eq-lemma-2-2} and \eqref{eq-lemma-2-3}, we have 
\begin{eqnarray*}
    \int H_p(z)dP_p(z) - \int H_p(z)dQ_p(z) &=& \left|\sum_{S\subset \{1,\ldots, K\}} (-1)^{|S|} \int \left\{ P_p(\widetilde s_S) - Q_p(\widetilde s_S)\right\} d H_p(s) \right| \\
    &\leq & \sum_{S\subset \{1,\ldots, K\}} \int \left| P_p(\widetilde s_S) - Q_p(\widetilde s_S) \right| d H_p(s) \\
    &\leq & (2^K-1) \sup_{s\in \mathbb{R}^K} |P_p(s) - Q_p(s)| = o(1), \quad a.s., 
\end{eqnarray*}
by noting the assumption \eqref{eq-lemma-2-0}. We complete the proof of this lemma. \epf

\begin{lemma}\label{lipschitz prod lemma}
    For $a_k, b_k \in [0,1], k=1,\ldots,K$, we have $$\left|\prod_{k=1}^K a_k-\prod\limits_{k=1}^K b_k\right|\leq \sum\limits_{k=1}^K |a_k-b_k|.$$
\end{lemma}
\proof
We show this lemma by applying mathematical induction. 
When $K=1$, $|a_1-b_1|\leq |a_1-b_1|$ holds; suppose when $K=j$ the inequality holds, then for $K=j+1$,
\begin{align*}
    \left|\prod_{k=1}^{j+1}a_k-\prod_{k=1}^{j+1}b_k\right|&=\left|(a_{j+1}-b_{j+1})\prod_{k=1}^j a_k+b_{j+1}\left(\prod_{k=1}^j a_k-\prod_{k=1}^jb_k\right)\right|\\
    &\leq |a_{j+1}-b_{j+1}|+\left|\prod_{k=1}^j a_k-\prod\limits_{k=1}^jb_k\right|\\
    &\leq \sum\limits_{k=1}^{j+1}|a_k-b_k|,
\end{align*}
which indicates when $K=j+1$, the inequality claimed by the lemma holds. We complete the proof of this lemma. \epf

\begin{lemma}\label{simple property for median of two sequence}
    For two sequence $\{a_k\}_{k=1}^K$, $\{b_k\}_{k=1}^K$ satisfying $\max\limits_{1\leq k \leq K} |a_k-b_k|\leq \epsilon$, we have $$ \left|\text{Median}\left\{a_i\right\}_{i=1}^K-\text{Median}\left\{b_i\right\}_{i=1}^K \right|\leq \epsilon.$$
\end{lemma}
\proof
We first show for each $i=1,\ldots, K$, 
\[a_{(i)}\le b_{(i)}+\epsilon,
\]
where $\{a_{(i)}\}_{i=1}^K$ and $\{b_{(i)}\}_{i=1}^K$ denote the corresponding order statistics.
Given an $i\in\{1,\dots ,K\}$, let
$
B_i:=\{k: b_k\le b_{(i)}\}
$, then $|B_i| \geq i$.  
For every $k\in B_i$, based the condition $\max\limits_{1\leq k \leq K} |a_k-b_k|\leq \epsilon$, we have
\[
a_k \le b_k+\epsilon \le b_{(i)}+\epsilon .
\]
Hence there are at least $i$ elements in $\{a_i\}_{i=1}^K$ that do not exceed $b_{(i)}+\epsilon$, and thus,
\[
a_{(i)} \leq b_{(i)}+\epsilon .
\]
Likewise, we have  
$$b_{(i)}\leq a_{(i)}+\epsilon.$$ 
If $K$ is an odd number, by taking $i = (p+1)/2$, we have 
\begin{eqnarray}
\left|\text{Median}\left\{a_i\right\}_{i=1}^K-\text{Median}\left\{b_i\right\}_{i=1}^K\right|=\left|a_{\left(\frac{p+1}{2}\right)}-b_{\left(\frac{p+1}{2}\right)}\right|\leq \epsilon. \label{eq-simple property for median of two sequence-1}
\end{eqnarray}
When $K$ is an even number, since $b_{(i)}-\epsilon\leq a_{(i)}\leq b_{(i)}+\epsilon$ holds for both $i=\frac{K}{2}$ and $\frac{K}{2}+1$, we have $$\frac{b_{\left(\frac{K}{2}\right)}+b_{\left(\frac{K}{2}+1\right)}-2\epsilon}{2}\leq \frac{a_{\left(\frac{K}{2}\right)}+a_{\left(\frac{K}{2}+1\right)}}{2}\leq \frac{b_{\left(\frac{p}{2}\right)}+b_{\left(\frac{p}{2}+1\right)}+2\epsilon}{2},$$ which indicates
\begin{eqnarray}
    \left|\text{Median}\left\{a_i\right\}_{i=1}^K-\text{Median}\left\{b_i\right\}_{i=1}^K\right| &=& \left| \frac{a_{\left(\frac{K}{2}\right)}+a_{\left(\frac{K}{2}+1\right)}}{2} - \frac{b_{\left(\frac{K}{2}\right)}+b_{\left(\frac{K}{2}+1\right)}}{2} \right|\nonumber \\
    &\leq & \epsilon \label{eq-simple property for median of two sequence-2}
\end{eqnarray}
Combining \eqref{eq-simple property for median of two sequence-1}
and \eqref{eq-simple property for median of two sequence-2}, we complete the proof of this lemma. \epf

\subsection{Proof of Theorem \ref{theorem-2}}

We first establish the $L_2$ convergence of the estimators for the mixture distribution, i.e., 
\[
\FF_{\widehat \theta}(z) =\widehat \lambda_0  \prod_{k=1}^K \widehat \FF_{0,k}(z_k) + \widehat \lambda_1 \prod_{k=1}^K \widehat \FF_{1,k}(z_k)
\]
to 
\[
\FF_{\theta_0}(z) = \lambda_{0,0}  \prod_{k=1}^K \FF_{0,k}(z_k) + \lambda_{1,0} \prod_{k=1}^K \FF_{1,k}(z_k),
\]
under the empirical measure $\FF(\cdot)$. We have the following lemma.

\begin{lemma} \label{lemma-mix-consistency}
    Assume Condition \ref{Condition-1}. We have, as $p\to \infty$,
    \begin{eqnarray*}
        \int \left\{ \FF_{\widehat \theta}(z) - \FF_{\theta_0}(z) \right\}^2 d \FF(z) = o(1), \quad a.s. \label{eq-lemma-mix-consistency-1}
    \end{eqnarray*}
\end{lemma}
\proof Recall that $\FF(\cdot)$ is the empirical c.d.f. of $\{Z_1, \ldots, Z_p\}$, we have 
\begin{eqnarray*}
    \ell(\theta) &=& \sum_{j=1}^p \sum_{i=1}^p \log \left[ \sum_{m=0}^1\lambda_m \prod_{k=1}^K F_{m,k}^{I_{i,j,k}}(Z_{j,k})\cdot \bar F_{m,k}^{1-I_{i,j,k}}(Z_{j,k}) \right] \\
    &=& \sum_{j=1}^p \sum_{i=1}^p \log f_{\theta}(Z_i, Z_j) = \frac{1}{p^2} \iint \log f_\theta(s,t) d\FF(s) d\FF(t), \label{eq-lemma-mix-consistency-2}
\end{eqnarray*}
where $f_\theta(s, t)$ is defined by \eqref{def-f-theta}. Since $\widehat \theta$ is the maximiser of $\ell(\theta)$, we have $\ell(\widehat \theta) \geq \ell(\theta_0)$, where 
\begin{eqnarray*}
    \theta_0 = \left\{p_0/p_1, \FF_{m,k}, m=1,2; k = 1,\ldots, K\right\}.  \label{eq-lemma-mix-consistency-3}
\end{eqnarray*}
Based on the concavity of the logarithm, we have 
\begin{eqnarray}
\iint
\log\!\Bigl\{\tfrac{f_{\hat{\theta}}(s,t)+f_{\theta_0}(s,t)}{2\,f_{\theta_0}(s,t)}\Bigr\}
\,d\FF(s)\,d\FF(t)
\;\ge\;0. \label{eq-lemma-mix-consistency-4}
\end{eqnarray}
Let 
\begin{eqnarray}
    M_n(\theta)= \iint
\log\!\Bigl\{\tfrac{f_{\theta}(s,t)+f_{\theta_0}(s,t)}{2\,f_{\theta_0}(s,t)}\Bigr\}
\,d\FF(s)\,d\FF(t) \equiv \mathI_1(\theta) + \mathI_2(\theta), \label{eq-lemma-mix-consistency-5}
\end{eqnarray}
where 
\begin{align*} \label{eq-lemma-mix-consistency-6}
\mathI_1(\theta)
&=
\iint
\log\!\Bigl\{\tfrac{f_{\theta}(s,t)+f_{\theta_0}(s,t)}{2\,f_{\theta_0}(s,t)}\Bigr\}
\,d\left\{\FF(s)-\FF_{\theta_0}(s)\right\} d\FF(t),
\\
\mathI_2(\theta)
&=
\iint
\log\!\Bigl\{\tfrac{f_{\theta}(s,t)+f_{\theta_0}(s,t)}{2\,f_{\theta_0}(s,t)}\Bigr\}
\,d\FF_{\theta_0}(s)\,d\FF(t).  
\end{align*}
We consider $\mathI_1(\theta)$ and $\mathI_2(\theta)$ separately. For $\mathI_2(\theta)$, based on the partition \( \mathbb{R}^K = \cup_{h=1}^Q A_h(t)\), and the inequality $-\frac12 \log x \geq 1-\sqrt{x}$ for any $x>0$, we have 
\begin{eqnarray} 
-\mathI_2(\theta)&=&\sum\limits_{h=1}^Q\iint_{s \in A_h(t)} -\log\!\Bigl\{\tfrac{f_{\theta}(s,t)+f_{\theta_0}(s,t)}{2\,f_{\theta_0}(s,t)}\Bigr\}\,d\FF_{\theta_0}(s)\,d\FF(t) \nonumber  \\
&=&
\sum_{h=1}^Q
\int
\FF_{\theta_0,h}(t)\left[
-\log\Bigl\{
\tfrac{\FF_{\theta,h}(t)+\FF_{\theta_0,h}(t)}{2\,\FF_{\theta_0,h}(t)}
\Bigr\}\right]
\,d\FF(t)
\nonumber \\
& \ge& \sum_{h=1}^Q
\int
\FF_{\theta_0,h}(t)\left[
2-2\sqrt{
\tfrac{\FF_{\theta,h}(t)+\FF_{\theta_0,h}(t)}{2\,\FF_{\theta_0,h}(t)}
\Bigr\}} \right]
\,d\FF(t) 
\nonumber \\
&=&
\sum_{h=1}^Q
\int
\left[
\sqrt{\tfrac{\FF_{\theta,h}(t)+\FF_{\theta_0,h}(t)}{2}}
- \sqrt{\FF_{\theta_0,h}(t)}
\right]^2
\,d\FF(t)
\label{eq-lemma-mix-consistency-7} \\
&\ge&
\frac{1}{16}
\sum_{h=1}^Q
\int
\left|\sqrt{\FF_{\theta,h}(t)}-\sqrt{\FF_{\theta_0,h}(t)}\right|^2
\,d\FF(t).
\label{eq-lemma-mix-consistency-8}\\
&\ge&
\frac{1}{64}
\sum_{h=1}^Q
\int
\bigl|\FF_{\theta,h}(t)-\FF_{\theta_0,h}(t) \bigr|^2
\,d\FF(t), \label{eq-lemma-mix-consistency-9} 
\end{eqnarray}
where from \eqref{eq-lemma-mix-consistency-7} to \eqref{eq-lemma-mix-consistency-8}, we have used the inequality that for any $x\geq 0, y\geq 0$, 
\[
\left|\sqrt{\frac{x+y}{2}} - \sqrt{x} \right| \geq \frac 14 \left|\sqrt{x} - \sqrt{y} \right|. 
\]

We proceed to consider $\mathI_1(\theta)$. We have 
\begin{align} \label{eq-lemma-mix-consistency-10}
\mathI_{1}(\theta)
&=
\sum_{h=1}^Q
\iint_{s\in A_{h}(t)}
\log\!\Bigl\{\tfrac{f_{\theta}(s,t)+f_{\theta_0}(s,t)}{2\,f_{\theta_0}(s,t)}\Bigr\}
d\left\{\FF(s)-\FF_{\theta_0}(s)\right\}\,d\FF(t)
\nonumber \\
&=\;
\sum_{h=1}^Q
\int
\left\{\FF_h(t)-\FF_{\theta_0,h}(t)\right\}
\log\!\Bigl\{\tfrac{\FF_{\theta,h}(t)+\FF_{\theta_0,h}(t)}{2\,\FF_{\theta_0,h}(t)}\Bigr\}
\,d\FF(t).
\end{align}
On the other hand, we are able to check that for each $h=1,\ldots, Q$, $Z_i \in \mathZ$, 
\begin{eqnarray}
    \FF_h(Z_i)>0\quad \text{implies}\quad \FF_{\theta_0, h}(Z_i) > 0  \label{eq-lemma-mix-consistency-10-1}
\end{eqnarray}
This can be done by checking every $h=1,\ldots, Q$, this statement holds. Here, we check only one, the others can be done follow the same strategy. In particular, we check the $h$, such that the corresponding $\Delta_h = (0,1,\ldots, 1)^T$. Based on \eqref{def-f-theta}, \eqref{def-F-theta-h-t}, \eqref{def-F-theta-h-t-0}, and \eqref{def-F-h-t},
\begin{eqnarray}
  \FF_{\theta_0, h}(z) &=& \lambda_{0,0} \left\{1-\FF_{0,1}(z_1)\right\}\prod_{k=2}^K  \FF_{0,k}(z_k)  + \lambda_{1,0} \left\{1-\FF_{1,1}(z_1)\right\}\prod_{k=2}^K  \FF_{1,k}(z_k) \label{eq-lemma-mix-consistency-11}\\
  \FF_h(z) &=& \frac 1p \sum_{j=1}^p I(Z_{j,1}>z_1)\prod_{k=2}^KI(Z_{j,k}\leq z_k). \label{eq-lemma-mix-consistency-12} 
\end{eqnarray}
Based on \eqref{eq-lemma-mix-consistency-12}, $\FF_h(Z_i)>0$ implies that there exists a $j$, such that 
\begin{eqnarray*}
I(Z_{j,1}>Z_{i,1}) =1, \quad \text{and for each }k=2,\ldots, K, I(Z_{j,k}\leq Z_{i,k}) = 1. 
\end{eqnarray*}
If $Z_j\in \mathX$, $I(Z_{j,1}>Z_{i,1})$ is a term in the summation of $1-\FF_{0,1}(Z_i) = \frac{1}{p_0} \sum_{j=1}^{p_0} I(X_{j,1} > Z_i)$, and thus $1-\FF_{0,1}(Z_i) > 0$; likewise, for each $k=2,\ldots, K$, $I(Z_{j,k}\leq Z_{i,k})$ is a term in the summation of $\FF_{0,k}(Z_i) = \frac{1}{p_0} \sum_{j=1}^{p_0} I(X_{j,1}\leq Z_i)$, and thus $\FF_{0,k}(Z_i) > 0$. As a consequence, in view of the structure of \eqref{eq-lemma-mix-consistency-11}, we conclude $\FF_{\theta_0, h}(Z_i) >0$, as its first term is greater than 0, i.e., 
\begin{eqnarray*}
    \left\{1-\FF_{0,1}(Z_{i,1})\right\}\prod_{k=2}^K  \FF_{0,k}(Z_{i,k}) > 0. \label{eq-lemma-mix-consistency-13} 
\end{eqnarray*}
Similarly, if $Z_j\in \mathY$, we can conclude $\FF_{\theta_0, h}(Z_i) >0$ by checking that its second term is greater than 0, i.e., 
\begin{eqnarray*}
    \left\{1-\FF_{1,1}(Z_{i,1})\right\}\prod_{k=2}^K  \FF_{1,k}(Z_{i,k})>0. \label{eq-lemma-mix-consistency-14} 
\end{eqnarray*}
In conclusion, we have checked $\FF_{\theta_0, h}(Z_i) >0$, and thus \eqref{eq-lemma-mix-consistency-10-1} is valid; or equivalently
\begin{eqnarray*}
\FF_{\theta_0, h}(Z_i) = 0 \quad \text{implies}\quad  \FF_h(Z_i)=0. \label{eq-lemma-mix-consistency-15} 
\end{eqnarray*}
For each $h$, consider the set $\widetilde \mathZ = \left\{\widetilde Z: \widetilde Z \in \mathZ\quad \text{and} \quad \FF_{\theta_0, h}(\widetilde Z) > 0\right\}$, and assume that $\widetilde \mathZ$ contains $p_h$ number of elements; denote by $\widetilde \FF_h(\cdot)$ the empirical c.d.f. based on $\widetilde \mathZ$. Furthermore, based on the structure of $\FF_{\theta_0,h}(\cdot)$, we have for any $\widetilde Z \in \widetilde \mathZ$, 
\begin{eqnarray*}
    \FF_{\theta_0, h}(\widetilde Z) \geq \min\left\{\frac{\lambda_{0,0}}{p_0^K}, \frac{\lambda_{1,0}}{p_1^K} \right\} \geq \frac{1}{p^K}. \label{eq-lemma-mix-consistency-16} 
\end{eqnarray*}

Based on the partition \( \mathbb{R}^K = \cup_{h=1}^Q A_h(t)\) and \eqref{eq-lemma-mix-consistency-10}, based on Condition \ref{Condition-1} and applying Lemma \ref{lemma-1}, and using the convention $0\log 0 = 0$, we have 
\begin{align} \label{eq-lemma-mix-consistency-17} 
    |\mathI_1(\theta)|&=\left|\sum_{h=1}^Q\int\bigl\{ \FF_h(t)-\FF_{\theta_0,h}(t)\}\log\!\Bigl\{\tfrac{\FF_{\theta,h}(t)+\FF_{\theta_0,h}(t)}{2\,\FF_{\theta_0,h}(t)}\Bigr\}\,d\FF(t) \right|.\nonumber \\
    &=\left|\sum_{h=1}^Q \frac{p_h}{p}\int\left\{\FF_h(t)-\FF_{\theta_0,h}(t)\right\}\log\!\Bigl\{\tfrac{\FF_h(t)+\FF_{\theta_0,h}(t)}{2\,\FF_{\theta_0,h}(t)}\Bigr\}\,d\tilde\FF_h(t)\right|.\nonumber \\
    &\leq \sum_{h=1}^Q \int
\left|\FF_h(t)-\FF_{\theta_0,h}(t)\right|
\;\left|\log\left\{\tfrac{\FF_{\theta,h}(t)+\FF_{\theta_0,h}(t)}{2\,\FF_{\theta_0,h}(t)}\right\}\right|
\,d\widetilde{\FF}_h(t).\nonumber \\
& \leq \sum_{h=1}^Q K \int
\left|\FF_h(t)-\FF_{\theta_0,h}(t)\right| (\log p)  d\widetilde{\FF}_h(t). \nonumber \\
& = o(1), \quad a.s.
\end{align}
Combining \eqref{eq-lemma-mix-consistency-4}, \eqref{eq-lemma-mix-consistency-5}, \eqref{eq-lemma-mix-consistency-9}, and \eqref{eq-lemma-mix-consistency-17} leads to 
\begin{eqnarray*}
    0 &\leq& \frac{1}{64}
\sum_{h=1}^Q
\int
\bigl|\FF_{\widehat \theta,h}(t)-\FF_{\theta_0,h}(t) \bigr|^2
\,d\FF(t) \leq -\mathI_2(\widehat \theta) \leq -\mathI_2(\widehat \theta) + M_n(\widehat \theta) \\
&=& \mathI_1(\widehat \theta) = o(1), \quad a.s.,
\end{eqnarray*}
which completes the proof of the lemma by choosing $h$ such that $\Delta_h = (1,\ldots, 1)$. \epf

With Lemma \ref{lemma-mix-consistency}, we proceed to show Theorem \ref{theorem-2}. We need to show that for each combination of $m_1, m_2\in\{0,1\}$, $k=1,\ldots, K$, we have, as $p\to \infty$,
\begin{eqnarray} \label{convergence of lambda}
   p_0/p = \widehat \lambda_0 + o(1), \quad a.s. 
\end{eqnarray}
and
\begin{equation}\label{marginal L2 with mixture}
    \int\left\{\widehat \FF_{m_1,k}(t)-\FF_{m_1,k}(t)\right\}^2 d \FF_{m_2,k}(t)=o(1),\quad a.s.
\end{equation}
To this end, let 
\begin{eqnarray}
\Omega=\left\{\omega: \lim\limits_{p\rightarrow \infty}\int \left\{\FF^\omega_{\widehat{\theta}}(t)-\FF_{\theta_0}^\omega(t)\right\}^2 d \FF^\omega(t)=0\right\}.  \label{def-Omega}
\end{eqnarray}
Then, based on Lemma \ref{lemma-mix-consistency}, $P(\Omega) = 1$. For each $\omega \in \Omega$ and $(m,k)$ combination, based on the Helly's extraction principle and Bolzano-Weierstrass theorem, there exists a subsequence of $\widehat \theta^{\omega}$ that converges to a proper limit, denoted by $\widetilde \theta^{\omega} = (\widetilde \lambda_0^{\omega}, \widetilde F_{0,1}^\omega, \ldots, \widetilde F_{1,K}^{\omega})$; along this subsequence, there exists a further subsequence, such that $\theta_0^{\omega}$ convergences to a proper limit, denoted by $\theta^{\omega} = \{ \lambda_0^\omega, F_{0,1}^\omega, \ldots, F_{1,K}^\omega\}$. Without loss of generality, we assume, as $p\to \infty$,
\begin{eqnarray*} \label{eq-theorem-2-2}
  \widehat \theta^{\omega} - \widetilde \theta^{\omega} = o(1) \quad \text{and}\quad \theta_0^\omega - \theta^\omega = o(1). 
\end{eqnarray*}
Based on $\widetilde \theta^{\omega}$ and $\theta^{\omega}$, we denote
\begin{eqnarray*} 
    \widetilde F^\omega(z) = \sum_{m=0}^1 \widetilde \lambda_m^{\omega} \prod_{k=1}^K \widetilde F_{m,k}^\omega (z) \quad \text{and} \quad F^\omega(z) = \sum_{m=0}^1 \lambda_m^{\omega} \prod_{k=1}^K F_{m,k}^\omega (z). \label{eq-theorem-2-3}
\end{eqnarray*}

Recalling $F^*(\cdot)$ defined in Condition \ref{Condition-4}, and based on Conditions \ref{Condition-2} and \ref{Condition-3}, we have 
\begin{eqnarray*}
    \sup_{z\in \mathbb{R}}|\FF(z) - F^*(z)| = o(1), \quad a.s. \label{eq-theorem-2-4}
\end{eqnarray*}
Without loss of generality, we assume for every $\omega \in \Omega$, 
\begin{eqnarray}
    \sup_{z\in \mathbb{R}}|\FF^\omega(z) - F^{*\omega}(z)| = o(1), \label{eq-theorem-2-5}
\end{eqnarray}
which together with the fact that $\FF_{\widehat \theta}^2(\cdot)$, $\FF_{\widetilde \theta^\omega}^2(\cdot)$, and $\FF_{\widehat \theta}(\cdot) \cdot \FF_{\widetilde \theta^\omega}(\cdot)$ are all c.d.f.s, and the conclusion in Lemma \ref{lemma-2}, leads to 
\begin{eqnarray*}
    \int \left\{\FF_{\widehat \theta}^\omega(z) -  \FF_{\widetilde \theta^\omega}(z) \right\}^2 d \left\{ \FF^\omega(z) - F^{*\omega}(z) \right\} = o(1),  \label{eq-theorem-2-6}
\end{eqnarray*}
Thus, we have 
\begin{eqnarray}
    \int \left\{\FF_{\widehat \theta}^\omega(z) -  \FF_{\widetilde \theta^\omega}(z) \right\}^2 d \FF^\omega(z) &=& \int \left\{\FF_{\widehat \theta}^\omega(z) -  \FF_{\widetilde \theta^\omega}(z) \right\}^2 d \left\{ \FF^\omega(z) - F^{*\omega}(z) \right\} \nonumber \\
    && + \int \left\{\FF_{\widehat \theta}^\omega(z) -  \FF_{\widetilde \theta^\omega}(z) \right\}^2 d  F^{*\omega}(z) \nonumber \\
    &=& o(1), \label{eq-theorem-2-7}
\end{eqnarray}
where the second term on the right hand side being $o(1)$ is based on the definition of $\widetilde \theta^\omega$ and the dominant convergence theorem. Similarly to the development of \eqref{eq-theorem-2-7}, we have  
\begin{eqnarray}
    \int \left\{\FF_{\theta_0}^\omega(z) -  \FF_{\theta^\omega}(z) \right\}^2 d \FF^\omega(z) = o(1).  \label{eq-theorem-2-8}
\end{eqnarray}
Moreover, based on the definition of $\Omega$ given by \eqref{def-Omega}, for each $\omega \in \Omega$,
    \begin{eqnarray}
        \int \left\{ \FF_{\widehat \theta}^\omega (z) - \FF_{\theta_0}^\omega (z) \right\}^2 d \FF^\omega(z) = o(1).  \label{eq-theorem-2-9}
    \end{eqnarray}
Combining \eqref{eq-theorem-2-7}--\eqref{eq-theorem-2-9}, we have 
\begin{eqnarray*}
    0\leq \|\FF_{\widetilde \theta^\omega} - \FF_{\theta^\omega}\|_{\FF^\omega, 2} &\leq& \|\FF_{\widehat \theta}^\omega -  \FF_{\widetilde \theta^\omega}\|_{\FF^\omega, 2} + \|\FF_{\theta_0}^\omega -  \FF_{\theta^\omega}\|_{\FF^\omega, 2} + \|\FF_{\widehat \theta}^\omega  - \FF_{\theta_0}^\omega\|_{\FF^\omega, 2} \\
    &=& o(1), \text{ as }p\to \infty,  \label{eq-theorem-2-10}
\end{eqnarray*}
which together with Lemma \ref{lemma-2} and \eqref{eq-theorem-2-5} leads to 
\begin{eqnarray*}
    \|\FF_{\widetilde \theta^\omega} - \FF_{\theta^\omega}\|_{F^{*\omega}, 2}^2 &=& \int \left\{\FF_{\widetilde \theta^\omega}(z) - \FF_{\theta^\omega}(z)\right\}^2 d\left\{ F^{*\omega}(z) - \FF^\omega(z)\right\} + \|\FF_{\widetilde \theta^\omega} - \FF_{\theta^\omega}\|_{\FF^\omega, 2}^2 \\
    &=& o(1), \text{ as } p\to \infty, \label{eq-theorem-2-11}
\end{eqnarray*}
or equivalently 
\begin{eqnarray*}
    \|\FF_{\widetilde \theta^\omega} - \FF_{\theta^\omega}\|_{F^{*\omega}, 2} = 0. \label{eq-theorem-2-12}
\end{eqnarray*}
This together with Condition \ref{Condition-4} concludes for every $m_1, m_2 \in \{0,1\}; k = 1,\ldots, K$, 
\begin{eqnarray}
\widetilde \lambda_0^\omega = \lambda_0^\omega \quad \text{and} \quad  \widetilde F_{m_1,k}^\omega(\cdot) = F_{m_1,k}^\omega(\cdot) \quad a.s. \text{ in }F_{m_2,k}^{*\omega}(\cdot). \label{eq-theorem-2-13}
\end{eqnarray}

Based on the definition of $\widetilde \lambda_0^\omega$ and $\lambda_0^\omega$, we immediately have 
\begin{eqnarray*}
    \widehat \lambda_0^\omega = \widetilde \lambda_0^\omega + o(1) = \lambda_0^\omega + o(1) = p_0/p + o(1),
\end{eqnarray*}
which verifies Part (a) of the theorem. 

Furthermore, based on Condition \ref{Condition-2}, we have, for $m=0,1$, 
\begin{eqnarray*}
\sup_{z} |\FF_m(z) - F_m^*(z)| = o(1), \quad a.s., \label{eq-theorem-2-14}
\end{eqnarray*}
which implies for every $m=0,1; k=1,\ldots, K$,
\begin{eqnarray}
    \sup_{z_k \in \mathbb{R}} |\FF_{m,k}(z_k) - F_{m,k}^*(z_k)| &=& \sup_{z_k \in \mathbb{R}} |\FF_m(\infty, \ldots, z_k, \ldots, \infty) - F_m^*(\infty, \ldots, z_k, \ldots, \infty)|\nonumber \\
    &\leq & \sup_{z} |\FF_m(z) - F_m^*(z)| = o(1), \quad a.s. \label{eq-theorem-2-15}
\end{eqnarray}

For $m_1, m_2 \in \{0,1\}$, $k = 1,\ldots, K$, we have
\begin{eqnarray}
    &&\int \left\{ \widehat \FF_{m_1,k}^\omega(z_k) - \FF_{m_1,k}^\omega(z_k) \right\}^2 d \FF_{m_2,k}^\omega(z_k) \nonumber \\
    &=& \int \left\{ \widehat \FF_{m_1,k}^\omega(z_k) - \FF_{m_1,k}^\omega(z_k) \right\}^2 d \left\{\FF_{m_2,k}^\omega(z_k) - F_{m_2,k}^{*\omega}(z_k) \right\} \nonumber \\
    && + \int \left\{ \widehat \FF_{m_1,k}^\omega(z_k) - \FF_{m_1,k}^\omega(z_k) \right\}^2 d F_{m_2,k}^{*\omega}(z_k).  \label{eq-theorem-2-16}
\end{eqnarray}
Consider the two terms on the right hand of \eqref{eq-theorem-2-16}. The first term is $o(1)$, based on \eqref{eq-theorem-2-15} and Lemma \ref{lemma-2}; the second term  is also $o(1)$, based on \eqref{eq-theorem-2-13}, the definition of $\widetilde F_{m_1,k}^\omega(\cdot)$ $F_{m_1,k}^\omega(\cdot)$, and by applying the dominant convergence theorem. As a consequence
\begin{eqnarray*}
    \int \left\{ \widehat \FF_{m_1,k}^\omega(z_k) - \FF_{m_1,k}^\omega(z_k) \right\}^2 d \FF_{m_2,k}^\omega(z_k) = o(1), 
\end{eqnarray*}
which together with the fact that $\omega \in \Omega$ is arbitrary and $P(\Omega) = 1$ verifies \eqref{marginal L2 with mixture}. We complete the proof of this theorem. 

\subsection{Proof of Theorem \ref{median fdr theorem}}

We first establish the convergence of the empirical distribution based on $\{M_i: i\in \mathH_0\}$ and $\{M_i: i\in \mathH_1\}$, where 
\begin{eqnarray*}
    M_i = \text{median}\{\FF_{0,1}(Z_{i,1}), \ldots, \FF_{0,K}(Z_{i,K})\}. 
\end{eqnarray*}
The intuition is that $M_i$ mimics the true value of $\widehat M_i$; thus, the corresponding empirical distributions mimic those of $\widehat M_i$. We have the following lemmas.

\begin{lemma} \label{key lemma in median approcah}
    Assume Condition \ref{Condition-1}. We have  
    \begin{eqnarray}
        \sup_{t\in \mathbb{R}} \left|M(t)-\frac{1}{p_0}\sum\limits_{i \in \mathcal{H}_0}I(M_i\leq t)\right|=o(1),\,a.s., \label{eq-key lemma in median approcah-1}
    \end{eqnarray}  
    where we recall that $M(\cdot)$ is the c.d.f. of the median of $K$ i.i.d. $\mathrm{U}[0,1]$ random variables. 
\end{lemma}
\proof Without loss of generality, we assume that $K$ is an odd number and thus we write $K= 2K_1+1$; the case that $K$ is an even number can be developed similarly. Let 
\begin{eqnarray*}
    \widetilde U = \text{median}\{U_1, \ldots, U_K\}, \quad \text{with } U_i\sim_{i.i.d.} \text{Uniform}[0,1]. \label{eq-key lemma in median approcah-2}
\end{eqnarray*}
Then,
\begin{eqnarray}
    M(t) &=& P(\widetilde U \le t) = P\left( \cup_{k=K_1+1}^K\cup_{S\subset \{1,\ldots, K\}; |S| = k} \left\{\cap_{j\in S} \{U_j\leq t\} \cap_{j\notin S}\{U_j>t\} \right\} \right) \nonumber \\
    &=& \sum_{k=K_1+1}^K \sum_{|S|=k} \left\{\prod_{j\in S}P(U_j\leq t) \prod_{j\notin S}P(U_j>t)\right\} \nonumber \\
    &=& \sum_{k=K_1+1}^K \sum_{|S|=k} \left\{\prod_{j\in S} t\prod_{j\notin S} (1-t)\right\}. \label{eq-key lemma in median approcah-3}
\end{eqnarray}
For $k=1,\ldots, K$, denote
\begin{eqnarray}
    H_k(t) = \frac{1}{p_0} \sum_{i\in \mathH_0} I\left\{\FF_{0,k}(Z_{i,k})\leq t\right\} = \frac{1}{p_0} \sum_{i=1}^{p_0} I\left\{\FF_{0,k}(X_{i,k})\leq t) \right\}. \label{eq-key lemma in median approcah-4}
\end{eqnarray}
The term on the left hand side of \eqref{eq-key lemma in median approcah-1} satisfies
\begin{eqnarray}
    \sup_{t\in \mathbb{R}}\left|M(t)-\frac{1}{p_0}\sum\limits_{i \in \mathcal{H}_0}I(M_i\leq t)\right| \leq \mathJ_1 +\mathJ_2, \label{eq-key lemma in median approcah-5}
\end{eqnarray}
where
\begin{eqnarray}
    \mathJ_1 &=& \sup\limits_{t\in \mathbb{R}}\left|M(t)-\sum\limits_{k=K_1+1}^K \sum\limits_{|S|=k}\left[\prod\limits_{j \in S}H_j(t)\prod\limits_{j\notin S}\{1-H_j(t)\}\right]\right| \label{eq-key lemma in median approcah-6-1} \\
    \mathJ_2 &=& \sup\limits_{t}\left|\sum\limits_{k=K_1+1}^K \sum\limits_{|S|=k}\left[\prod\limits_{j \in S}H_j(t)\prod\limits_{j\notin S}\{1-H_j(t)\}\right]-\frac{1}{p_0}\sum\limits_{i \in \mathcal{H}_0}I(M_i\leq t) \right|.  \label{eq-key lemma in median approcah-6}  
\end{eqnarray}
Next, we consider $\mathJ_1$ and $\mathJ_2$ separately. For $\mathJ_1$, note that $\FF_{0,k}(\cdot)$ is the empirical c.d.f. of $\{X_{i,k},i=1,\ldots, p_0\}$, thus, 
\begin{eqnarray*}
    \{\FF_{0,k}(X_{1,k}),\ldots, \FF_{0,k}(X_{p_0,k}) \} = \left\{\frac{1}{p_0},\ldots, \frac{p_0}{p_0} \right\}, \label{eq-key lemma in median approcah-7}
\end{eqnarray*}
which together with the definition of $H_k(t)$ given in \eqref{eq-key lemma in median approcah-4}, $H_k(t)$ is the empirical c.d.f. of the set $\{1/p_0,\ldots, p_0/p_0\}$. We have for every $t\in [0,1]$, 
\begin{eqnarray*}
    H_k(t) = \frac{\lfloor p_0 t\rfloor}{p_0},  \label{eq-key lemma in median approcah-8}
\end{eqnarray*}
which together with \eqref{eq-key lemma in median approcah-3} and \eqref{eq-key lemma in median approcah-6-1}, and applying Lemma \ref{lipschitz prod lemma}, leads to 
\begin{eqnarray}
    \mathJ_1 &=& \sup\limits_{t\in \mathbb{R}}\left|\sum\limits_{k=K_1+1}^K \sum\limits_{|S|=k}\left[\prod_{j\in S} t\prod_{j\notin S} (1-t) - \prod\limits_{j \in S}H_j(t)\prod\limits_{j\notin S}\{1-H_j(t)\}\right]\right| \nonumber \\
    &\leq& \sup\limits_{t\in \mathbb{R}}\sum\limits_{k=K_1+1}^K \sum\limits_{|S|=k} \left|\prod_{j\in S} t\prod_{j\notin S} (1-t) - \prod\limits_{j \in S}H_j(t)\prod\limits_{j\notin S}\{1-H_j(t)\}\right| \nonumber \\
    &\leq& \sup\limits_{t\in \mathbb{R}}\sum\limits_{k=K_1+1}^K \sum\limits_{|S|=k} \sum_{j=1}^K|t- H_j(t)| \nonumber \\
    &=& \sup\limits_{t\in \mathbb{R}}\sum\limits_{k=K_1+1}^K \sum\limits_{|S|=k} \sum_{j=1}^K \left|t -  \frac{\lfloor p_0 t\rfloor}{p_0}\right| \nonumber \\
    &\leq & \sum\limits_{k=K_1+1}^K \sum\limits_{|S|=k} \sum_{j=1}^K \frac{1}{p_0} = o(1), \quad a.s. \label{eq-key lemma in median approcah-9}
\end{eqnarray}
We proceed to consider $\mathJ_2$. We can write 
\begin{align*} \label{eq-key lemma in median approcah-10}
    \frac{1}{p_0}\sum\limits_{i \in H_0}I(M_i\leq t)&=\frac{1}{p_0}\sum\limits_{i \in \mathcal{H}_0}\sum\limits_{k=K_1+1}^K\sum\limits_{|S|=k}\prod\limits_{j\in S}I\{\FF_{0,j}(Z_{i,j})\leq t \}\prod_{j \notin S}I\{\FF_{0,j}(Z_{i,j})>t\}\\
    &=\sum\limits_{k=K_1+1}^K \sum\limits_{|S|=k}\left[\frac{1}{p_0}\sum\limits_{i=1}^{p_0}\prod\limits_{j\in S}I\{\FF_{0,j}(X_{i,j})\leq t \}\prod_{j \notin S}I\{\FF_{0,j}(X_{i,j})>t\}\right], 
\end{align*}
which, together with the definition of $\mathJ_2$ given by \eqref{eq-key lemma in median approcah-6}, results in
\begin{eqnarray}
    \mathJ_2 \leq \sum\limits_{k=K_1+1}^K \sum\limits_{|S|=k} \sup_{t\in \mathbb{R}} \mathJ_{2,S}(t), \label{eq-key lemma in median approcah-11}
\end{eqnarray}
where
\begin{eqnarray}
    \mathJ_{2,S}(t) = \left|\prod\limits_{j \in S}H_j(t)\prod\limits_{j\notin S}\{1-H_j(t)\} - \frac{1}{p_0}\sum\limits_{i=1}^{p_0}\prod\limits_{j\in S}I\{\FF_{0,j}(X_{i,j})\leq t \}\prod_{j \notin S}I\{\FF_{0,j}(X_{i,j})>t\} \right|.  \label{eq-key lemma in median approcah-12}
\end{eqnarray}
On the other hand, denote by $X_{(1),k} \leq \ldots \leq X_{(p_0),k}$ the order statistics of $X_{1,k}, \ldots, X_{p_0,k}$.
Then, for every $i=1,\ldots, p_0$ and $t\in[0,1]$, 
``$\FF_{0,k}(X_{i,k})\leq t$" if and only if  ``$X_{i,k} \leq X_{(\lfloor tp_0\rfloor), k}$", which implies
\begin{eqnarray}
    I\left\{\FF_{0,k}(X_{i,k})\leq t\right\} = I\left\{X_{i,k} \leq X_{(\lfloor tp_0\rfloor), k} \right\}.  \label{eq-key lemma in median approcah-13}
\end{eqnarray}
Based on \eqref{eq-key lemma in median approcah-13} and referring to the definition of $H_k(t)$ defined by \eqref{eq-key lemma in median approcah-4}, we have
\begin{eqnarray}
    H_k(t) = \frac{1}{p_0} \sum_{i=1}^{p_0} I\left\{\FF_{0,k}(X_{i,k})\leq t) \right\} = \frac{1}{p_0} \sum_{i=1}^{p_0} I\left\{ X_{i,k} \leq X_{(\lfloor tp_0\rfloor), k} \right\} = \FF_{0,k}\left(X_{(\lfloor tp_0\rfloor), k}\right). \label{eq-key lemma in median approcah-14}
\end{eqnarray}
Combining \eqref{eq-key lemma in median approcah-12}, \eqref{eq-key lemma in median approcah-13}, and \eqref{eq-key lemma in median approcah-14}, we have 
\begin{eqnarray*}
    \mathJ_{2,S}(t) &=& \left|\prod\limits_{j \in S}\FF_{0,j}\left(X_{(\lfloor tp_0\rfloor), j}\right) \prod\limits_{j\notin S}\left\{1-\FF_{0,j}\left(X_{(\lfloor tp_0\rfloor), j}\right) \right\} \right.\\
    &&- \left. \frac{1}{p_0}\sum\limits_{i=1}^{p_0}\prod\limits_{j\in S}I\{X_{i,j} \leq X_{(\lfloor tp_0\rfloor), j} \}\prod_{j \notin S}I\{X_{i,j} > X_{(\lfloor tp_0\rfloor), j}\} \right|. \label{eq-key lemma in median approcah-15}
\end{eqnarray*}
Note that $S$ is a subset of $\{1,\ldots, K\}$, and refer to the partition \( \mathbb{R}^K = \cup_{h=1}^Q A_h(t)\) given in Section \ref{section-notations}, consider the $h_S$, such that $\Delta_{h_S} = (\delta_1, \ldots, \delta_K)$ with $\delta_k= I(k\in S), k = 1,\ldots, K$. It is straightforward to check that 
\begin{eqnarray*}
    \mathJ_{2,S}(t) = \left| \FF_{0,\theta_0,h_S}(X_{(\lfloor tp_0\rfloor)}) - \FF_{0,h_S}(X_{(\lfloor tp_0\rfloor)}) \right|, \label{eq-key lemma in median approcah-16}
\end{eqnarray*}
where $X_{(\lfloor tp_0\rfloor)} = \left( X_{(\lfloor tp_0\rfloor),1}, \ldots, X_{(\lfloor tp_0\rfloor),K} \right)^T$
As a consequence, based on Condition \ref{Condition-1} and Lemma \ref{lemma-1}, we have
\begin{eqnarray}
\sup_{t\in \mathbb{R}}  \mathJ_{2,S}(t) &=& \sup_{t\in \mathbb{R}} \left| \FF_{0,\theta_0,h_S}(X_{(\lfloor tp_0\rfloor)}) - \FF_{0,h_S}(X_{(\lfloor tp_0\rfloor)}) \right| \nonumber \\
&\leq& \sup_{z\in \mathbb{R}^K} \left| \FF_{0,\theta_0,h_S}(z) - \FF_{0,h_S}(z) \right| = o(1), \quad a.s. \label{eq-key lemma in median approcah-17}
\end{eqnarray} 
Combining \eqref{eq-key lemma in median approcah-11} and \eqref{eq-key lemma in median approcah-17}, we have 
\begin{eqnarray*}
    \mathJ_2 = o(1), \quad a.s., 
\end{eqnarray*}
which further combined with \eqref{eq-key lemma in median approcah-5} and \eqref{eq-key lemma in median approcah-9} verifies \eqref{eq-key lemma in median approcah-1}. We complete the proof of this lemma. \epf

\begin{lemma} \label{key lemma in median approcah 1}
    Assume Condition \ref{Condition-1}. We have
\begin{eqnarray*}
    \sup_{t\in [0,1]} \left|\frac{1}{p_1}\sum\limits_{i \in \mathH_1}I(M_i\leq t) - G(t)\right| = o(1), \quad a.s.,
\end{eqnarray*}
where 
\begin{eqnarray}
    G(t) &=& \sum\limits_{k=K_1+1}^K\sum\limits_{|S|=k} \prod_{j\in S} G_j(t) \prod_{j \notin S}\left\{1-G_j(t) \right\} \nonumber \\
    G_j(t) &=& \frac{1}{p_1} \sum_{i\in \mathH_1} I\left\{\FF_{0,j}(Z_{i,j}) \leq t\right\}, \quad \text{for } j=1,\ldots, K. \label{eq-def-G-t}
\end{eqnarray}
\end{lemma}
\proof Let $X_{(1),k} \leq \ldots \leq X_{(p_0),k}$ be the order statistics of $X_{1,k}, \ldots, X_{p_0,k}$; likewise, let $Y_{(1),k} \leq \ldots \leq Y_{(p_1),k}$ be the order statistics of $Y_{1,k}, \ldots, Y_{p_1,k}$. For each $t\in [0,1]$ and $k=1,\ldots, K$, define
\begin{eqnarray*}
    s_{t,k} = \max\left\{i: Y_{(i),k} < X_{(\lfloor p_0t\rfloor + 1), k}\right\}, \label{eq-key lemma in median approcah 1-1}
\end{eqnarray*}
which indicates 
\begin{eqnarray*}
Y_{(1),k} \leq \ldots \leq Y_{(s_{t,k}), k}< X_{(\lfloor p_0t\rfloor + 1), k} \leq Y_{(s_{t,k}+1), k}\leq \ldots \leq Y_{(p_1), k}. \label{eq-key lemma in median approcah 1-2}
\end{eqnarray*}
Based on the definition of $\FF_{0,k}(\cdot)$, we have, for any $t\in[0,1]$, 
\begin{eqnarray}
    \FF_{0,k}(Y_{i,k}) \leq t &\Leftrightarrow& \FF_{0,k}(Y_{i,k}) \leq \frac{\lfloor p_0t \rfloor}{p_0} \Leftrightarrow Y_{i,k} < X_{(\lfloor p_0t \rfloor +1), k} \nonumber \\ &\Leftrightarrow& Y_{i,k} \leq Y_{(s_{t,k}), k}. \label{eq-key lemma in median approcah 1-3}
\end{eqnarray}
Referring to $G_j(t), j=1,\ldots, K$ defined in \eqref{eq-def-G-t},  we have 
\begin{eqnarray*}
    G_j(t) &=& \frac{1}{p_1} \sum_{i\in \mathH_1} I\left\{\FF_{0,k}(Z_{i,j}) \leq t \right\} = \frac{1}{p_1} \sum_{i=1}^{p_1} I\left\{\FF_{0,j}(Y_{i,j}) \leq t \right\}\\
    &=& \frac{1}{p_1} \sum_{i=1}^{p_1} I\left\{ Y_{i,j} \leq Y_{(s_{t,j}), j}\right\} = \FF_{1,j}\left(Y_{(s_{t,j}), j}\right) \label{eq-key lemma in median approcah 1-4}
\end{eqnarray*}
Similarly to the proof of Lemma \ref{key lemma in median approcah}, we assume that $K = 2K_1 + 1$ is an odd number, and $S$ denotes a subset of $\{1,\ldots, K\}$. We have 
\begin{eqnarray}
    G(t) &=& \sum\limits_{k=K_1+1}^K\sum\limits_{|S|=k} \prod_{j\in S} G_j(t) \prod_{j \notin S}\left\{1-G_j(t) \right\} \nonumber \\
    &=& \sum\limits_{k=K_1+1}^K\sum\limits_{|S|=k} \prod_{j\in S} \FF_{1,j}\left(Y_{(s_{t,j}), j}\right) \prod_{j \notin S}\left\{1-\FF_{1,j}\left(Y_{(s_{t,j}), j}\right) \right\} \nonumber \\
    &=& \sum\limits_{k=K_1+1}^K\sum\limits_{|S|=k}\FF_{1,\theta_0,h_S}(Y_{(s_t)}),
    \label{eq-key lemma in median approcah 1-5}
\end{eqnarray}
where $Y_{(s_t)} = (Y_{(s_{t,1}), 1}, \ldots, Y_{(s_{t,K}), K})^T$, ``$\FF_{1,\theta_0,h}$" is defined in \eqref{F-F-h-relation}, and $h_S$ refers to the partition defined in Section~\ref{section-notations}, such that $\Delta_{h_S} = (\delta_1, \ldots, \delta_K)$ with $\delta_k = I(k \in S)$ for $k = 1, \ldots, K$.
Furthermore, based on \eqref{eq-key lemma in median approcah 1-3}, we have 
\begin{eqnarray}
    \nonumber \frac{1}{p_1}\sum\limits_{i \in \mathH_1}I(M_i\leq t)&=&\frac{1}{p_1}\sum\limits_{i \in \mathcal{H}_1}\sum\limits_{k=K_1 +1}^K\sum\limits_{|S|=k}\prod\limits_{j\in S}I(\FF_{0,j}(Z_{i,j})\leq t )\prod_{j \notin S}I(\FF_{0,j}(Z_{i,j})>t)\nonumber \\
    &=&\sum\limits_{k=K_1 + 1}^K\sum\limits_{|S|=k}\left\{\frac{1}{p_1}\sum\limits_{i=1}^{p_1}\prod\limits_{j\in S}I(\FF_{0,j}(Y_{i,j})\leq t )\prod_{j \notin S}I(\FF_{0,j}(Y_{i,j})>t)\right\} \nonumber \\
    &=& \sum\limits_{k=K_1 + 1}^K\sum\limits_{|S|=k}\left\{\frac{1}{p_1}\sum\limits_{i=1}^{p_1}\prod\limits_{j\in S}I(Y_{i,j} \leq Y_{(s_{t,j}), j} )\prod_{j \notin S}I(Y_{i,j} > Y_{(s_{t,j}), j})\right\} \nonumber \\
    &=& \sum\limits_{k=K_1 + 1}^K\sum\limits_{|S|=k} \FF_{1,h_S}(Y_{(s_t)}), \label{eq-key lemma in median approcah 1-6}
\end{eqnarray}
with ``$\FF_{1,h}$" defined in \eqref{F-F-h-relation}. Combining \eqref{eq-key lemma in median approcah 1-5} and \eqref{eq-key lemma in median approcah 1-6}, we have 
\begin{eqnarray*}
    \sup_{t\in [0,1]} \left|\frac{1}{p_1}\sum\limits_{i \in \mathH_1}I(M_i\leq t) - G(t)\right| &\leq& \sum\limits_{k=K_1 + 1}^K\sum\limits_{|S|=k} \sup_{t\in[0,1]} \left| \FF_{1,h_S}(Y_{(s_t)}) - \FF_{1,\theta_0, h_S}(Y_{(s_t)}) \right|\\
    &\leq & \sum\limits_{k=K_1 + 1}^K\sum\limits_{|S|=k} \sup_{z\in \mathbb{R}^K} \left| \FF_{1,h_S}(z) - \FF_{1,\theta_0, h_S}(z) \right|\\
    &=& o(1), \quad a.s.,
\end{eqnarray*}
based on Condition \ref{Condition-1} and Lemma \ref{lemma-1}. We complete the proof of this lemma. \epf

We proceed to show Theorem \ref{median fdr theorem}. Referring to $\widehat \FDP(t)$, $\FDP(t)$, and $t_\alpha$ respectively defined in \eqref{eq-def-hat-FDP-t-median}, \eqref{eq-def-FDP-t-median}, and \eqref{eq-t-alpha-median}, we have,
\begin{eqnarray}
    \FDP(t_\alpha) \leq  \widehat \FDP(t_\alpha) + \left| \widehat \FDP(t_\alpha) - \FDP(t_\alpha) \right| \leq \alpha + \left| \widehat \FDP(t_\alpha) - \FDP(t_\alpha) \right|, \label{eq-median fdr theorem-1}
\end{eqnarray}
and based on Part (a) of Theorem \ref{theorem-2}, for any $t\in[0,1]$, we have 
\begin{eqnarray}
    && \widehat \FDP(t) - \FDP(t) = \frac{\widehat \lambda_0\cdot M(t) - \frac{1}{p} \sum_{i\in \mathH_0} I(\widehat M_i \leq t)}{\frac{1}{p}\cdot \max\left\{\sum\limits_{i=1}^pI(\widehat M_i\leq t),1\right\}} \nonumber\\
    &=& \frac{(\widehat \lambda_0 - \lambda_{0,0})\cdot M(t) -  \lambda_{0,0}\left\{\frac{1}{p_0} \sum_{i\in \mathH_0} I(\widehat M_i \leq t) - M(t)\right\}}{\frac{1}{p}\cdot \max\left\{\sum\limits_{i=1}^pI(\widehat M_i\leq t),1\right\}} \nonumber \\
    &=& \frac{o(1) -  \lambda_{0,0}\left\{\frac{1}{p_0} \sum_{i\in \mathH_0} I(\widehat M_i \leq t) - M(t)\right\}}{\frac{1}{p}\cdot \max\left\{\sum\limits_{i=1}^pI(\widehat M_i\leq t),1\right\}}, \quad a.s. \label{eq-median fdr theorem-2}
\end{eqnarray}
Furthermore, based on Lemma \ref{key lemma in median approcah}, we have 
\begin{eqnarray}
    && \sup_{t\in \mathbb{R}}\left|\frac{1}{p_0} \sum_{i\in \mathH_0} I(\widehat M_i \leq t) - M(t)\right| \nonumber \\ &\leq&   \sup_{t\in \mathbb{R}} \left|M(t)-\frac{1}{p_0}\sum\limits_{i \in \mathcal{H}_0}I(M_i\leq t)\right| + \sup_{t \in \mathbb{R}} \left| \frac{1}{p_0}\sum\limits_{i \in \mathcal{H}_0}I(M_i\leq t) - \frac{1}{p_0} \sum_{i\in \mathH_0} I(\widehat M_i \leq t) \right| \nonumber \\
    &=& \sup_{t \in \mathbb{R}} \left| \mathJ_3(t) \right| + o(1), \quad a.s. \label{eq-median fdr theorem-3}
\end{eqnarray}
where
\begin{eqnarray}
    \mathJ_3(t) = \frac{1}{p_0}\sum\limits_{i \in \mathcal{H}_0}I(M_i\leq t) - \frac{1}{p_0} \sum_{i\in \mathH_0} I(\widehat M_i \leq t). \label{eq-median fdr theorem-4}
\end{eqnarray}
Next, we derive the asymptotic property of $\sup_{t\in \mathbb{R}} |\mathJ_3(t)|$. For any $\epsilon >0$, let
\begin{eqnarray*}
    \delta_{p, k,\epsilon}=\frac{1}{p_0}\sum\limits_{i\in \mathH_0} I\left(\left|\widehat{\FF}_{0,k}(Z_{i,k})-\FF_{0,k}(Z_{i,k})\right|>\epsilon\right). \label{eq-median fdr theorem-5}
\end{eqnarray*}
Based on Part (b) of Theorem \ref{theorem-2} and Chebyshev's inequality, we have 
\begin{eqnarray}
    \delta_{p, k,\epsilon} & = & \frac{1}{p_0}\sum\limits_{i\in \mathH_0} I\left(\left|\widehat{\FF}_{0,k}(Z_{i,k})-\FF_{0,k}(Z_{i,k})\right|>\epsilon\right) \nonumber \\
    & = & P\left( \left|\widehat{\FF}_{0,k}(X)-\FF_{0,k}(X)\right| > \epsilon \right), \quad X \sim \FF_{0,k}(\cdot) \nonumber \\
    &\leq & \frac{E\left\{ \widehat{\FF}_{0,k}(X)-\FF_{0,k}(X)\right\}^2}{\epsilon^2} \nonumber \\
    &= & \frac{1}{\epsilon^2} \int \left\{ \widehat{\FF}_{0,k}(x)-\FF_{0,k}(x) \right\}^2 d \FF_{0,k}(x) = o(1), \quad a.s. \label{eq-median fdr theorem-6}
\end{eqnarray}
where ``$P$" and ``$E$" are computed conditioned on the data, in the sense $\widehat \FF_{0,k}(\cdot)$ and $\FF_{0,k}(\cdot)$ are treated as non-random functions.  Denote 
\begin{eqnarray}
    \mathS_{\epsilon, k} = \left\{i\in \mathH_0: \left| \widehat{\FF}_{0,k}(Z_{i,k})-\FF_{0,k}(Z_{i,k}) \right| > \epsilon \right\} \subset \mathH_0, \label{eq-median fdr theorem-7}
\end{eqnarray}
and let $\mathS_\epsilon = \cup_{k=1}^K \mathS_{\epsilon, k}$. Then $\mathS_\epsilon \subset \mathH_0$, and based on \eqref{eq-median fdr theorem-6} and \eqref{eq-median fdr theorem-7}, we have 
\begin{eqnarray}
\frac{|\mathS_{\epsilon}|}{p_0} \leq \frac{\sum_{k=1}^K |\mathS_{\epsilon, k}|}{p_0} = \sum_{k=1}^K \delta_{p, k,\epsilon} = o(1), \quad a.s. \label{eq-median fdr theorem-8}
\end{eqnarray}
Referring to \eqref{eq-median fdr theorem-4}, we have 
\begin{eqnarray}
    |\mathJ_3(t)| &\leq& \left| \frac{1}{p_0}\sum\limits_{i \in \mathS_\epsilon}I(M_i\leq t) - \frac{1}{p_0} \sum_{i\in \mathS_\epsilon} I(\widehat M_i \leq t) \right| \nonumber\\
    && + \left| \frac{1}{p_0}\sum\limits_{i \in \mathH_0 \setminus \mathS_\epsilon}I(M_i\leq t) - \frac{1}{p_0} \sum_{i\in \mathH_0 \setminus \mathS_\epsilon} I(\widehat M_i \leq t) \right| \nonumber \\
    &\leq & \frac{2|\mathS_\epsilon|}{p_0} + \left| \frac{1}{p_0}\sum\limits_{i \in \mathH_0 \setminus \mathS_\epsilon}I(M_i\leq t) - \frac{1}{p_0} \sum_{i\in \mathH_0 \setminus \mathS_\epsilon} I(\widehat M_i \leq t) \right| \nonumber  \\
    &=& o(1) + \mathJ_{3,1}(t), \quad a.s. \label{eq-median fdr theorem-9}
\end{eqnarray}
where $o(1)$ is uniform in $t\in [0,1]$, 
\begin{eqnarray}
    \mathJ_{3,1}(t) = \left| \frac{1}{p_0}\sum\limits_{i \in \mathH_0 \setminus \mathS_\epsilon}I(M_i\leq t) - \frac{1}{p_0} \sum_{i\in \mathH_0 \setminus \mathS_\epsilon} I(\widehat M_i \leq t) \right|.  \label{eq-median fdr theorem-10}
\end{eqnarray}
For $\mathJ_{3,1}(t)$, based on the definition of $\mathS_\epsilon$, for every $i \in \mathH_0 \setminus \mathS_\epsilon$,  
\begin{eqnarray*}
\left| \widehat \FF_{0,k}(Z_{i,k})  - \FF_{0,k}(Z_{i,k})\right| \leq \epsilon, \quad \text{for every }k=1,\ldots, K, \label{eq-median fdr theorem-11}
\end{eqnarray*}
and thus based on Lemma \ref{simple property for median of two sequence}, we have
\begin{eqnarray*}
    \left|\widehat M_i - M_i \right| \leq \epsilon, \label{eq-median fdr theorem-12}
\end{eqnarray*}
which further implies 
\begin{eqnarray}
    \frac{1}{p_0}\sum\limits_{i \in \mathcal{H}_0\setminus \mathS_\epsilon}I(M_i\leq t-\epsilon)\leq\frac{1}{p_0}\sum\limits_{i \in \mathcal{H}_0\setminus \mathS_\epsilon}I(\hat{M}_{i}\leq t)\leq \frac{1}{p_0}\sum\limits_{i \in \mathcal{H}_0\setminus \mathS_\epsilon}I(M_i\leq t+\epsilon). \label{eq-median fdr theorem-13}
\end{eqnarray}
Clearly
\begin{eqnarray}
    \frac{1}{p_0}\sum\limits_{i \in \mathcal{H}_0\setminus \mathS_\epsilon}I(M_i\leq t-\epsilon)\leq\frac{1}{p_0}\sum\limits_{i \in \mathcal{H}_0\setminus \mathS_\epsilon}I({M}_{i}\leq t)\leq \frac{1}{p_0}\sum\limits_{i \in \mathcal{H}_0\setminus \mathS_\epsilon}I(M_i\leq t+\epsilon). \label{eq-median fdr theorem-14}
\end{eqnarray}
Combining \eqref{eq-median fdr theorem-10}, \eqref{eq-median fdr theorem-13} and \eqref{eq-median fdr theorem-14}, we have 
\begin{eqnarray}
    \mathJ_{3,1}(t) &\leq& \left|\frac{1}{p_0}\sum\limits_{i \in \mathcal{H}_0\setminus \mathS_\epsilon}I(M_i\leq t+\epsilon)- \frac{1}{p_0}\sum\limits_{i \in \mathcal{H}_0\setminus \mathS_\epsilon}I(M_i\leq t-\epsilon) \right| \nonumber \\
    & = & \frac{1}{p_0} \sum_{i\in \mathcal{H}_0\setminus \mathS_\epsilon} I\left\{ M_i \in (t-\epsilon, t+\epsilon] \right\}  = o(1) + M(t+\epsilon) - M(t-\epsilon), \quad a.s. \label{eq-median fdr theorem-15} 
\end{eqnarray}
based on \eqref{eq-median fdr theorem-8} and Lemma \ref{key lemma in median approcah}, where $o(1)$ is uniform in $t\in[0,1]$.  Combining \eqref{eq-median fdr theorem-3}, \eqref{eq-median fdr theorem-9}, \eqref{eq-median fdr theorem-15}, and based on the fact that $\epsilon>0$ is arbitrary,  leads to 
\begin{eqnarray}
    \sup_{t\in \mathbb{R}}\left|\frac{1}{p_0} \sum_{i\in \mathH_0} I(\widehat M_i \leq t) - M(t)\right| = o(1), \quad a.s., \label{eq-median fdr theorem-16} 
\end{eqnarray}
which establishes the asymptotic property of the numerator of on the far right hand side of \eqref{eq-median fdr theorem-2}. We proceed to consider the asymptotic property of its denominator.  First, we consider the asymptotic property of 
\begin{eqnarray}
\mathJ_4(t) = \frac{1}{p_1} \sum_{i\in \mathH_1}I(\widehat M_i\leq t) - \frac{1}{p_1} \sum_{i\in \mathH_1} I(M_i \leq t) \label{eq-median fdr theorem-17} 
\end{eqnarray}
For any $\epsilon >0$, let
\begin{eqnarray*}
    \widetilde \delta_{p, k,\epsilon}=\frac{1}{p_1}\sum\limits_{i\in \mathH_1} I\left(\left|\widehat{\FF}_{0,k}(Z_{i,k})-\FF_{0,k}(Z_{i,k})\right|>\epsilon\right).  \label{eq-median fdr theorem-18} 
\end{eqnarray*}
Based on Part (b) of Theorem \ref{theorem-2} and Chebyshev's inequality, we have
\begin{eqnarray}
    \widetilde \delta_{p, k,\epsilon} & = & \frac{1}{p_1}\sum\limits_{i\in \mathH_1} I\left(\left|\widehat{\FF}_{0,k}(Z_{i,k})-\FF_{0,k}(Z_{i,k})\right|>\epsilon\right) \nonumber \\
    & = & P\left( \left|\widehat{\FF}_{0,k}(Y)-\FF_{0,k}(Y)\right| > \epsilon \right), \quad Y \sim \FF_{1,k}(\cdot) \nonumber \\
    &\leq & \frac{E\left\{ \widehat{\FF}_{0,k}(Y)-\FF_{0,k}(Y)\right\}^2}{\epsilon^2} \nonumber \\
    &= & \frac{1}{\epsilon^2} \int \left\{ \widehat{\FF}_{0,k}(y)-\FF_{0,k}(y) \right\}^2 d \FF_{1,k}(y) = o(1), \quad a.s. \label{eq-median fdr theorem-19} 
\end{eqnarray}
where ``$P$" and ``$E$" are computed conditioned on the data, in the sense $\widehat \FF_{0,k}(\cdot)$ and $\FF_{0,k}(\cdot)$ are treated as non-random functions. Denote 
\begin{eqnarray}
    \widetilde \mathS_{\epsilon, k} = \left\{i\in \mathH_1: \left| \widehat{\FF}_{0,k}(Z_{i,k})-\FF_{0,k}(Z_{i,k}) \right| > \epsilon \right\} \subset \mathH_1, \label{eq-median fdr theorem-20} 
\end{eqnarray}
and let $\widetilde \mathS_\epsilon = \cup_{k=1}^K \widetilde \mathS_{\epsilon, k}$. Then $\widetilde \mathS_\epsilon \subset \mathH_1$, and based on \eqref{eq-median fdr theorem-19} and \eqref{eq-median fdr theorem-20}, we have 
\begin{eqnarray}
\frac{|\widetilde \mathS_{\epsilon}|}{p_1} \leq \frac{\sum_{k=1}^K |\widetilde \mathS_{\epsilon, k}|}{p_1} = \sum_{k=1}^K \widetilde \delta_{p, k,\epsilon} = o(1), \quad a.s. \label{eq-median fdr theorem-21} 
\end{eqnarray}
Referring to the definition of $\mathJ_4(t)$ in \eqref{eq-median fdr theorem-17}, we have 
\begin{eqnarray}
    |\mathJ_4(t)| &\leq& \left| \frac{1}{p_1}\sum\limits_{i \in \widetilde \mathS_\epsilon}I(M_i\leq t) - \frac{1}{p_1} \sum_{i\in \widetilde \mathS_\epsilon} I(\widehat M_i \leq t) \right| \nonumber\\
    && + \left| \frac{1}{p_1}\sum\limits_{i \in \mathH_1 \setminus \widetilde \mathS_\epsilon}I(M_i\leq t) - \frac{1}{p_1} \sum_{i\in \mathH_1 \setminus \widetilde \mathS_\epsilon} I(\widehat M_i \leq t) \right| \nonumber \\
    &\leq & \frac{2|\widetilde \mathS_\epsilon|}{p_1} + \left| \frac{1}{p_1}\sum\limits_{i \in \mathH_1 \setminus \widetilde \mathS_\epsilon}I(M_i\leq t) - \frac{1}{p_1} \sum_{i\in \mathH_1 \setminus \widetilde \mathS_\epsilon} I(\widehat M_i \leq t) \right| \nonumber  \\
    &=& o(1) + \mathJ_{4,1}(t), \quad a.s. \label{eq-median fdr theorem-22} 
\end{eqnarray}
where $o(1)$ is uniform in $t\in [0,1]$, 
\begin{eqnarray}
    \mathJ_{4,1}(t) = \left| \frac{1}{p_1}\sum\limits_{i \in \mathH_1 \setminus \widetilde \mathS_\epsilon}I(M_i\leq t) - \frac{1}{p_1} \sum_{i\in \mathH_1 \setminus \widetilde \mathS_\epsilon} I(\widehat M_i \leq t) \right|.  \label{eq-median fdr theorem-23} 
\end{eqnarray}
For $\mathJ_{4,1}(t)$, based on the definition of $\widetilde \mathS_\epsilon$, for every $i \in \mathH_1 \setminus \widetilde \mathS_\epsilon$,  
\begin{eqnarray*}
\left| \widehat \FF_{0,k}(Z_{i,k})  - \FF_{0,k}(Z_{i,k})\right| \leq \epsilon, \quad \text{for every }k=1,\ldots, K, \label{eq-median fdr theorem-24} 
\end{eqnarray*}
and thus based on Lemma \ref{simple property for median of two sequence}, we have
\begin{eqnarray*}
    \left|\widehat M_i - M_i \right| \leq \epsilon, \label{eq-median fdr theorem-25} 
\end{eqnarray*}
which further implies 
\begin{eqnarray}
    \frac{1}{p_1}\sum\limits_{i \in \mathcal{H}_1\setminus \widetilde \mathS_\epsilon}I(M_i\leq t-\epsilon)\leq\frac{1}{p_1}\sum\limits_{i \in \mathcal{H}_1\setminus \widetilde\mathS_\epsilon}I(\hat{M}_{i}\leq t)\leq \frac{1}{p_1}\sum\limits_{i \in \mathcal{H}_1\setminus \widetilde \mathS_\epsilon}I(M_i\leq t+\epsilon). \label{eq-median fdr theorem-26} 
\end{eqnarray}
Clearly
\begin{eqnarray}
    \frac{1}{p_1}\sum\limits_{i \in \mathcal{H}_1\setminus \widetilde \mathS_\epsilon}I(M_i\leq t-\epsilon)\leq\frac{1}{p_1}\sum\limits_{i \in \mathcal{H}_1\setminus \widetilde \mathS_\epsilon}I({M}_{i}\leq t)\leq \frac{1}{p_1}\sum\limits_{i \in \mathcal{H}_1\setminus \widetilde \mathS_\epsilon}I(M_i\leq t+\epsilon). \label{eq-median fdr theorem-27} 
\end{eqnarray}
Combining \eqref{eq-median fdr theorem-23}, \eqref{eq-median fdr theorem-26} and \eqref{eq-median fdr theorem-27}, we have 
\begin{eqnarray}
    \mathJ_{4,1}(t) &\leq& \left|\frac{1}{p_1}\sum\limits_{i \in \mathcal{H}_1\setminus \widetilde \mathS_\epsilon}I(M_i\leq t+\epsilon)- \frac{1}{p_1}\sum\limits_{i \in \mathcal{H}_1\setminus \widetilde \mathS_\epsilon}I(M_i\leq t-\epsilon) \right| \nonumber \\
    & = & \frac{1}{p_1} \sum_{i\in \mathcal{H}_1\setminus \widetilde \mathS_\epsilon} I\left\{ M_i \in (t-\epsilon, t+\epsilon] \right\}  = o(1) + G(t+\epsilon) - G(t-\epsilon), \quad a.s. \label{eq-median fdr theorem-28}  
\end{eqnarray}
based on \eqref{eq-median fdr theorem-21} and Lemma \ref{key lemma in median approcah 1}, where $G(\cdot)$ is defined by \eqref{eq-def-G-t}. Based on Condition \ref{Condition-5},  \eqref{eq-median fdr theorem-28}, and that $\epsilon$ is arbitrary, we have 
\begin{eqnarray*}
    \sup_{t\in [0,1]}\mathJ_{4,1}(t) = o(1), \quad a.s. \label{eq-median fdr theorem-30}
\end{eqnarray*}
which together with \eqref{eq-median fdr theorem-22} implies that $\mathJ_4(t)$ defined by \eqref{eq-median fdr theorem-17} satisfies 
\begin{eqnarray}
    \sup_{t\in [0,1]} |\mathJ_4(t)| = \sup_{t\in [0,1]} \left|\frac{1}{p_1} \sum_{i\in \mathH_1}I(\widehat M_i\leq t) - \frac{1}{p_1} \sum_{i\in \mathH_1} I(M_i \leq t)\right| = o(1), \quad a.s.  \label{eq-median fdr theorem-31}
\end{eqnarray}
Combining \eqref{eq-median fdr theorem-16}, \eqref{eq-median fdr theorem-31}, Condition \ref{Condition-5}, and Lemma \ref{key lemma in median approcah 1} leads to 
\begin{eqnarray}
    &&\sup_{t\in[0,1]} \left|\frac{1}{p} \sum\limits_{i=1}^pI(\widehat M_i\leq t) - \left\{ \lambda_{0,0} M(t) + \lambda_{1,0} G^*(t) \right\} \right|\nonumber \\
    &\leq & \lambda_{0,0} \sup_{t\in [0,1]}\left| \frac{1}{p_0} \sum_{i\in \mathH_0} I(\widehat M_i \leq t) - M(t) \right| + \lambda_{1,0} \sup_{t \in [0,1]} \left| \frac{1}{p_1} \sum_{i\in \mathH_1} I(\widehat M_i \leq t) - G^*(t)\right| \nonumber \\
    &\leq & \lambda_{0,0} \sup_{t\in [0,1]}\left| \frac{1}{p_0} \sum_{i\in \mathH_0} I(\widehat M_i \leq t) - M(t) \right| \nonumber \\ 
    &&+ \lambda_{1,0} \sup_{t \in [0,1]} \left| \frac{1}{p_1} \sum_{i\in \mathH_1} I(\widehat M_i \leq t) - \frac{1}{p_1} \sum_{i\in \mathH_1} I(M_i \leq t)\right| \nonumber \\
    && +  \lambda_{1,0} \sup_{t \in [0,1]} \left|  \frac{1}{p_1} \sum_{i\in \mathH_1} I(M_i \leq t) - G(t)\right| + \lambda_{1,0} \sup_{t \in [0,1]} \left|  G(t) - G^*(t)\right| \nonumber \\
    &=& o(1), \quad a.s. \label{eq-median fdr theorem-32}
\end{eqnarray}
Thus, referring to the definition of $\widehat{\FDP}(t)$ in \eqref{eq-def-hat-FDP-t-median}, based on \eqref{eq-median fdr theorem-32}, Theorem \ref{theorem-2}, and Condition \ref{Condition-5},  we have 
\begin{eqnarray*}
    \widehat{\FDP}(t_\alpha^\infty) &=& \frac{\lambda_0^* M(t_\alpha^\infty)}{\lambda_0^* M(t_\alpha^\infty) + (1-\lambda_0^*)G^*(t_\alpha^\infty)} + o(1), \quad a.s. \\
    &<& \alpha -c + o(1), \quad a.s., \label{eq-median fdr theorem-33}
\end{eqnarray*}
for a universal constant $c>0$. 
Thus, as $p \to \infty $, we have 
\begin{eqnarray*}
    \widehat{\FDP}(t_\alpha^\infty)\le \alpha, \quad a.s. \label{eq-median fdr theorem-34}
\end{eqnarray*}
which together with the definition of $t_\alpha$ in \eqref{eq-t-alpha-median} concludes 
\begin{eqnarray}
    t_\alpha \geq t_\alpha^\infty > 0, \quad a.s.  \label{eq-median fdr theorem-35}
\end{eqnarray}
Combining \eqref{eq-median fdr theorem-32} and \eqref{eq-median fdr theorem-35}, we have, as $p \to \infty$, 
\begin{eqnarray}
    \frac{1}{p} \max\left\{ \sum\limits_{i=1}^pI(\widehat M_i\leq t_\alpha), 1 \right\}  \geq   \frac{1}{2}\lambda_{0,0} M(t_\alpha) + \frac{1}{2}\lambda_{1,0} G^*(t_\alpha)   \geq \frac{1}{2}\lambda_{0,0} M(t_\alpha^\infty) >0, \quad a.s. \label{eq-median fdr theorem-36}
\end{eqnarray}
Combining \eqref{eq-median fdr theorem-1}, \eqref{eq-median fdr theorem-16}, \eqref{eq-median fdr theorem-2}, and \eqref{eq-median fdr theorem-36} leads to 
\begin{eqnarray*}
    \FDP(t_\alpha) \leq  \alpha + \left| \widehat \FDP(t_\alpha) - \FDP(t_\alpha) \right| \leq \alpha + \frac{o(1) -\lambda_{0,0} o(1)}{\frac{1}{2}\lambda_{0,0} M(t_\alpha^\infty )} = \alpha + o(1), \quad a.s.
\end{eqnarray*}
We complete the proof of the theorem.

\section{Details for Section 3.2 of the Main Article} \label{Section-Details Main 3.2}

In this section, we provide more details about the development of likelihood and the algorithm to estimate model parameters in Section 3.2 of the main article. 

We pretend that $Z_1, \ldots, Z_p$ are i.i.d., and
\begin{eqnarray*}
    Z_i \sim \FF_\theta(z) = \sum_{m=0}^1 \lambda_m \prod_{k=1}^K F_{m,k}(z_k). 
\end{eqnarray*}

For each $i,j=1,\ldots, p; \  i\neq j$, let $I_{i,j} = I(Z_i \leq Z_j) = (I_{i,j,1}, \ldots, I_{i,j,K})^T$, where $I_{i,j,k} = I(Z_{i,k} \leq Z_{j,k})$. Then, for any $z = (z_1, \ldots, z_K)^T \in \mathbb{R}^K$, we have:
\begin{eqnarray*}
    P\left(I(Z_i \leq Z_j) = I(z\leq Z_j) \Big| Z_j \right) = \sum_{m=0}^1\lambda_m \prod_{k=1}^K F_{m,k}^{I(z_k \leq Z_{j,k})}(Z_{j,k})\cdot \bar F_{m,k}^{I(z_k > Z_{j,k})}(Z_{j,k}),
\end{eqnarray*}
where $\bar F_{m,k}(\cdot) = 1-F_{m,k}(\cdot)$. Then, by treating $I_{i,j}, i = 1,\ldots, p$ as multivariate responses and $Z_j, j=1,\ldots, p$ as covariates, the conditional log-likelihood is given by 
\begin{eqnarray*}
    \ell_j(\theta) =  \sum_{i=1}^p \log \left[ \sum_{m=0}^1\lambda_m \prod_{k=1}^K F_{m,k}^{I_{i,j,k}}(Z_{j,k})\cdot \bar F_{m,k}^{1-I_{i,j,k}}(Z_{j,k}) \right]. 
\end{eqnarray*}
We then aggregate this conditional log-likelihood across the $Z_j$'s using the idea of composite log-likelihood \citep{Cristiano2011, Kwonsang2023}, and obtain our log-likelihood objective function: 
\begin{eqnarray*}
    \ell(\theta) = \sum_{j=1}^p \sum_{i=1}^p \log \left[ \sum_{m=0}^1\lambda_m \prod_{k=1}^K F_{m,k}^{I_{i,j,k}}(Z_{j,k})\cdot \bar F_{m,k}^{1-I_{i,j,k}}(Z_{j,k}) \right]. 
\end{eqnarray*}
The estimator for $\theta$ is defined as 
\begin{eqnarray}
    \widehat \theta = \{\widehat \lambda_0, \widehat \FF_{0,1}(\cdot), \ldots, \widehat \FF_{1,K}(\cdot)\} = \arg\max_{\theta \in \Theta} \ell(\theta), \label{eq-theta-hat-def}
\end{eqnarray}
where 
\[\Theta = \left\{\theta: \lambda_0 \in (0,1), F_{m,k}(\cdot) \text{ are c.d.f.s }, m =0,1; k=1,\ldots, K\right\}.\]

To solve the optimization in \eqref{eq-theta-hat-def} and ensure that the c.d.f. estimates represent genuine distributions, we incorporate the EM algorithm with the pool-adjacent-violator algorithm (PAVA) and active set methods \citep{Ayer1955, de2009}. For $s = 1,2,\ldots$, let $\theta^{(s)} = \left\{\lambda_m^{(s)}, \FF_{m,k}^{(s)}, m=0,1; k = 1,\ldots, K\right\}$ be the estimate in the $s$th step. Let 
\begin{eqnarray*}
    \lambda_{i,j,m}^{(s)} = \frac{\lambda_m^{(s)}\prod_{k=1}^K \left[ \left\{ \FF_{m,k}^{(s)}(Z_{j,k}) \right\}^{I_{i,j,k}} \left\{ \bar\FF_{m,k}^{(s)}(Z_{j,k}) \right\}^{1-I_{i,j,k}}  \right]}{\sum_{m_1=1}^2 \lambda_{m_1}^{(s)}\prod_{k=1}^K \left[ \left\{ \FF_{m_1,k}^{(s)}(Z_{j,k}) \right\}^{I_{i,j,k}} \left\{ \bar\FF_{m_1,k}^{(s)}(Z_{j,k}) \right\}^{1-I_{i,j,k}}  \right]}. 
\end{eqnarray*}
Then, $\lambda_m^{(s+1)}$ and $\FF^{(s+1)}$, for $m = 0,1$ and $k = 1, \ldots, K$, are updated as the maximizer of 
\begin{equation*}
\sum_{j=1}^p \sum_{i=1}^p \sum_{m=0}^1 \lambda_{i,j,m}^{(s)} \log \left[\lambda_m \prod_{k=1}^p \left\{F_{m,k}^{I_{i,j,k}}(Z_{j,k}) \bar F_{m,k}^{1-I_{i,j,k}}(Z_{j,k}) \right\}\right],
\end{equation*}
subject to $\lambda_0 + \lambda_1 = 1$ and $F_{m,k}$ is a c.d.f. for each $m = 0, 1; k = 1,\ldots, K$. This leads to 
\begin{eqnarray*}
    \lambda_0^{(s+1)} = \frac{1}{p^2}\sum_{j=1}^p\sum_{i=1}^p  \lambda_{i,j,0}^{(s)},
\end{eqnarray*}
$\lambda_1^{(s+1)} = 1 - \lambda_0^{(s+1)}$, and 
\begin{eqnarray}
    \FF_{m,k}^{(s+1)}(\cdot) &=& \arg\max_{F_{m,k}}\sum_{j=1}^p \sum_{i=1}^p \lambda_{i,j,m}^{(s)} \left\{ I_{i,j,k} \log F_{m,k}(Z_{j,k}) + (1-I_{i,j,k}) \log \bar F_{m,k}(Z_{j,k}) \right\} \nonumber \\
    &=& \arg\max_{F_{m,k}} \sum_{j=1}^p w_{m,j}^{(s)} \left[ \xi_{m,k,j}^{(s)} \log F_{m,k}(Z_{j,k}) + \left\{1-\xi_{m,k,j}^{(s)}\right\} \log \bar F_{m,k}(Z_{j,k}) \right], \label{eq-F-est-algorithm}
\end{eqnarray}
subject to $F_{m,k}(\cdot)$ being a c.d.f., where
\begin{eqnarray*}
    \xi_{m,k,j}^{(s)} 
    =\frac{\sum_{i=1}^p \lambda_{i,j,m}^{(s)} I_{i,j,k}}{\sum_{i=1}^p \lambda_{i,j,m}^{(s)}}, \qquad  w_{m,j}^{(s)} = \sum_{i=1}^p \lambda_{i,j,m}^{(s)}. 
\end{eqnarray*}
The solution for \eqref{eq-F-est-algorithm} is equivalent to that of the weighted isotonic regression: 
$$
\FF_{m,k}^{(s+1)}(\cdot) = \arg\min_{F_{m,k}}\sum_{j=1}^p w_{m,j}^{(s)}\left\{ \xi_{m,k,j}^{(s)} - F_{m,k}(t_{j,k}) \right\}^2,
$$
subject to $F_{m,k}(\cdot)$ being a c.d.f., which can be solved by the PAVA and active set methods. 

The next theorem shows that the above algorithm increases $\ell(\theta)$ in every updating step; the proof is similar to that of Theorem 1 in \citet{YuQinLi2026}. 

\begin{theorem}
For every $s\geq 1$ and initial value $\theta^{(1)} \in \Theta$, we have $\ell(\theta^{(s+1)}) \geq \ell(\theta^{(s)})$. 
\end{theorem}
\proof We can write 
\begin{eqnarray}
\ell(\theta) &=& \sum_{j=1}^p \sum_{i=1}^p \log \left[\sum_{m=0}^1 \lambda_m \prod_{k=1}^K \left\{F_{m,k}^{I(Z_{i,k} \leq Z_{j,k})}(Z_{j,k}) \bar F_{m,k}^{I(Z_{i,k}> Z_{j,k})}(Z_{j,k}) \right\}\right] \nonumber \\
&=& \sum_{j=1}^p \sum_{i=1}^p \log \left\{\sum_{m=0}^1 \lambda_m G_m(Z_i, Z_j)\right\}, \nonumber \label{eq-Appendix-A-1}
\end{eqnarray}
where 
\begin{eqnarray}
    G_m(Z_i, Z_j) = \prod_{k=1}^K \left\{F_{m,k}^{I_{i,j,k}}(Z_{j,k}) \bar F_{m,k}^{1-I_{i,j,k}}(Z_{j,k}) \right\}. \label{eq-Appendix-A-1-1}
\end{eqnarray}
Suppose that $\theta^{(1)}$ is updated from $\theta$ based on our algorithm. It suffices to show that $\ell\left(\theta^{(1)}\right) \geq  \ell\left(\theta\right)$. We consider
\begin{eqnarray}
\ell(\theta^{(1)}) -  \ell(\theta) &=& \sum_{j=1}^p \sum_{i=1}^p \log \left\{\frac{\sum_{m=0}^1 \lambda_m^{(1)} G_m^{(1)}(Z_i, Z_j)}{\sum_{m=0}^1 \lambda_m G_m(Z_i, Z_j)}\right\} \nonumber \\
&=& \sum_{j=1}^p \sum_{i=1}^p \log \sum_{m = 0}^1 \lambda_{i,j,m} \frac{\lambda_m^{(1)} G_m^{(1)}(Z_i, Z_j)}{\lambda_m G_m(Z_i, Z_j)} \nonumber \\
&\geq& \sum_{j=1}^p \sum_{i=1}^p \sum_{m = 0}^1 \lambda_{i,j,m} \log \frac{\lambda_m^{(1)} G_m^{(1)}(Z_i, Z_j)}{\lambda_m G_m(Z_i, Z_j)} \nonumber \\
&=& \left\{\ell_1\left(\lambda_0^{(1)}\right) - \ell_1(\lambda_0)\right\}  + \left\{\ell_2\left(G_0^{(1)}, G_1^{(1)}\right) - \ell_2\left( G_0, G_1\right), \label{eq-Appendix-A-2}\right\}
\end{eqnarray}
where
\begin{eqnarray}
\lambda_{i,j,m} &=& \frac{\lambda_m G_m(Z_i, Z_j)}{\sum_{m=0}^1 \lambda_m G_m(Z_i, Z_j)} \label{eq-Appendix-A-3}, \nonumber \\
\ell_1(\lambda_0) &=&  \sum_{j=1}^p \sum_{i=1}^p \sum_{m = 0}^1 \lambda_{i,j,m} \log \lambda_m \label{eq-Appendix-A-4},\\
\ell_2\left( G_0, G_1\right) &=& \sum_{j=1}^p \sum_{i=1}^p \sum_{m = 0}^1 \lambda_{i,j,m} \log G_m(Z_i, Z_j).  \label{eq-Appendix-A-5}
\end{eqnarray}
Applying the method of Lagrange multiplier, it is straightforward to verify that subject to $ \lambda_0 + \lambda_1 = 1$, $\ell_1(\lambda_0)$ given in \eqref{eq-Appendix-A-4} is maximized when 
\begin{eqnarray*}
    \lambda_m^* = \frac{\sum_{j=1}^p\sum_{i=1}^p \lambda_{i,j,m}}{p^2}, \label{eq-Appendix-A-6}
\end{eqnarray*}
which is equal to $\lambda_m^{(1)}$ given in our algorithm, and thus 
\begin{eqnarray}
    \ell_1\left(\lambda_0^{(1)}\right) - \ell_1(\lambda_0) \geq 0. \label{eq-Appendix-A-7}
\end{eqnarray}
Based on \eqref{eq-Appendix-A-1-1} and \eqref{eq-Appendix-A-5}, we have
\begin{eqnarray}
    &&\ell_2\left( G_0, G_1\right) \nonumber \\
    &=& \sum_{j=1}^p \sum_{i=1}^p \sum_{m = 0}^1 \lambda_{i,j,m} \log G_m(Z_i, Z_j) \nonumber \\
    &=& \sum_{j=1}^p \sum_{i=1}^p \sum_{m = 0}^1 \sum_{k=1}^p \lambda_{i,j,m} \Big\{I_{i,j,k}\log F_{m,k}(Z_{j,k}) + (1- I_{i,j,k})\log \bar F_{m,k}(Z_{j,k}) \Big\} \nonumber \\
    &=& \sum_{m=0}^1 \sum_{k=1}^p \ell_{2, m,k}(F_{m,k}), \label{eq-Appendix-A-8}
\end{eqnarray}
where 
\begin{eqnarray*}
    \ell_{2, m,k}(F_{m,k}) = \sum_{j=1}^p\sum_{i=1}^p \lambda_{i,j,m} \Big\{I_{i,j,k}\log F_{m,k}(Z_{j,k}) + (1- I_{i,j,k})\log \bar F_{m,k}(Z_{j,k}) \Big\}.  \label{eq-Appendix-A-9}
\end{eqnarray*}
Based on our algorithm, $\FF_{m,k}^{(1)}$ is the maximizer of $\ell_{2, m,k}(F_{m,k})$, and therefore $\ell_{2, m,k}(\FF_{m,k}^{(1)})\geq \ell_{2, m,k}(F_{m,k})$, which together with \eqref{eq-Appendix-A-8} implies 
\begin{eqnarray}
    \ell_2\left( G_0^{(1)}, G_1^{(1)}\right) - \ell_2\left( G_0, G_1\right) \geq 0. \label{eq-Appendix-A-10}
\end{eqnarray}
Combining \eqref{eq-Appendix-A-2}, \eqref{eq-Appendix-A-7}, and \eqref{eq-Appendix-A-10}, we conclude that 
$
\ell\left(\theta^{(1)}\right) \geq  \ell\left(\theta\right), 
$
completing the proof of this theorem. \epf

\section{Additional Simulation Examples} \label{Simulation-one-sided}


In addition to Studies 1--6 in Section 5 of the main article, 
we investigate additional examples under one-sided hypotheses:
\begin{eqnarray*}
    H_{0,j}: \mu_{0,j} \geq \mu_{1,j} \quad \text{versus} \quad H_{1,j}: \mu_{0,j} < \mu_{1,j}, \quad j=1,\ldots,p.
\end{eqnarray*}

The simulation settings largely follow those in Section 5.1 of the main article, except that the distributions $F(\cdot)$ and $G(\cdot)$ are generated differently. Specifically, for the first group, we generate $U_{i,j} \sim F_j$ for $i=1,\ldots,n_0$ and $j=1,\ldots,p$, allowing $F_j$ to vary across genes. For the second group, we generate $V_{i,j} \sim \mathcal{N}(1,1)$ for $i=1,\ldots,n_1$ and $j=1,\ldots,p$, so that $\mu_{1,j} = 1$ for all $j$. The specifications of $F_j$ are given below.

\begin{itemize}
    \item Study 7: For each $j=1,\ldots,p$, if $j \in \mathcal{H}_0$, generate $\mu_{0,j} \sim \text{Uniform}(1.2, 1.25)$; if $j \in \mathcal{H}_1$, generate $\mu_{0,j} \sim \text{Uniform}(0.75, 0.8)$. Then $F_j \sim  \mathcal{N}(\mu_{0,j},1)$. The covariance matrix $\Sigma_\rho$ follows a compound symmetry structure with $\sigma_{i,j} = \rho$ for $i \neq j$.

    \item Study 8: The setup for $\mu_{0,j}$ and $F_j$ is the same as in Study 7. The covariance matrix $\Sigma_\rho$ follows an $\text{AR}(1)$ structure with $\sigma_{i,j} = \rho^{|i-j|}$.

    \item Study 9: For each $j=1,\ldots,p$, if $j \in \mathcal{H}_0$, generate $\mu_{0,j} \sim 1 + 0.12\,\text{Beta}(9,1)$; if $j \in \mathcal{H}_1$, generate $\mu_{0,j} \sim 0.6 + 0.12\,\text{Beta}(9,1)$. Then $F_j \sim \mathcal{N}(\mu_{0,j},1)$. The covariance matrix $\Sigma_\rho$ follows a compound symmetry structure with $\sigma_{i,j} = \rho$ for $i \neq j$.

    \item Study 10: The setup for $\mu_{0,j}$ and $F_j$ is the same as in Study 9. The covariance matrix $\Sigma_\rho$ follows an $\text{AR}(1)$ structure with $\sigma_{i,j} = \rho^{|i-j|}$.
\end{itemize}

We note that the SDA method proposed by \citet{DuGuoSunZou2023SDA} is developed for two-sided testing. The PFA method of \citet{fan2012estimating} can, in principle, be applied to one-sided tests; however, its available implementation in \texttt{R} package \texttt{pfa} \citep{RpackagePFA} is developed for two-sided hypotheses. Therefore, in this section, we compare our method only with the BH procedure of \citet{BenjaminiHochberg1995}.

The simulated data are processed in the same manner as in Section 5.2 of the main article to obtain the test statistics. The results for Studies 7--10 are presented in Figure~\ref{Figure-Examples-7--10}. From the figure, we have the following observations. For all the examples presented, the BH procedure is overly conservative, with FDP values almost approaching zero. In contrast, our method maintains reasonable FDP control, with all examples showing FDP values very close to the nominal level across all values of \( \rho \). Regarding TDP, our method outperforms BH in every example. For Studies 7, 8, and 9, the TDP of our method is notably high, with values approaching 1. For Study 10, the TDP is around 0.83. In comparison, the TDP of the BH method is much lower, around 0.3 for Studies 7 and 8, and about 0.6 for Studies 9 and 10.

The observed differences between our method and the BH procedure are expected. Under the null hypothesis, the means of the resulting two-sample test statistics are not centered around zero, but rather greater than zero. Consequently, the empirical distributions of the p-values under the null hypothesis are not uniformly distributed but are skewed to the left. This skewness causes the BH procedure to be overly conservative, resulting in a loss of power. In fact, we expect that similar behavior might occur in any multiple testing procedure that relies on the theoretical (asymptotic) distributions of individual test statistics, regardless of how well the method handles dependence structures. The root cause of this problem lies not in the dependence structure, but in the empirical performance of the test statistics under the null hypothesis. In contrast, our method directly estimates the empirical distributions of the null statistics, allowing it to bypass this skewness problem. As a result, our method provides reasonable control of the FDP at the nominal level of 0.2, while maintaining good power, as evidenced by the significantly higher TDP values.

    \begin{figure}[H]
    \centering
    \begin{minipage}{0.49\linewidth}
        \centering
        \includegraphics[width=\linewidth]{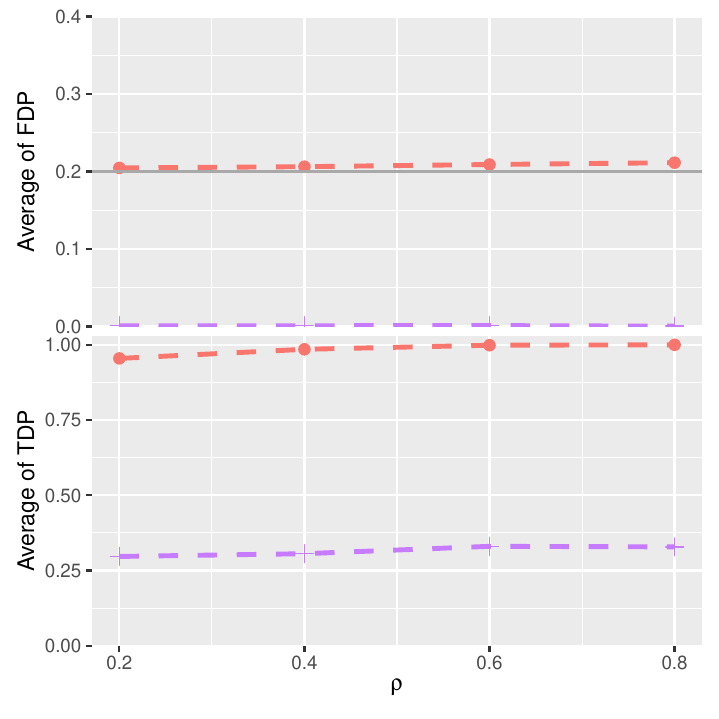}
    \end{minipage}\hfill
    \begin{minipage}{0.49\linewidth}
        \centering
        \includegraphics[width=\linewidth]{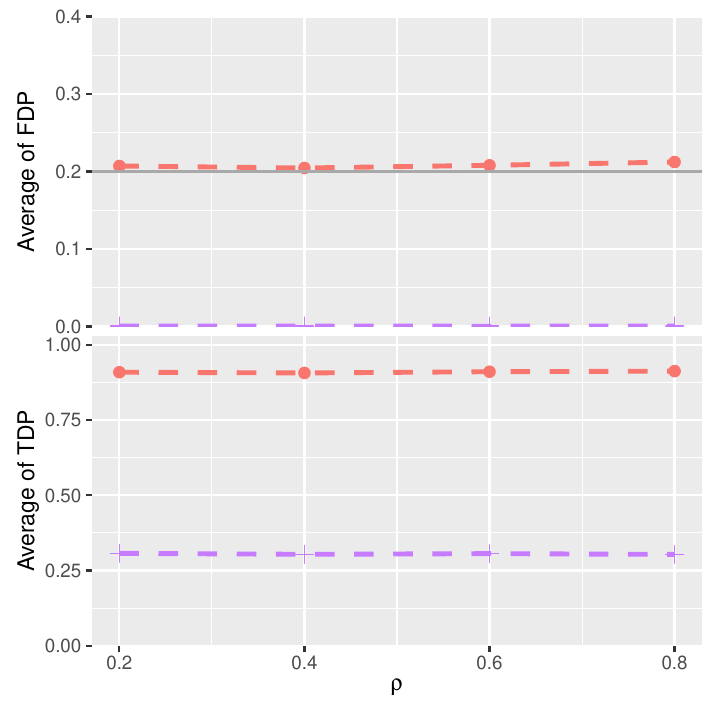}
    \end{minipage}

    \vspace{0.3em}

    \begin{minipage}{0.49\linewidth}
        \centering
        \includegraphics[width=\linewidth]{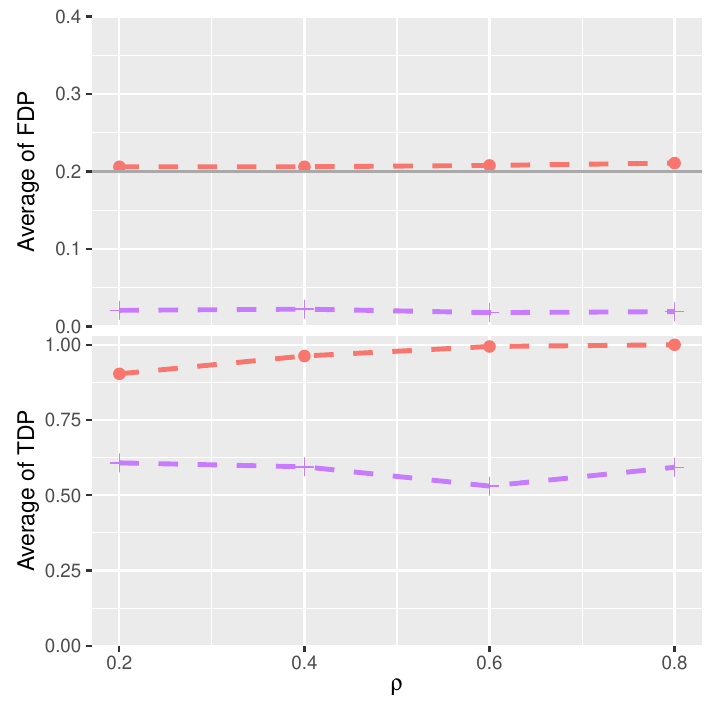}
    \end{minipage}\hfill
    \begin{minipage}{0.49\linewidth}
        \centering
        \includegraphics[width=\linewidth]{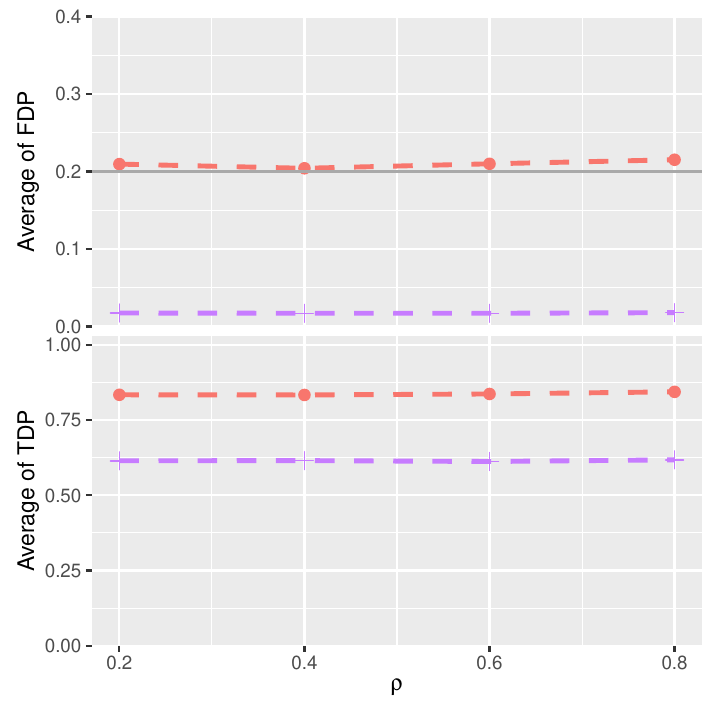}
    \end{minipage}

    \vspace{0.4em}

    \includegraphics[width=0.7\linewidth]{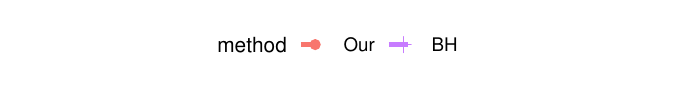}
    \vspace{0.3em}
    \caption{Results of Studies 7--10: top left panel: Study 7; top right panel: Study 8; bottom left panel: Study 9; bottom right panel: Study 10.}
    \label{Figure-Examples-7--10}
\end{figure}

\bibliographystyle{apalike}
\bibliography{ref}